\documentclass[10pt]{article}

\usepackage[margin=0.95in]{geometry}
\usepackage{amsthm,amsfonts,amssymb,amsmath,comment,slashed,mathtools}
\usepackage[T1]{fontenc} 
\usepackage[utf8]{inputenc}

\usepackage[affil-it]{authblk} 

\usepackage[xetex]{graphicx} 
\usepackage{caption}
\usepackage{subcaption}

\usepackage{thm-restate} 

\usepackage{here}
\usepackage{xcolor}  
\usepackage[capitalize]{cleveref} 

\usepackage[
backend=biber, bibencoding=ascii,
giveninits=true,doi=false,isbn=false,url=false,sorting=nyt]{biblatex} 
\renewbibmacro{in:}{} 
\bibliography{ads_interior}
\theoremstyle{plain}
\newtheorem{definition}{Definition}[section] 
\newtheorem{prop}[definition]{Proposition}

\newtheorem{lemma}[definition]{Lemma}
\newtheorem{cor}[definition]{Corollary}
\newtheorem{rmk}[definition]{Remark}
\newtheorem{conjecture}{Conjecture}
\newtheorem{thmrough}{Theorem}
\newtheorem*{notation}{Notation}
\newtheorem{theorem}[definition]{Theorem}
\renewcommand{\d}{\mathrm{d}}
\newcommand{\Hp}{\mathcal{H}}

\newcommand{\Ch}{\mathcal{CH}}
\newcommand{\V}{{\mathcal{V}}}
\newcommand{\red}{{r_{\mathrm{red}}}}
\newcommand{\Rred}{{\mathcal{R}_{\mathrm{red}}}}
\newcommand{\cv}[1]{{\underline{\mathcal{C}}_{#1}}}

\newcommand{\dvol}{\mathrm{dvol}}
\newcommand{\rblue}{r_{\mathrm{blue}}}
\newcommand{\dvpsi}{{\tilde{\nabla}_v \psi}}
\newcommand{\dupsi}{{\tilde{\nabla}_u \psi}}
\newcommand{\vp}{\langle v \rangle^p}
\newcommand{\up}{\langle u \rangle^p}
\makeatletter
\renewcommand{\paragraph}[1]{%
	\par 
	\addvspace{\medskipamount}
	\textit{#1\@addpunct{.}}\enspace\ignorespaces
}
\makeatother
\numberwithin{equation}{section}
\numberwithin{theorem}{section}

\title{Uniform boundedness and continuity at the Cauchy horizon \\for linear waves on Reissner--Nordström--AdS black holes}
\author{Christoph Kehle\thanks{c.kehle@maths.cam.ac.uk}}
\affil{\small  Department of Pure Mathematics and Mathematical Statistics,\\ University~of~Cambridge,~Wilberforce~Road,~Cambridge CB3 0WB, United Kingdom \vskip.1pc \  }

\date{July 22, 2019}
\begin{document}
\maketitle
\thispagestyle{empty}
\begin{abstract}
	Motivated by the Strong Cosmic Censorship Conjecture for asymptotically AdS spacetimes, we initiate the study of massive scalar waves satisfying $\Box_g \psi - \mu \psi =0$ on the interior of Anti-de~Sitter (AdS) black holes. We prescribe initial data on a spacelike hypersurface of a Reissner--Nordström--AdS black hole and impose Dirichlet (reflecting) boundary conditions at infinity. It was known previously that such waves only decay at a sharp logarithmic rate (in contrast to a polynomial rate as in the asymptotically flat regime) in the black hole exterior. In view of this slow decay, the question of uniform boundedness in the black hole interior and continuity at the Cauchy horizon has remained up to now open. We answer this question in the affirmative.
\end{abstract}
\tableofcontents
\thispagestyle{empty}
\newpage
\section{Introduction}
\setcounter{page}{1}
We initiate the study of (massive) linear waves satisfying \begin{align}\Box_g \psi -\mu \psi =0\label{eq:wave0}\end{align} on the \emph{interior} of asymptotically Anti-de~Sitter (AdS) black holes $(\mathcal{M},g)$. In the context of asymptotically AdS spacetimes it is natural to consider (possibly negative) mass parameters $\mu$ satisfying the Breitenlohner--Freedman \cite{breitenlohner} bound $\mu > \frac{3}{4}\Lambda$, where $\Lambda <0$ is the cosmological constant of the underlying spacetime. In particular, this covers the conformally invariant operator with $\mu = \frac{2}{3}\Lambda$. We will consider Reissner--Nordström--AdS (RN--AdS) black holes \cite{rnads} which can be viewed as the simplest model in the context of the question of stability of the Cauchy horizon. These spacetimes are spherically symmetric solutions of the Einstein equations 
\begin{align}\label{eq:Einstein}\tag{EE}
	\mathrm{Ric}_{\mu\nu}  - \frac{1}{2} \mathrm R g_{\mu\nu} + \Lambda g_{\mu\nu} = 8 \pi T_{\mu\nu}
\end{align} coupled to the Maxwell equations via the energy momentum tensor $T_{\mu\nu}$. Our main result \cref{thm:thmrough} (see \cref{thm:main} in \cref{sec:mainthm} for its precise formulation) is the statement of \textbf{uniform boundedness} in the black hole interior and \textbf{continuity} at the Cauchy horizon of solutions to \eqref{eq:wave0} arising from initial data on a spacelike hypersurface on RN--AdS. We moreover assume Dirichlet (reflecting) boundary conditions at infinity.
Our result is surprising because in contrast to black hole backgrounds with non-negative cosmological constants ($\Lambda \geq 0$), the decay of $\psi$ in the exterior region for asymptotically AdS black holes ($\Lambda <0$) is only \emph{logarithmic} as shown by Holzegel--Smulevici \cite{decaykg} (cf.\  {polynomial} \cite{price,partiii,aretakis} ($\Lambda =0$) and {exponential} \cite{decaydesitter,dyatlov} ($\Lambda >0$)). Indeed, the logarithmic decay is too slow to adapt the mechanism exploited in previous studies of black hole interiors \cite{dafermos_cauchy_horizon,anneboundedness,dafermos2017interior}. The proof of our main theorem will now follow a new approach, combining physical space estimates with Fourier based estimates exploited in the scattering theory developed in \cite{kehle2018scattering}.

In the rest of the introduction we will give some background on the problem and formulate our main result \cref{thm:thmrough}.
\paragraph{\textbf{The Cauchy horizon and the Strong Cosmic Censorship Conjecture}}
The main motivation for studying linear waves on black hole interiors is to shed light on one of the most fundamental puzzles in general relativity: The Kerr(--de~Sitter or --Anti-de~Sitter) and Reissner--Nordström (--de~Sitter or --Anti-de~Sitter) black holes share the property that in addition to the \emph{event horizon} $\mathcal H$, they hide another horizon, the so-called \emph{Cauchy horizon} $\mathcal{CH}$, in their interiors.
\footnote{More precisely, this holds true for subextremal and non-trivially rotating(charged) Kerr(Reissner--Nordström) black holes which we will assume for the rest of the paper, unless explicitly stated otherwise.}
 This Cauchy horizon defines the boundary beyond which initial data on a spacelike hypersurface (together with boundary conditions at infinity in the asymptotically AdS case)  no longer uniquely determine the spacetime as a solution of \eqref{eq:Einstein}.
 In particular, these spacetimes admit infinitely many smooth extensions beyond their Cauchy horizons solving \eqref{eq:Einstein}.
 This severe violation of determinism is conjectured to be an artifact of the high degree of symmetry in those explicit spacetimes and generically, due to blue-shift instabilities, it is expected that a singularity ought to form at or before the Cauchy horizon. This is known as the \emph{Strong Cosmic Censorship Conjecture} (SCC) \cite{penrose1974gravitational,christo}. A full resolution of the SCC conjecture would also include a precise description of the breakdown of regularity at or before the Cauchy horizon. 
 
 We first present the $C^0$ formulation of SCC (see \cite{christo,dafermos2017interior}), which can be seen as the strongest inextendibility statement in this context.  
 \begin{conjecture}[$C^0$ formulation of strong cosmic censorship]\label{conject0}
For generic compact or asymptotically flat (asymptotically Anti-de~Sitter) vacuum initial data, the maximal Cauchy development of \eqref{eq:Einstein} is inextendible as a Lorentzian manifold with $C^0$ (continuous) metric.
 \end{conjecture} 
Surprisingly, the $C^0$ formulation (\cref{conject0}) was recently proved to be false for both cases $\Lambda =0$ and $\Lambda>0$ (see discussion later, \cite{dafermos2017interior}). However, the following weaker, yet well-motivated, formulation introduced by Christodoulou in \cite{christo} is still expected to hold true (at least) in the asymptotically flat case ($\Lambda=0$).
\begin{conjecture}[Christodoulou's re-formulation of strong cosmic censorship]\label{conj:2}
	For generic asymptotically
	flat vacuum initial data, the maximal Cauchy development of \eqref{eq:Einstein} is inextendible as a Lorentzian manifold with $C^0$ (continuous) metric and locally square integrable Christoffel symbols.
\end{conjecture}
 In order to gain insight about SCC, the most naive approach (often referred to as ``poor man's linearization'') is to study solutions of \eqref{eq:wave0} with $\mu=0$ on a fixed explicit black hole spacetime (e.g.\ Kerr or Reissner--Nordström). This can be considered as the most naive toy model for \eqref{eq:Einstein} with initial data close to Kerr or Reissner--Nordström data, for which many features of \eqref{eq:Einstein} including the non-linear terms and the tensorial structure are neglected; see the pioneering works for asymptotically flat ($\Lambda =0$) black holes \cite{Simpson1973,mcnamara_behavior,mcnamara_instab,chandra}. Under the identification $\psi \sim g$ and $\partial \psi \sim \Gamma$, where $\psi$ is a solution to \eqref{eq:wave0}, \cref{conject0} corresponds to a failure of $\psi$ to be continuous ($C^0$) at the Cauchy horizon. Similarly, \cref{conj:2} corresponds to a failure of $\psi$ to lie in $H^1_\mathrm{loc}$ at the Cauchy horizon.
 
 \paragraph{\textbf{The state of the art for $\Lambda =0$ and $\Lambda >0$}} The definitive disproof \cite{dafermos2017interior} of \cref{conject0} was preceded by corresponding results on the level of \eqref{eq:wave0}.
 
\paragraph{Linear level for $\Lambda =0$}In the asymptotically flat case ($\Lambda =0$) it was shown in \cite{anneboundedness,franzenkerr} (see also \cite{hintz_dS}) that solutions of \eqref{eq:wave0} with $\mu=0$ arising from data on a spacelike hypersurface remain continuous and uniformly bounded (no $C^0$ blow-up) at the Cauchy horizon of general subextremal Kerr or Reissner--Nordström black hole interiors. (For the extremal case see \cite{gajic1,gajic2}.) The key method for the proof is to use the polynomial decay on the event horizon proved in \cite{partiii} (with rate $|\psi| \lesssim v^{-p}$ and $p>1$) and propagate it into the interior. The boundedness and continuity of $\psi$ at the Cauchy horizon was then concluded from red-shift estimates, energy estimates associated to the novel vector field \begin{align}
\label{eq:x} S = |u|^p \partial_u + |v|^p \partial_v\end{align}
and commuting with angular momentum operators followed by Sobolev embeddings. Here $u,v$ are Eddington--Finkelstein-type null coordinates in the interior.

Besides the above $C^0$ boundedness, it was proved that the (non-degenerate) local energy at the Cauchy horizon blows up for a generic set of solutions $\psi$ in Reissner--Nordström \cite{Lukreissner} and Kerr \cite{timetranslation} black holes. (Note that this blow-up is compatible with the finiteness of the flux associated to \eqref{eq:x} because $\partial_v$ and $\partial_u$ degenerate at the Cauchy horizons $\Ch_A$ and $\Ch_B$, respectively.) A similar blow-up behavior was obtained for Kerr in \cite{luk2016kerr} assuming lower bounds on the energy decay rate of a solution along the event horizon. These results support \cref{conj:2} at least on the level of \eqref{eq:wave0}.
  
Another type of result that has been shown in \cite{kehle2018scattering} is a finite energy scattering theory for solutions of \eqref{eq:wave0} (with $\mu=0$) from the event horizon $\mathcal{H}_A^+ \cup \mathcal{H}_B^+$ to the Cauchy horizon $\Ch_A \cup \Ch_B$ in the interior of Reissner--Nordström black holes. In this scattering theory a linear isomorphism between the degenerate energy spaces (associated to the Killing field $T = \partial_ v -  \partial_u$) corresponding to the event and Cauchy horizon was established. The question reduced to obtaining uniform control over  transmission and reflection coefficients $\mathfrak T(\omega,\ell)$ and $\mathfrak R(\omega,\ell)$ corresponding to fixed frequency solutions. Intuitively, for a purely incoming wave at the event horizon $\mathcal{H}_A^+$, the transmission and reflection coefficients correspond to the amount of $T$-energy scattered to $\Ch_B$ and $\Ch_A$, respectively. Indeed, the theory also carries over to $\Lambda \neq 0$ and $\mu \neq 0$ \emph{except} for the $\omega=0$ frequency.  This will turn out to be important for the present paper.

\paragraph{Linear level for $\Lambda >0$} For Kerr(and Reissner--Nordström)--de Sitter ($\Lambda >0$) it was shown in \cite{hintz_vasy} that solutions of \eqref{eq:wave0} (with $\mu=0$) also remain bounded up to and including the Cauchy horizon. 
 Note that in both cases, $\Lambda =0$ and $\Lambda >0$, the proofs rely crucially on quantitative decay along the event horizon (polynomial for $\Lambda =0$ and exponential for $\Lambda >0$). 

 On the other hand the exponential convergence on the event horizon of a Kerr--de~Sitter black hole is in direct competition with the exponential blue-shift instability and the question of local energy blow-up at the Cauchy horizon for \eqref{eq:wave0} is more subtle, see the conjecture in \cite{nospacelike} and the more recent \cite{dafermos2018rough,dias2018strong,dias2018strong2}.
 
\paragraph{Nonlinear level for $\Lambda =0$ and $\Lambda >0$}Now we turn to the full nonlinear problem for \eqref{eq:Einstein}. As mentioned before, for the Einstein vacuum equations Dafermos--Luk showed that the Kerr Cauchy horizon is $C^0$ stable \cite{dafermos2017interior}, i.e.\ the spacetime is extendible as a $C^0$ Lorentzian manifold. Note that this definitively falsifies \cref{conject0} for $\Lambda=0$ (subject only to the completion of a proof of the nonlinear stability of the Kerr exterior).
 In principle, their proof of $C^0$ extendibility also applies to the interior of Kerr--de~Sitter black holes, where the exterior has been proved to be stable for slowly rotating Kerr--de~Sitter black holes \cite{hintz2016global}, thus falsifying \cref{conject0} for $\Lambda >0$.
 
Nonlinear inextendibility results at the Cauchy horizon have been proved only in spherical symmetry: Coupling the Einstein equation  \eqref{eq:Einstein} to a Maxwell--Scalar field system, it is proved in \cite{dafermos_cauchy_horizon} that the Cauchy horizon is $C^0$ stable, yet $C^2$ unstable \cite{luk2017strong,luk2017strong2,dafermos_cauchy_horizon} for a generic set of spherically symmetric initial data. See also the pioneering work in \cite{internal90,ori91}. This shows the $C^2$ formulation of SCC (but not yet \cref{conj:2}) in spherical symmetry. See \cite{costa2,costa3} for work in the $\Lambda >0$ case. The question of any type of nonlinear instability of the Cauchy horizon without symmetry assumptions and the validity of \cref{conj:2} (even restricted to  a neighborhood of Kerr) have yet to be understood.

\paragraph{\textbf{Linear waves and SCC for asymptotically AdS black holes}} The situation is changed radically if one considers asymptotically Anti-de~Sitter ($\Lambda <0$) spacetimes.
\begin{figure}
	\centering
	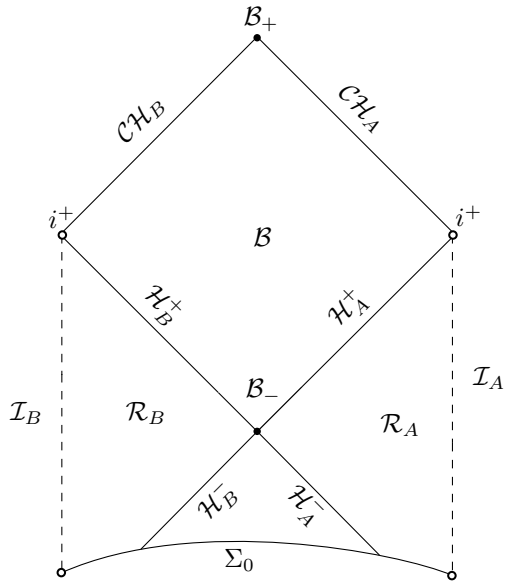
	\caption{Penrose diagram of the maximal Cauchy development of Reissner--Nordström--AdS or Kerr--AdS data on a spacelike surface $\Sigma_0$ with Dirichlet (reflecting) boundary conditions prescribed on null infinity $\mathcal{I}= \mathcal{I}_A \cup \mathcal{I}_B$.}
	\label{fig:adsintro}
\end{figure}
 Due to the timelike nature of null infinity $\mathcal I = \mathcal{I}_A \cup \mathcal{I}_B $, see for example \cref{fig:adsintro}, these spacetimes are not globally hyperbolic. For well-posedness of \eqref{eq:Einstein} and \eqref{eq:wave0} it is required to impose also boundary conditions at infinity.
 The most natural conditions are Dirichlet (reflecting) boundary conditions, see \cite{FRIEDRICH1995125}. Before we address the question of stability of the Cauchy horizon, it is essential to understand the behavior in the exterior region of Kerr--AdS or Reissner--Nordström--AdS.
  \paragraph{Logarithmic decay for linear waves on the exterior of Kerr--AdS and Reissner--Nordström--AdS} For the massive linear wave equation 
  \eqref{eq:wave0} on Kerr--AdS and Reissner--Nordström--AdS, Holzegel--Smulevici showed in \cite{decaykg} stability in the exterior region. Indeed, they proved that solutions decay at least at logarithmic rate towards $i^+$ (cf.\ polynomial ($\Lambda =0$) and exponential ($\Lambda >0$)) assuming the Hawking--Reall \cite{hawking_reall} bound\footnote{Note that otherwise exponentially growing mode solutions can be constructed as shown in \cite{dold}.} $r_+ > |a|l$ and the Breitenlohner--Freedman \cite{breitenlohner} bound $\mu > \frac{3}{4}\Lambda$. Moreover, they showed that solutions of \eqref{eq:wave0} with fixed angular momentum actually decay exponentially on the exterior of Reissner--Nordström--AdS. (This is in contrast to the asymptotically flat case, in which fixed angular momentum solutions of \eqref{eq:wave0} decay {polynomially} on the exterior of Reissner--Nordström.)
  However, their main insight was that a suitable infinite sum of such rapidly decaying fixed angular momentum solutions,  possessing finite energy in some weighted norm, indeed achieves the logarithmic decay rate  \cite{lower_bound}. This is due to the presence of stable trapping. Note that this sharpness can also be concluded from later work showing the existence of quasinormal modes converging to the real axis at an exponential rate as the real part of the frequency and angular momentum tend to infinity \cite{quasi_warnick,quasinormal_gannot}. (For some asymptotically flat five dimensional black holes a similar inverse logarithmic lower bound was shown in \cite{gabriele}.)
   
 \paragraph{Strong Cosmic Censorship for AdS black holes} 
 With the logarithmic decay on the exterior in hand, we turn to the question of the stability of the Cauchy horizon. Indeed, the logarithmic decay rate on the exterior is too slow to follow the methods involving the red-shift vector field and the vector field $S$ as in \eqref{eq:x} (see discussion before) to prove uniform boundedness and $C^0$ (continuous) extendibility at the Cauchy horizon of solutions to \eqref{eq:wave0}. More specifically, after propagating the logarithmic decay through the red-shift region, the energy flux associated to $S$ is infinite on a $\{ r= const.\}$ hypersurface in the black hole interior due to the slow logarithmic decay towards $i^+$. Thus, the question of whether to expect the validity of \cref{conject0} for asymptotically AdS black holes appears to be completely open. (See also the paragraph in the end of the introduction discussion a possible nonlinear instability in the exterior.)
 
 The present paper is an attempt to shed some first light on SCC in the asymptotically AdS case: We will show (\cref{thm:thmrough}) that, \emph{despite the slow decay on the exterior}, boundedness in the interior and continuous extendibility to the Cauchy horizon still holds for solutions of \eqref{eq:wave0} on  Reissner--Nordström--AdS black holes. The additional phenomenon which we exploit to prove boundedness is that the trapped frequencies responsible for slow decay have high energy with respect to the $T$ vector field and can be bounded using the scattering theory developed in \cite{kehle2018scattering}.
Thus, for Reissner--Nordström--AdS, the analog of \cref{conject0} is false on the linear level, just as in the $\Lambda \geq 0$ cases. See however our remarks on Kerr--AdS later in the introduction. 

\paragraph{\textbf{The massive linear wave equation on Reissner--Nordström--AdS}} As mentioned above, we will consider the massive linear wave equation \begin{align}\label{eq:wave} \Box_{g_{\mathrm{RNAdS}}} \psi + \frac{\alpha}{l^2} \psi =0\end{align} 
for AdS radius $l^2 := - \frac{3}{\Lambda}$ on a fixed subextremal Reissner--Nordström--AdS black hole with mass parameter $M>0$ and charge parameter $0<|Q|<M$. Moreover, we assume the so-called Breitenlohner--Freedman
bound \cite{breitenlohner} for the Klein--Gordon mass parameter $\alpha<\frac 94$, which includes the conformally invariant case $\alpha =2$. This bound is required to obtain well-posedness \cite{wellposedgustav,warnick_massive_wave,vasy_wave} of \eqref{eq:wave}.
 
Recall from the discussion above that solutions with fixed angular momentum $\ell$ actually decay exponentially in the exterior region. For such solutions with fixed $\ell$, uniform boundedness with upper bound $C=C_\ell$ in the interior and continuity at the Cauchy horizon can be shown using the methods involving the vector field $S$ as in \eqref{eq:x}. Note however that this does not imply that a general solution remains bounded in the interior as the constant $C_\ell$ is not summable: $ \sum_{\ell=0}^{L} C_\ell \sim e^{ L} \to + \infty$ as $L\to\infty$. Note in particular that, as a result of this, one cannot study the new non-trivial aspect of this problem restricted to spherical symmetry. (Nevertheless, see \cite{ori2} for a discussion of the Ori model for RN--AdS black holes.)  
\paragraph{\textbf{Main theorem: Uniform boundedness and continuity at the Cauchy horizon}}We now state a rough version of our main result. See \cref{thm:main} for the precise statement.
\begin{thmrough}[Rough version of \cref{thm:main}]\label{thm:thmrough}
Let $\psi$ be a solution to \eqref{eq:wave} arising from smooth and compactly supported initial data $(\psi_0,\psi_1)$ posed on a spacelike hypersurface $\Sigma_0$ as depicted in \cref{fig:adsintro}. Then, $\psi$ remains \textbf{uniformly bounded} in the black hole interior
\begin{align*}
|\psi|\leq C,
\end{align*}
where $C$ is constant depending on the parameters $M,Q,l,\alpha$, the choice of $\Sigma_0$ and on some higher order Sobolev norm of the initial data $(\psi_0,\psi_1)$. Moreover, $\psi$ can be extended \textbf{continuously} across the Cauchy horizon. 
\end{thmrough} 
As we have explained above, the main difficulty compared to the asymptotically flat case, where the analysis was carried out entirely in physical space and requires inverse polynomial decay in the exterior \cite{anneboundedness}, is the slow decay of $\psi$ along the event horizon. Our strategy is to decompose the solution $\psi$ in a low and high frequency part $\psi = \psi_\flat + \psi_\sharp$ with respect to the Killing field $T =\frac{\partial}{\partial t}$ and treat each term separately.

For the low frequency part $\psi_\flat$, we will show a superpolynomial decay rate in the exterior, see already \cref{prop:psiflat}. For this part we also use integrated energy decay estimates for bounded angular momenta $\ell$ established in \cite{decaykg}. This superpolynomial decay in the exterior is sufficient so as to follow the method of \cite{anneboundedness} with vector fields of the form \eqref{eq:x} to show boundedness and continuity at the Cauchy horizon, up to the additional difficulty caused by the fact that we allow a possibly negative Klein--Gordon mass parameter. The violation of the dominant energy condition due to the presence of a negative mass term can be overcome with twisted derivatives \cite{warnick_massive_wave,warnick_boundedness_and_growth}, which provide a useful framework to replace Hardy inequalities for the lower order terms in this context.  
 
For the high frequency part $\psi_\sharp$, which is exposed to stable trapping and does in general only decay at a sharp logarithmic rate in the exterior, the key ingredient is the scattering theory developed in \cite{kehle2018scattering} (see discussion above). 
More specifically, the uniform bounds for the transmission and reflections coefficients $\mathfrak T$ and $\mathfrak R$ for $|\omega| \geq \omega_0$ proved in \cite{kehle2018scattering} turn out to be useful for the high frequency part $\psi_\sharp$. These bounds allow us to control $|\psi_\sharp|$ at the Cauchy horizon by the $T$-energy norm on the event horizon commuted with angular derivatives. The $T$-energy flux on the event horizon is in turn bounded from initial data by a simple application of the $T$-energy identity in the exterior.
\emph{In particular, no quantitative decay along the event horizon is used for the high frequency part $\psi_\sharp$.} This is what allows us to overcome the problem of slow logarithmic decay.
\paragraph{\textbf{Outlook on Kerr--AdS}}
  We strongly believe that our arguments also apply to axially symmetric solutions $\psi$ of \eqref{eq:wave} on a Kerr--AdS black hole. For general non-axisymmetric solutions, however, the question of uniform boundedness and continuity at the Cauchy horizon is less clear. Indeed, specific high frequency solutions which decay at a logarithmic decay rate can be considered as ``low frequency'' solutions when frequency is measured with respect to the Killing generator of the Cauchy horizon. In fact, it might well be the case that for solutions of \eqref{eq:wave} on Kerr--AdS there is $C^0$ blow-up at the Cauchy horizon, supporting the validity of \cref{conject0} after all in this context!
 \paragraph{\textbf{Instability of asymptotically AdS spacetimes?}} 
Turning to the fully nonlinear dynamics, there is another scenario which could happen. Recall that Minkowski space ($\Lambda =0$) and de~Sitter space ($\Lambda >0$) have been proved to be nonlinearly stable \cite{Friedrich1986,christodoulou2014global}.  Anti-de~Sitter space ($\Lambda <0$), however, is expected to be nonlinearly unstable with Dirichlet conditions imposed at infinity. This was recently proved in \cite{moschidis2017einstein,moschidis2017proof,moschidis2,moschidis3} for appropriate matter models. See also the original conjecture in \cite{dafermosholzegelconj} and the numerical results in \cite{bizon}. Similarly, for Kerr--AdS (or Reissner--Nordström--AdS), the slow logarithmic decay on the linear level proved in \cite{lower_bound} could in fact give rise to nonlinear instabilities in the exterior.\footnote{Note that in contrast, nonlinear stability for spherically symmetric perturbations of Schwarzschild--AdS was shown for Einstein--Klein--Gordon systems \cite{schwarzschild_ads_stable}.} If indeed the exterior of Kerr--AdS was nonlinearly unstable, linear analysis like that in the present paper would be manifestly inadequate and the question of the validity of Strong Cosmic Censorship would be thrown even more open!
Refer to the introduction of \cite{dafermos2017interior} for a more elaborate discussion.
 \paragraph{\textbf{Outline}}
This paper is organized as follows. In \cref{sec:prelims} we set up the spacetime and summarize relevant previous work. In \cref{sec:mainthm} we state and prove our main result \cref{thm:main}. Parts of the proof require a separate analysis which are  treated in \cref{sec:lowfreq} and \cref{sec:highfreq}.
 
\paragraph{\textbf{Acknowledgment}} The author would like to express his gratitude to Mihalis Dafermos and Yakov Shlapen\-tokh-Rothman for many valuable discussions and helpful remarks. The author thanks John Anderson, Anne Franzen, Dejan Gajic, Jonathan Luk, Georgios Moschidis, Federico Pasqualotto, Igor Rodnianski and Claude Warnick. The author also thanks two anonymous referees for their helpful comments.
This work was supported by the EPSRC grant EP/L016516/1.
The author thanks Princeton University
for hosting him as a VSRC.
\section{Preliminaries}
\label{sec:prelims}
We start by setting up the Reissner--Nordström--AdS spacetime (see \cite{rnads}) and defining relevant norms and energies. We will also introduce useful coordinate systems.
\subsection{The Reissner--Nordström--AdS black hole}\label{sec:prelim21}
We are ultimately interested in the behavior of solutions to \eqref{eq:wave}  to the future of a spacelike hypersurface $\Sigma_0$ as depicted in \cref{fig:adsintro}. For technical reasons (Fourier space decompositions are non-local operations) we will however construct also parts to the past of $\Sigma_0$. 
In the following will define the spacetime pictured in \cref{fig:adsrn}.  
\begin{figure}[ht]
	\centering
		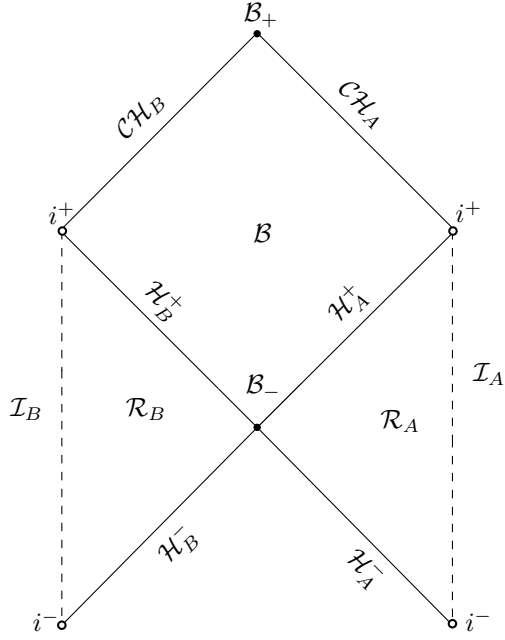
		\caption{Penrose diagram of the constructed spacetime $(\mathcal M_{\mathrm{RNAdS}},g_{\mathrm{RNAdS}})$}
		\label{fig:adsrn}
\end{figure}
\subsubsection{Construction of the spacetime \texorpdfstring{$(\mathcal{M}_\mathrm{RNAdS},g_{\mathrm{RNAdS}})$}{Mg}}
First, for black hole parameters $M>0, Q\neq 0, l^2\neq 0$ define the polynomial  \begin{align}\Delta_{M,Q,l}(r):=r^2 - {2M}r + \frac{r^4}{l^2} + {Q^2}\end{align}
and define the non-degenerate set 
\begin{align}
\mathcal P :=\{ (M,Q,l) \in (0,\infty)\times \mathbb R \times (0,\infty) \colon \Delta_{M,Q,l}(r) \text{ has two postive roots satisfying } 0 < r_- < r_+  \}.
\end{align} 
Note that $\mathcal{P}$ defines black hole parameters in the subextremal range.
From now on, we will consider \emph{\textbf{fixed}} parameters $M,Q,l,\alpha$, where \begin{align}\label{eq:parameters}(M,Q,l) \in \mathcal P \text{ and } \alpha < \frac 94 . \end{align}
Note that $M$ is the mass parameter, $Q$ the charge parameter of the black hole and $l = \sqrt{-\frac{3}{\Lambda}}$ is the Anti-de~Sitter radius.
For this specific choice of parameters we will also write $\Delta (r):= \Delta_{M,Q,l}(r)$ and denote by $0<r_-<r_+$ the positive roots of $\Delta$.

 Now, let the two exterior regions $\mathcal{R}_A$, $\mathcal{R}_B$ and the black hole region $\mathcal{B}$ be smooth four dimensional manifolds diffeomorphic to $\mathbb{R}^2 \times \mathbb{S}^2$.
On $\mathcal{R}_A, \mathcal{R}_B$ and $\mathcal{B}$ we introduce global\footnote{Up to the known degeneracy of spherical coordinates at the poles of the sphere.} coordinate charts: \begin{align}\nonumber &(r_{\mathcal{R}_A}  ,t_{\mathcal{R}_A}, \theta_{\mathcal{R}_A},\varphi_{\mathcal{R}_A} )  \in (r_+,\infty)\times \mathbb{R}\times \mathbb{S}^2,\\& (r_{\mathcal{R}_B}  ,t_{\mathcal{R}_B}, \theta_{\mathcal{R}_B},\varphi_{\mathcal{R}_B} )  \in (r_+,\infty)\times \mathbb{R}\times \mathbb{S}^2, \label{eq:trcoords}\\ & (r_{\mathcal{B}}  ,t_{\mathcal{B}}, \theta_{\mathcal{B}},\varphi_{\mathcal{B}} )  \in (r_-,r_+)\times \mathbb{R}\times \mathbb{S}^2.\nonumber \end{align}
 If it is clear from the context which coordinates are being used, we will omit their subscripts throughout the paper. 
Again, on the manifolds $\mathcal{R}_A,\mathcal{R}_B$ and $\mathcal B$ we define---using the coordinates $(t,r,\theta,\varphi)$ on each of the patches---the  Reissner--Nordström--Anti-de~Sitter metric
\begin{align}
	g := - \frac{\Delta(r)}{r^2}\d t\otimes\d t + \frac{r^2 }{\Delta(r)}\d r \otimes\d r + r^2 (\d\theta \otimes\d\theta + \sin^2\theta \d \varphi\otimes \d\varphi ).
\end{align}
On each of $\mathcal{R}_A,\mathcal{R}_B$ and $\mathcal{B}$, we define time orientations using the vector field $\partial_{t_{\mathcal{R}_A}}$ on $\mathcal{R}_A$, $-\partial_{t_{\mathcal{R}_B}}$ on $\mathcal{R}_B$ and  $-\partial_{r_\mathcal{B}}$ on $\mathcal{B}$.

We will also define the tortoise coordinate $r_\ast$ by
\begin{align}
	\frac{\d r_\ast}{\d r} := \frac{r^2}{\Delta}
\end{align}
in $\mathcal{R}_A$, $\mathcal{R}_B$ and $\mathcal{B}$ independently. This defines $r_\ast$ up to an unimportant constant. Then, in each of the regions $\mathcal{R}_A$, $\mathcal{R}_B$ and $\mathcal{B}$, we define null coordinates by
\begin{align}
	v=r_\ast + t   \text{ and } u = r_\ast-t,
\end{align}
where for example for the $v$ coordinate on $\mathcal{R}_A$, we will use the notation $v_{\mathcal{R}_A}$ and analogously for the other regions.
Note that throughout the paper we will use the notation $^\prime$ for derivatives $\frac{\partial}{\partial r_\ast}$. 
\paragraph{\textbf{Patching the regions} \texorpdfstring{$\mathcal{R}_A,\mathcal{R}_B$ \textbf{and} $\mathcal{B}$ \textbf{together}}{Ra and Rb together}}
Now, we patch the regions $\mathcal{R}_A$, $\mathcal{R}_B$ and $\mathcal{B}$ together. We begin by attaching the future (resp.\ past) event horizon $\mathcal{H}_A^+$ (resp.\ $\mathcal{H}_A^-$) to  $\mathcal{R}_A$ by formally\footnote{This can be made rigorous using ingoing Eddington--Finkelstein coordinates ($r,v,\varphi,\theta$) adapted to the event horizon. Since this is well-known, we avoid introducing yet another coordinate system.} setting
\begin{align}
\mathcal{H}_A^+ := \{  u_{\mathcal{R}_A} = -\infty\} \text{ and } \mathcal{H}_A^- := \{  v_{\mathcal{R}_A} = -\infty\}.
\end{align}
 Similarly, we attach $\mathcal{H}_B^+ := \{ v_{\mathcal{R}_B} = -\infty \}$ and $\mathcal{H}_B^- := \{ u_{\mathcal{R}_B} = -\infty \}$ to $\mathcal{R}_B$.
In the $(u_\mathcal{B},v_\mathcal{B})$ coordinates associated to $\mathcal{B}$ we make the identifications $\mathcal{H}_A^+ = \{ u_{\mathcal{B}} = - \infty \}$ and $\mathcal{H}_B^+ = \{ v_{\mathcal{B}} =-\infty\}$.
Then, we attach the Cauchy horizon $\Ch_A:= \{ v_{\mathcal{B}} = + \infty\}$ and $\Ch_B := \{ u_{\mathcal{B}} = + \infty\}$ to $\mathcal{B}$. 

Finally, we attach the past (resp.\ future) bifurcation sphere $\mathcal{B}_-$ (resp.\ $\mathcal{B}_+$) to $\mathcal{B}$ as 
\begin{align}
	\mathcal{B}_- := \{ u_{\mathcal{B}} = -\infty, v_{\mathcal{B}} = -\infty \} \text{ and } \mathcal{B}_+ := \{ u_{\mathcal{B}} = +\infty, v_{\mathcal{B}}= + \infty\}.
\end{align}
We shall also set $\Ch := \Ch_A \cup \Ch_B\cup \mathcal{B}_+$.
Note that all horizons  $\Hp_A^+, \Hp_A^-, \Hp_B^+, \Hp_B^-, \Ch_A$ , and $\Ch_B$ are diffeomorphic to $\mathbb{R} \times \mathbb{S}^2$ and the past (future) bifurcation sphere $\mathcal{B}_-$ ($\mathcal{B}_+$) is diffeomorphic to $\mathbb{S}^2$. 
Moreover, we identify $\mathcal{B}_-$ with $ \{ u_{\mathcal{R}_A} = -\infty, v_{\mathcal{R}_A} = -\infty \} $ and also with $ \{ u_{\mathcal{R}_B} = -\infty, v_{\mathcal{R}_B} = -\infty \} $. The resulting manifold will be called $\mathcal M_{\mathrm{RNAdS}}$. Note that, $g$ extends to a smooth Lorentzian metric on $\mathcal M_{\mathrm{RNAdS}}$ which we will call $g_{\mathrm{RNAdS}}$ and in particular, $(\mathcal M_{\mathrm{RNAdS}},g_{\mathrm{RNAdS}})$ is a time oriented smooth Lorentzian manifold with corners. We illustrate the constructed spacetime as a Penrose diagram in \cref{fig:adsrn}.
Note that the vector field $\partial_t$ defined on $\mathcal{R}_A$, $\mathcal{R}_B$ and $\mathcal{B}$, respectively, extends to a smooth Killing field on $\mathcal{M}_{\mathrm{RNAdS}}$, which we will from now on call $T$.  Moreover, the standard angular momentum operators $\mathcal{W}_i$ for $i=1,2,3$, the generators of $\mathfrak{so}(3)$ defined as
\begin{align}
	\mathcal{W}_1 =  \sin\varphi \partial_\theta + \cot\theta \cos\varphi \partial_\varphi, \mathcal{W}_2 =  -\cos\varphi \partial_\theta + \cot\theta \sin\varphi \partial_\varphi , \mathcal{W}_3 = - \partial_\varphi 
\end{align} are Killing vector fields. It shall be noted that $\mathcal{W}_i$ for $i=1,2,3$ are spacelike everywhere, whereas $T$ is future-directed timelike on $\mathcal{R}_A$, spacelike on $\mathcal{B}$ and past-directed timelike on $\mathcal{R}_B$. Moreover, $T$ is future-directed null on $\mathcal{H}_A^-,\mathcal{H}_A^+, \mathcal{CH}_B$, past-directed null on  $\mathcal{H}_B^-,\mathcal{H}_B^+, \mathcal{CH}_A$ and vanishes on $\mathcal{B}_-,\mathcal{B}_+$. Finally, note that one can attach conformal timelike boundaries $\mathcal{I}_A$ and $\mathcal{I}_B$ corresponding to $\{r_{\mathcal{R}_A} =+\infty\}$ and $\{ r_{\mathcal{R}_B} = +\infty \}$, respectively.\footnote{Note that $\mathcal{I}_A$ and $\mathcal{I}_B$ are not contained in $\mathcal{M}_{\mathrm{RNAdS}}$.}

\subsubsection{Initial hypersurface \texorpdfstring{$\Sigma_0$}{sigma}}
We will impose initial data on a spacelike hypersurface $\Sigma_0$ to be made precise in the following. Note that we can choose for convenience that the spacelike hypersurface $\Sigma_0$ lies to the future of the past bifurcation sphere $\mathcal{B}_-$. Indeed, by general theory (an energy estimate in a compact region) this can be assumed without loss of generality \cite{lecture_notes}. 
More precisely, let $\Sigma_0$ be a 3 dimensional connected, complete and spherically symmetric spacelike hypersurface extending to the conformal infinity $\mathcal{I} =\mathcal{I}_A \cup \mathcal{I}_B$. Moreover, assume that $\mathcal{B}_- \subset J^-(\Sigma_0)\setminus \Sigma_0$.

\begin{figure}[ht]
	\centering
	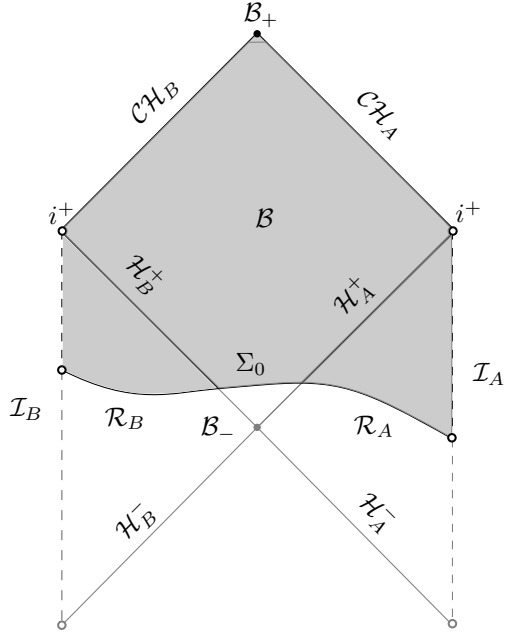
	\caption{The shaded region of interest lies in the future of $\Sigma_0$.}
	\label{fig:adsrn_Sigma}
\end{figure}

 A possible choice of $\Sigma_0$ is denoted in \cref{fig:adsrn_Sigma}. We are ultimately interested in the shaded region to the future of $\Sigma_0$.
 For the rest of the paper, we will consider such a $\Sigma_0$ to be \emph{\textbf{fixed}}. 

\subsection{Conventions}\label{sec:convention}
With $a\lesssim b$ for $a\in \mathbb{R}$ and $b\geq 0$ we mean that there exists a constant $C(M,Q,l,\alpha, \Sigma_0)$ with $a\leq C b$. If $C(M,Q,l,\alpha, \Sigma_0)$ depends on an additional parameter, say $\ell$, we will write $a \lesssim_\ell b$.  We also use $a\sim b$ for some $a,b\geq 0$ if there exist constants $C_1(M,Q,l,\alpha,\Sigma_0), C_2(M,Q,l,\alpha,\Sigma_0) >0$ with $C_1 a \leq b \leq C_2 a$. 
We shall also make use of the standard Landau notation $O$ and $o$ \cite{olver}. To be more precise, let $X$ be a point set (e.g.\ $X=\mathbb R, [a,b], \mathbb C$)  with limit point $c$. As $x\to c$ in $X$, $f(x) = O(g(x))$ means $\frac{|f(x)|}{|g(x)|} \leq C(M,Q,l,\alpha)$ holds in a fixed neighborhood of $c$. We write $O_\ell(g(x))$ if the constant $C$ depends on an additional parameter $\ell$.
For the standard volume form in spherical coordinates $(\varphi,\theta)$ on the sphere $\mathbb{S}^2$ we will use the notation $\d\sigma_{\mathbb{S}^2} := \sin\theta\d\varphi\d\theta$.
Finally, let the Japanese symbol be defined as $\langle x \rangle := \sqrt{ 1+ x^2 }$ for $x \in \mathbb R$. 
\subsection{Norms and Energies}
We are interested in solutions to the massive wave equation~\eqref{eq:wave} associated to the metric $g_{\mathrm{RNAdS}}$ on a subextremal  Reissner--Nordström AdS black hole with black hole parameters $M,Q,l$ as in \eqref{eq:parameters}.
In view of the timelike boundaries $\mathcal{I}_A$ and $\mathcal{I}_B$, we need to specify boundary conditions on $\mathcal{I}_A$ and $\mathcal{I}_B$ in addition to prescribing data on the spacelike hypersurface ${\Sigma}_0$, cf. \cref{fig:adsrn_Sigma}. We will use  Dirichlet (reflecting) boundary conditions which can be viewed as the most natural conditions in the context of stability of the Cauchy horizon. In principle, however, in view of \cite{warnick_massive_wave}, we could also use more general boundary conditions like Neumann or Robin conditions. We will now introduce an appropriate foliation and norms in order to state the well-posedness statement in \cref{sec:wellposedness}.

We will foliate $\mathcal{R}_A \cup \mathcal{R}_B \cup \mathcal{H}_A^+\cup \mathcal{H}_B^+ \cup \mathcal{B}$ with spacelike hypersurfaces. To do so, we let  $\mathcal T$ be a smooth future-directed causal vector field on $\mathcal{R}_A \cup \mathcal{R}_B \cup \mathcal{H}_A^+\cup \mathcal{H}_B^+ \cup \mathcal{B}$ with the properties that 
\begin{align*}
	\mathcal{T} = \begin{cases}
	T & \text{ on } \mathcal{R}_A\cup \mathcal{H}_A^+\\
	-T & \text{ on } \mathcal{R}_B\cup \mathcal{H}_B^+
	\end{cases}
\end{align*} and that $\mathcal{T}$ is a future-directed timelike vector field on $\mathcal{B}$. Now, define the leaves
\begin{align}\label{eq:foliation}
	\Sigma_{t^\ast}:= \Phi^{\mathcal{T}}(t^\ast)[\Sigma_0],
\end{align}
where $\Phi^{\mathcal{T}}$ is the flow generated by $\mathcal{T}$ and $t^\ast\in \mathbb R$ is its affine parameter. We have illustrated some leaves in \cref{fig:adsrn_translation}. 
\begin{figure}[ht]
	\centering
	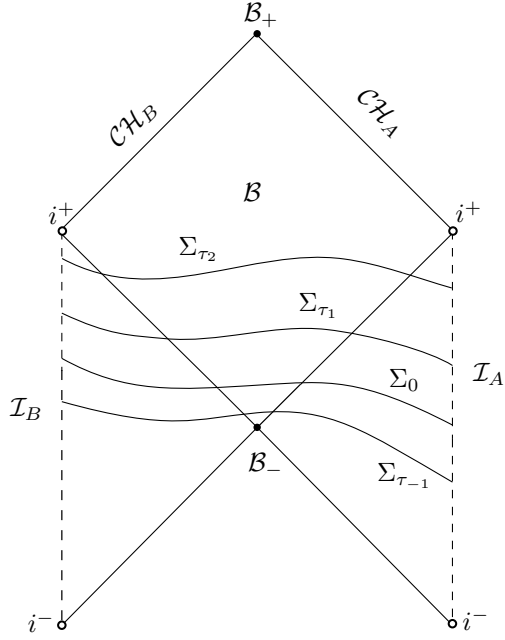
	\caption{Illustration of the foliation with leaves $\Sigma_\tau$ defined in \eqref{eq:foliation}.}
	\label{fig:adsrn_translation}
\end{figure}
\subsubsection{Further coordinates in the exterior region}
In the region $\mathcal{R}_A \cup \Hp_A^+$, we moreover define a global (up to the well-known degeneracy on $\mathbb{S}^2$) coordinate system $(t^\ast,r,\varphi,\theta)$, where $t^\ast$ is the affine parameter of the flow generated by $\mathcal{T}$. Note that on $\mathcal{R}_A \cup \Hp_A^+$ we have $\partial_{t^\ast} =T$ such that $t^\ast( t_2,r) - t^\ast(t_1,r) = t_2 - t_1$ and $t(t_2^\ast ,r ) - t(t_1^\ast ,r) = t_2^\ast - t_1^\ast$. Similarly, we can define such a coordinate system on $\mathcal{R}_B$. 
\subsubsection{Norms on hypersurfaces \texorpdfstring{$\Sigma_{t^\ast}$}{Sigmat}}
By construction $\Sigma_{t^\ast}$ intersects $\mathcal{R}_A$, $\mathcal{R}_B$ and $\mathcal{B}$. 
We will now define norms on $\Sigma_{t^\ast}$ which are adaptations of the norms introduced in \cite{wellposedgustav}. We define
\begin{align}\label{eq:norm}
	\| \psi\|_{H^{k,s}_{\mathrm{RNAdS}}(\Sigma_{t^\ast})}^2 := \|\psi \|^2_{H^k(\Sigma_{t^\ast} \cap \mathcal{B})} + \| \psi \|_{H^{k,s}_{\mathrm{AdS}}(\Sigma_{t^\ast} \cap (\mathcal{R}_A \cup \mathcal{H}_A^+))}^2 + \| \psi \|_{H^{k,s}_{\mathrm{AdS}}(\Sigma_{t^\ast} \cap (\mathcal{R}_B \cup \mathcal{H}_B^+))}^2
\end{align}
and \begin{align}\label{eq:chclass}
	 CH_{\mathrm{RNAdS}}^2:= C^2(\mathbb R_{t^\ast};H^{0,-2}_{\mathrm{RNAdS}} (\Sigma_{t^\ast}) )  \cap C^1 (\mathbb R_{t^\ast}; H^{1,0}_{\mathrm{RNAdS}}(\Sigma_{t^\ast})) \cap C^0(\mathbb R_{t^\ast}; H^{2,0}_{\mathrm{RNAdS}}(\Sigma_{t^\ast})),
\end{align}
where each of the terms appearing in \eqref{eq:norm} will be defined in the following.
\paragraph{\textbf{Norms in the interior region}}
We begin by defining the  first term in \eqref{eq:norm}. We define  $\|\cdot \|^2_{H^k(\Sigma_{t^\ast} \cap \mathcal{B})}$ as the standard Sobolev norm of order $k$ on the Riemannian manifold $(\Sigma_{t^\ast }\cap \mathcal{B},g_{\mathrm{RNAdS}}\restriction_{\Sigma_{t^\ast }\cap \mathcal{B}})$. 

\paragraph{\textbf{Norms in the exterior region}}
Due to the symmetry of the regions $\mathcal{R}_A$ and $\mathcal{R}_B$, we will only define the norms on $\mathcal{R}_A$ in the following. The norms on $\mathcal{R}_B$ are be constructed analogously.
We use the coordinates $(t^\ast, r, \theta,\varphi)$ in $\mathcal{R}_A$ to define the norms
\begin{align*}
	\| \psi \|_{H^{0,s}_{\mathrm{AdS}}(\Sigma_{t^\ast} \cap \mathcal{R}_A )}^2 &:= \int_{\Sigma_{t^\ast} \cap \mathcal{R}_A } r^s |\psi|^2 r^2 \d r \sin\theta \d \theta \d \varphi\\
	\| \psi \|_{H^{1,s}_{\mathrm{AdS}}(\Sigma_{t^\ast} \cap \mathcal{R}_A )}^2 &:=	\| \psi \|_{H^{0,s}_{\mathrm{AdS}}(\Sigma_{t^\ast} \cap \mathcal{R}_A) }^2+  \int_{\Sigma_{t^\ast} \cap \mathcal{R}_A } r^s \left( r^2 |\partial_r \psi|^2 + |\slashed \nabla \psi|^2\right) r^2 \d r \sin\theta \d \theta \d \varphi\\
	\| \psi \|_{H^{2,s}_{\mathrm{AdS}}(\Sigma_{t^\ast} \cap \mathcal{R}_A )}^2 &:=	\| \psi \|^2_{H_{\mathrm{AdS}}^{1,s}({\Sigma_{t^\ast} \cap \mathcal{R}_A })} + \int_{\Sigma_{t^\ast} \cap \mathcal{R}_A } r^s \left( r^4 |\partial_r^2\psi|^2 +r^2 | \slashed \nabla\partial_r \psi|^2 +|\slashed \nabla \slashed \nabla \psi|^2 \right) r^2 \d r \sin\theta \d \theta \d \varphi 
\end{align*}
and similarly for higher order norms. Here and in the following we denote with $\slashed \nabla$ and $\slashed g$ the induced covariant derivative and the induced metric, respectively,  on spheres of constant ($t^\ast,r$). We will also use the notation $|\slashed \nabla \psi|^2:= \slashed g(\slashed \nabla \psi, \slashed\nabla \psi)$.
Now having defined \eqref{eq:norm}, we will define energies in the following.

 \subsubsection{Energies on hypersurfaces \texorpdfstring{$\Sigma_{t^\ast}$}{Sigmat}}
 We set
 \begin{align}\label{eq:energiestotal}
 E_i[\psi](t^\ast) := E_i^A[\psi](t^\ast) + E_i^B[\psi](t^\ast) + E_i^{\mathcal B}[\psi](t^\ast)
 \end{align}
 for $i=1,2$, where all terms in \eqref{eq:energiestotal} will be defined in the following.
\paragraph{\textbf{Energies in the interior region}}
In the interior region we are not concerned with $r$-weights and define the energies as
\begin{align}E_1^{\mathcal{B}}[\psi](t^\ast) &:= \|\psi \|_{H^1(\Sigma_{t_\ast}\cap \mathcal B)}^2 + \| \partial_{t^\ast} \psi \|_{L^2(\Sigma_{t_\ast} \cap \mathcal{B})}^2, \\
E_2^{\mathcal{B}}[\psi](t^\ast) &:= \|\psi \|_{H^2(\Sigma_{t_\ast}\cap \mathcal B)}^2 + \| \partial_{t^\ast} \psi \|_{H^1(\Sigma_{t_\ast} \cap \mathcal{B})}^2 + \| \partial_{t^\ast}^2 \psi\|_{L^2(\Sigma_{t_\ast} \cap \mathcal{B})}^2 . \end{align}
\paragraph{\textbf{Energies in the exterior region}}
To define the energies in the exterior region, it is convenient to start with defining the following energy densities
\begin{align*}
&e_1[\psi]:= \frac{1}{r^2} |\partial_{t^\ast} \psi|^2 + r^2 |\partial_r \psi|^2 + |\slashed \nabla \psi|^2 + |\psi|^2\\
&e_2[\psi] := e_1[\psi]+ e_1[\partial_{t^\ast} \psi] + \sum_{i =1}^3 e_1[\mathcal{W}_i \psi] + r^4 |\partial_r \partial_r  \psi|^2 + r^2|\slashed \nabla \partial_r\psi|^2 + |\slashed\nabla\slashed\nabla \psi|^2
\end{align*}
and their integrals as
\begin{align}
	E_i^A[\psi](t^\ast) := \int_{\Sigma_{t^\ast} \cap (\mathcal{R}_A \cup \mathcal{H}_A^+)} e_i[\psi] r^2 \d r \sin\theta\d \theta \d \varphi
\end{align}
for $i = 1,2$. Note that we will write $E_i^B$ for the analogous energy restricted to $\mathcal{R}_B$.

Also remark the following relation between the norms and energies defined above 
\begin{align*}
&E_1^A[\psi] = \| \psi \|^2_{H^{1,0}_{\mathrm{Ads}} (\Sigma_{t^\ast} \cap \mathcal{R}_A) } + \| \partial_{t^\ast} \psi \|^2_{H^{0,-2}_{\mathrm{AdS}}(\Sigma_{t^\ast} \cap\mathcal{R}_A)}, \\
&	E_2^A[\psi] \sim \sum_i \| \mathcal{W}_i \psi \|_{H^{1,0}_{\mathrm{Ads}} (\Sigma_{t^\ast} \cap \mathcal{R}_A )}^2 +  \| \partial_{t^\ast} \psi \|^2_{H^{1,0}_{\mathrm{Ads}} (\Sigma_{t^\ast}  \cap \mathcal{R}_A)} + \| \psi \|^2_{H^{2,0}_{\mathrm{Ads}}(\Sigma_{t^\ast} \cap \mathcal{R}_A)}+ \|  \partial_{t^\ast}^2 \psi \|^2_{H^{0,-2}_{\mathrm{AdS}}(\Sigma_{t^\ast} \cap \mathcal{R}_A)}.
\end{align*} 
	\subsection{Well-posedness and mixed boundary value Cauchy problem}\label{sec:wellposedness}
Having set up the spacetime and the norms, we will restate the well-posedness result for \eqref{eq:wave} as a mixed boundary value-Cauchy problem. For asymptotically AdS spacetimes, well-posedness was first proved  in \cite{wellposedgustav}.

\begin{theorem}[\cite{wellposedgustav}]\label{prop:wellposedness}  Let the Reissner--Nordström--AdS parameters $(M,Q,l)$ and the Klein--Gordon mass $\alpha < \frac 94$ be as in \eqref{eq:parameters}.
	Let initial data $(\psi_0 ,\psi_1)\in C_c^\infty(\Sigma_0) \times C_c^\infty(\Sigma_0) $ be prescribed on the spacelike hypersurface $\Sigma_0$ and impose Dirichlet (reflecting) boundary conditions on $\mathcal{I}= \mathcal{I}_A \cup \mathcal{I}_B$. 
	
	Then, there exists a smooth solution $\psi \in C^\infty(\mathcal{M}_\mathrm{RNAdS}\setminus \mathcal{CH})$ of \eqref{eq:wave} such that $\psi\restriction_{\Sigma_0} = \psi_0$, $\mathcal{T}\psi\restriction_{\Sigma_0} = \psi_1$. The solution $\psi$ is also unique in the class $C(\mathbb{R}_{t^\ast};H^{1,0}_{\mathrm{RNAdS}}(\Sigma_{t^\ast}))\cap C^1(\mathbb R_{t^\ast}; H^{0,-2}(\Sigma_{t^\ast}))$.
\end{theorem} 
\begin{rmk}The well-posedness statement in \cref{prop:wellposedness} holds true for a more general class of initial data, called a $H^2_{\mathrm{AdS}}$ initial data triplet which give rise to a solution in $CH_{\mathrm{RNAdS}}^2$, see \cite{wellposedgustav}. 
\end{rmk}
\subsection{Energy identities and estimates}
In order to prove energy estimates, it turns out to be useful to introduce two types of energy-momentum tensors. Besides the standard energy-momentum tensor associated to \eqref{eq:wave}, a suitable \textit{twisted} energy-momentum tensor plays an important role in our estimates. Indeed, due to the negative mass term, the standard energy-momentum tensor does not satisfy the dominant energy condition. However, the dominant energy condition can be restored for the twisted  energy-momentum tensor introduced in \cite{breitenlohner,warnick_massive_wave}. In particular, these twisted energies will be used in the interior region, whereas in the exterior region we will work with the standard energy-momentum tensor. 
We will first review the energy estimates in the exterior.
\subsubsection{Energy estimates in the exterior region}
\paragraph{\textbf{Energy-momentum tensor}}
For a smooth function $\phi$ we define
\begin{align}
	\mathbf{T}_{\mu\nu}[\phi] := \operatorname{Re}(\partial_\mu \phi \overline{ \partial_\nu \phi}) - \frac 12 g_{\mu\nu} \left(  \overline{ \partial_\alpha \phi} \partial^\alpha \phi- \frac{\alpha}{l^2} |\phi|^2 \right) .
\end{align}
For a smooth vector field $X$ we also define \begin{align}\label{eq:defjx}
J^X[\phi]:= \mathbf{T}[\phi](X,\cdot) \text{ and } K^X[\phi] := {}^X\pi_{\mu\nu} \textbf{T}^{\mu\nu}[\phi],
\end{align} 
where ${}^X\pi := \mathcal{L}_X g$ is the deformation tensor. The term $K^X$ is often referred to as the ``bulk term'' and satisfies \begin{align}\label{eq:defkx}K^X [\phi] = \nabla^\mu J^X_\mu[\phi]\end{align} if $\phi$ is a solution to \eqref{eq:wave}. Note that if $X$ is Killing, then $K^X$ vanishes. More generally, integrating \eqref{eq:defkx} one obtains an energy identity relating boundary and bulk terms. For more details about the energy-momentum tensor and its usage for standard energy estimates we refer to \cite{lecture_notes}.

\paragraph{\textbf{Boundedness and decay in the exterior region}}
In the exterior regions $\mathcal{R}_A$ and $\mathcal{R}_B$ we have energy decay and boundedness results which have been proved in \cite{wellposedgustav,Gustav_early,decaykg,lower_bound}\footnote{Strictly speaking, in \cite{decaykg} this has been only explicitly proved for Kerr--AdS which includes Schwarzschild--AdS.  However, the same proof as for Schwarzschild--AdS works \emph{completely analogously} for Reissner--Nordström--AdS and we shall not repeat these arguments here.}. To state them we make the following choice of volume forms and normals on the event horizon. We set $\dvol_{\mathcal{H}_A^+} = r^2 \d t^\ast \d \sigma_{\mathbb S^2}$ and $n_{\mathcal{H}_A^+} = T$ and similarly for $\mathcal{H}_B^+$. Moreover, we denote by $\dvol_{\Sigma_{t^\ast}} \sim r \d r \sin\theta\d\theta\d\varphi$ the induced volume form on the spacelike hypersurface $\Sigma_{t^\ast}\cap \mathcal{R}_A$ and by $n^\mu_{\Sigma_t^\ast}$ its future-directed unit normal. 
 We summarize these energy identities and estimates in the following.

\begin{prop}[\cite{wellposedgustav}] 	A solution $\psi$ to \eqref{eq:wave} arising from smooth and compactly supported data on $\Sigma_0$ as in  \cref{prop:wellposedness} satisfies
	\begin{align}\label{eq:energyidnofluxthroughI}
 \int_{\Sigma_{t^\ast_2} \cap \mathcal{R}_A} J_\mu^T[\psi] n^{\mu}_{\Sigma_{t^\ast_2}} \dvol_{\Sigma_{t^\ast_2}} + \int_{\mathcal{H}_A^+(t^\ast_1,t^\ast_2)} J^T_\mu [\psi] n^\mu_{\mathcal{H}^+_A} \dvol_{\mathcal{H}^+_A} = \int_{\Sigma_{t^\ast_1} \cap \mathcal{R}_A} J_\mu^T[\psi] n^{\mu}_{\Sigma_{t^\ast_1}} \dvol_{\Sigma_{t^\ast_1}},
 \end{align}
 where $t^\ast_1 \leq t^\ast_2$ and $\mathcal{H}_A^+ (t_1^\ast,t_2^\ast):= \mathcal{H}_A^+ \cap \{t_1^\ast \leq t^\ast \leq t_2^\ast \}$. The analogous energy identity holds in $\mathcal{R}_B$. In particular,  \eqref{eq:energyidnofluxthroughI} shows that the $T$-energy flux through $\mathcal{I}=\mathcal{I}_A \cup \mathcal{I}_B$ vanishes. 
 
 Moreover, the $T$-energy flux through the event horizon is bounded by initial data
 \begin{align}
 \int_{\Hp_A^+} J^T_\mu [\psi] n^\mu_{\Hp_A^+} \dvol_{\Hp_A^+} + 	\int_{\Hp_B^+} J^T_\mu [\psi] n^\mu_{\Hp_B^+} \dvol_{\Hp_B^+} \lesssim E_1[\psi](0)\label{eq:boundednenergyflux}.
 \end{align}
 Finally, note that 
 \begin{align}
 \int_{\Sigma_{t^\ast}\cap \mathcal{R}_A} J_\mu^T[ {{\psi}}] n^\mu_{\Sigma_{t^\ast}}  \dvol_{\Sigma_{t^\ast}} \sim	\int_{\Sigma_{t_\ast} \cap \mathcal{R}_A} \left[ r^{-2} |\partial_{t^\ast} {\psi}|^2 + \frac{\Delta}{r^2} |\partial_{r} \psi|^2 +    |\slashed\nabla {\psi}|^2 +  |{{\psi}}|^2\right] r^2  \d r \sin\theta \d\theta\d\varphi. \label{eq:hardyinequ}
 \end{align}
 Remark that \eqref{eq:hardyinequ} follows from a Hardy inequality (see \cite[Equation~(50)]{Gustav_early}) which is used to absorb the (possibly) negative contribution from the Klein--Gordon mass term.
 \end{prop}

 \begin{theorem}[{\cite[Theorem~1.1]{lower_bound}}, {\cite[Section~12]{decaykg}}]\label{prop:basicenergyest}
 	A solution $\psi$ to \eqref{eq:wave} arising from smooth and compactly supported data on $\Sigma_0$ as in  \cref{prop:wellposedness} satisfies
 	\begin{align}
 	&	\int_{\Sigma_{t^\ast} \cap \mathcal{R}_A  } e_1[\psi] r^2 \sin\theta \d r  \d\theta \d\varphi \lesssim \int_{\Sigma_{0} \cap \mathcal{R}_A } e_1[\psi] r^2 \sin\theta \d r  \d\theta \d\varphi,\\
 	&\int_{\Sigma_{t^\ast} \cap \mathcal{R}_A} e_2[\psi] r^2 \sin\theta \d r  \d\theta \d\varphi \lesssim \int_{\Sigma_{0} \cap \mathcal{R}_A } e_2[\psi] r^2 \sin\theta \d r  \d\theta \d\varphi,
 	\end{align}
 	and similarly for higher order norms. Moreover, we have the energy decay statements
 	\begin{align}
 	\int_{\Sigma_{t^\ast} \cap \mathcal{R}_A } e_1[\psi] r^2 \sin\theta \d r  \d\theta \d\varphi \lesssim \frac{1}{[\log(2+ t^\ast)]^2}  \int_{\Sigma_{0} \cap \mathcal{R}_A } e_2[\psi] r^2 \sin\theta \d r  \d\theta \d\varphi
 	\end{align}
 	for $t^\ast \geq 0$ and the pointwise decay
 	\begin{align}\label{eq:pointwisedecayontheexterior}
 	\sup_{\Sigma_{t^\ast} \cap \mathcal{R}_A}|\psi|^2 \lesssim \frac{1}{[\log(2+ t^\ast)]^2}  \int_{\Sigma_{0} \cap \mathcal{R}_A }( e_2[\psi]  + e_2[\partial_{t^\ast}\psi] )r^2 \sin\theta \d r  \d\theta \d\varphi
 	\end{align}
 	for $t^\ast \geq 0$
 	in the exterior region $\mathcal{R}_A$ and similarly in $\mathcal{R}_B$.
 	Moreover, just like for Schwarzschild--AdS (cf.~\cite{decaykg}), fixed angular frequencies decay exponentially. More precisely, let $Y_{\ell m}$ denote the spherical harmonics and let $\psi$ be a solution to \eqref{eq:wave} arising from smooth and compactly supported data on $\Sigma_0$. If there exists an $L\in \mathbb{N}$ with $\langle \psi ,Y_{m\ell} \rangle_{L^2(\mathbb{S}^2)} =0 $ for $\ell \geq L$, then 
 	\begin{align}	\label{eq:expdecay}
 	\int_{\Sigma_{t^\ast} \cap \mathcal{R}_A} e_1[\psi] r^2 \sin\theta \d r  \d\theta \d\varphi \lesssim \exp\left(- e^{-C(M,Q,l,\alpha) L}  t^\ast\right) \int_{\Sigma_{0} \cap \mathcal{R}_A } e_1[\psi] r^2 \sin\theta \d r  \d\theta \d\varphi,
 	\end{align}
 	for $t^\ast \geq 0$ and a constant $ C(M,Q,l,\alpha)>0$ only depending on the parameters $M,Q,l,\alpha$.
 \end{theorem}
 \begin{rmk} \label{rmk:rmkdecay} Note that \eqref{eq:expdecay}  also  implies pointwise exponential decay for $\psi$ (assuming $\langle \psi ,Y_{\ell m} \rangle_{L^2(\mathbb{S}^2)} =0 $ for $\ell \geq L$) and all higher derivatives of $\psi$ using standard techniques like commuting with $T$ and $\mathcal{W}_i$, elliptic estimates as well as applying a Sobolev embedding. Moreover, the previous estimates above also hold true for a the more general class of solutions $CH_{\mathrm{RNAdS}}^2$. See {\cite{wellposedgustav}} or \cite[Theorem~4.1]{decaykg} for more details. \end{rmk}
 \begin{rmk}
 	The previous decay estimates have only been stated to the future of $\Sigma_0$ in the region $\mathcal{R}_A$, nevertheless, they  also hold in $\mathcal{R}_B$. Moreover, they also hold true to the past of $\Sigma_0$ for an appropriate foliation for which the leaves intersect $\mathcal{H}_A^-$ and $\mathcal{H}_B^-$, and are transported along the flow of $-T$ for $\mathcal{R}_A\cup\mathcal{H}_A^- $ and along the flow of $T$ for $\mathcal{R}_B\cup\mathcal{H}_B^-$.
 \end{rmk} 
 We now turn to the energy estimates in the interior region $\mathcal{B}$.
\subsubsection{Energy estimates in the interior region}
\paragraph{\textbf{Twisted energy-momentum tensor}}
We begin by defining twisted derivatives. 
 \begin{definition}[Twisted derivative]
 	For a smooth and nowhere vanishing function $f$ we define the \textbf{twisted derivative}
 	\begin{align}
 	\tilde \nabla_\mu := f \nabla_\mu \left(\frac{\cdot}{f}\right)
 	\end{align}
 	and its formal adjoint
 	\begin{align}
 	\tilde \nabla_\mu^\ast:= - \frac{1}{f}\nabla_\mu(f \cdot). 
 	\end{align}
 	We shall refer to $f$ as the \textbf{twisting function}.
 	\label{defn:twistedder}
 \end{definition} 
 \begin{rmk}
 	Note that we can rewrite the Klein--Gordon equation \eqref{eq:wave} in terms of the twisted derivatives as
 	\begin{align}
 	- \tilde \nabla_\mu^\ast \tilde{\nabla}^\mu {{\psi}} - \V {{\psi}}  =0,
 	\end{align}
 	where the potential $\V$ is given by
 	\begin{align}\label{eq:potential}
 	\V=  - \left( \frac{\alpha}{l^2} + \frac{\Box_g f}{f}\right).
 	\end{align}
 \end{rmk}
 Now, we also associate a twisted energy-momentum tensor to the twisted derivatives.
 \begin{definition}[Twisted energy-momentum tensor]Let $f$ be smooth and nowhere vanishing and $\tilde \nabla$ as defined in \cref{defn:twistedder}. We define the \textbf{twisted energy-momentum tensor} associated to \eqref{eq:wave} and $f$ as
 	\begin{align}\label{eq:twistedemt}
 	\tilde {\mathbf T}_{\mu\nu} [{{\phi}}] :=\mathrm{Re}\left( \overline{\tilde \nabla_\mu {{\phi}} }\tilde \nabla_\nu {{\phi}} \right) - \frac{1}{2}g_{\mu\nu} (\overline{\tilde \nabla_\sigma {{\phi}}} \tilde \nabla^\sigma {{\phi}} + \V |{\phi}|^2),
 	\end{align}
 	where $\V$ is as in \eqref{eq:potential} and $\phi$ is any smooth function.
 \end{definition}
 We will now compute the divergence of the twisted energy-momentum tensor.
 \begin{prop}[{\cite[Proposition~3]{warnick_boundedness_and_growth}}]Let $\phi$ be a smooth function and $f$ be a smooth nowhere vanishing twisting function. Then,
 	\begin{align}
 	\nabla_\mu {\tilde {\mathbf T}^\mu}_{\nu} [\phi] = \mathrm{Re}\left(\left(- \tilde \nabla_\mu^\ast \tilde{\nabla}^\mu {{\phi}} - \V {{\phi}} \right) \tilde{\nabla}_\nu \phi\right) + \tilde S_\nu [\phi],
 	\end{align}
 	where \begin{align}
 	\tilde S_\nu[\phi] = \frac{\tilde \nabla^\ast_\nu (f\V)}{2f} |{\phi}|^2 + \frac{\tilde \nabla^\ast_\nu f}{2f} \overline{\tilde{\nabla}_\sigma {{\phi}}} \tilde \nabla^\sigma {\phi}.
 	\end{align}
 	Now, assume that $\phi$ moreover satisfies \eqref{eq:wave} and $X$ is a smooth vector field. Set \begin{align}
 	\tilde J^X_\mu[\phi] := \tilde {\mathbf{T}}_{\mu\nu}[\phi] X^\nu \text{ and }  \; \tilde K^X[\phi] := ~^X\pi_{\mu\nu}\tilde {\mathbf T}^{\mu\nu}[\phi] + X^\nu \tilde S_\nu[\phi]. 
 	\end{align}
 	Then,
 	\begin{align}\label{eq:twiteddefjx}
 	\nabla^\mu \tilde J^X_\mu [\phi ]= \tilde K^X[\phi].
 	\end{align}
 	Finally, note that if the twisting function $f$ associated to $\tilde \nabla$ is chosen such that $\V \geq 0$, then $\tilde{\mathbf T}_{\mu\nu}$ satisfies the dominant energy condition, i.e.\ if $X$ is a future pointing causal vector field, then so is $-\tilde J^X$. \end{prop}
 	
 	We will make use of the twisted energy-momentum tensor in the interior region $\mathcal{B}$ for which we use null coordinates $(u_{\mathcal B}, v_{\mathcal B})$ introduced in \cref{sec:prelim21}. For the rest of the subsection we will drop the index $\mathcal{B}$. Then, setting  \begin{align}
 	\Omega^2(u,v) := -\left(1-\frac{2M}{r} + \frac{Q^2}{r^2} + \frac{r^2}{l^2}\right),
 	\end{align}
 	where $r = r(u,v)$, we write the metric in the interior region $\mathcal{B}$ as 
 	\begin{align}\label{eq:metricinterior}
 	g_{\mathrm{RNAdS}} = -\frac{\Omega^2(u,v)}{2} (\d u \otimes\d v + \d v \otimes \d u) + r^2(u,v) \d \sigma_{\mathbb{S}^2}.
 	\end{align}
 	Note that in the interior we have $r_- < r(u,v) < r_+$ and $\d r_\ast = \frac{r^2}{\Delta} \d r$. 
 In \cref{prop:appendix} in the appendix we have written out the components of the twisted energy-momentum tensor, the twisted 1-jets $\tilde J^X$ and the twisted bulk term $\tilde K^X$ in null components. 
 We will use the notation $\mathcal C_{u_1} := \{ u = u_1 \} $, $\underline{\mathcal C}_{v_1} = \{v = v_1  \}$ for null cones and $\Sigma_{r_1} = \{ r = r_1 \}$ for spacelike hypersurfaces in the interior. Furthermore, we set (in mild abuse of notation)
 \begin{align}
& \mathcal{C}_{u_1}(v_1,v_2) := \{ u = u_1\}\cap \{ v_1 \leq v \leq v_2\},\\
 & \mathcal{C}_{u_1}(r_1,r_2) := \{ u = u_1\}\cap \{ r_1 \leq r \leq r_2\}
 \end{align}
 and analogously for $\Sigma$ and $\underline{\mathcal{C}}$.
 We will also make use of the following notation. For any $\tilde r\in (r_-,r_+)$ we set 
 \begin{align*}
 &v_{\tilde r}(u) := 2r_\ast(\tilde r) - u,\\
 &u_{\tilde r}(v) := 2 r_\ast(\tilde r) - v
 \end{align*}
 and for hypersurfaces with constant $u,v,r$ we denote $n_{{\mathcal{C}}_u},n_{\underline{\mathcal{C}}_v},n_{\Sigma_r}$ as their normals.\footnote{For null hypersurfaces there does not exist a unit norm normal vector, however, for a fixed volume form, there exists a canonical normal vector which we will choose here. Our choice of volume forms and the corresponding normals can be found in \cref{appendix}.}
 
 \paragraph{\textbf{Twisted red-shift vector field}}
 \begin{prop}\label{prop:redshift} 
 	There exist a $\red \in (r_-,r_+)$, a constant $b(M,Q,l,\alpha)>0$, a nowhere vanishing smooth function $f$ associated to the twisted energy momentum tensor and a future directed timelike vector field $N$ such that\begin{align}
 	0 \leq \tilde J^N_\mu[\phi] n_{\underline{\mathcal{C}}_v}^\mu \leq 	b \tilde K^N[\phi]
 	\end{align}
 	for $\mathcal{R}_{\mathrm{red}}:= \{ \red  \leq r\leq r_+ \} \cap \{v \geq 1\}$ and any smooth solution $\phi$ to \eqref{eq:wave}.
 	\begin{proof}
 		This is proven in \cref{sec:twistedredshift}.
 	\end{proof}
 \end{prop}
 
 We will now prove the main estimate which we will use in the  red-shift region in the interior. 
 \begin{prop}\label{prop:redshiftprelim} Let $\phi$ be a smooth solution to \eqref{eq:wave} and let $r_0 \in [\red,r_+)$. Then, for any $1\leq v_1 \leq v_2$ we obtain
 	 \begin{align}
 	 \nonumber
 	 \int_{\cv{v_2}(r_0,r_+)} \tilde J^N_\mu [{\phi}] n^\mu_{\cv{v}} \dvol_{\cv{v}} + \int_{\Sigma_{r_0}(v_1,v_2)} &\tilde J_\mu^N[{\phi}] n^\mu_{\Sigma_r} \dvol_{\Sigma_r}  +  \int_{v_1}^{v_2}  \int_{\cv{v}(r_0,r_+)} \tilde J^N_\mu [{\phi}] n^\mu_{\cv{v}} \dvol_{\cv{v}} \d v \\&\lesssim  \int_{\cv{v_1}(r_0,r_+)}\tilde J^N_\mu [{\phi}] n^\mu_{\cv{v}} \dvol_{\cv{v}}  + \int_{\mathcal{H}(v_1,v_2)}\tilde J_\mu^N[{\phi}] n^\mu_{\mathcal{H}^+} \dvol_{\Hp^+}.
 	 \end{align}
 \end{prop}
 \begin{proof}We apply the energy identity (spacetime integral of \eqref{eq:twiteddefjx}) in the region $\mathcal{R}(v_1,v_2) := \{ r_0 \leq r \leq r_+\} \cap \{ v_1 \leq v \leq v_2\}$ to obtain
 	\begin{align}
 	\nonumber
 	\int_{\cv{v_2}(r_0,r_+)} \tilde J^N_\mu [{\phi}] n^\mu_{\cv{v}} \dvol_{\cv{v}} + \int_{\Sigma_{r_0}(v_1,v_2)} &\tilde J_\mu^N[{\phi}] n^\mu_{\Sigma_r} \dvol_{\Sigma_r}  + \int_{\mathcal{R}(v_1,v_2)} \tilde K^N[\phi] \dvol \\&= \int_{\cv{v_1}(r_0,r_+)}\tilde J^N_\mu [{\phi}] n^\mu_{\cv{v}} \dvol_{\cv{v}}  + \int_{\mathcal{H}(v_1,v_2)}\tilde J_\mu^N[{\phi}] n^\mu_{\mathcal{H}^+} \dvol_{\Hp^+}.
 	\end{align}
 	Finally, the claim follows from \cref{prop:redshift}.
 \end{proof}
 \paragraph{\textbf{Twisted no-shift vector field}}
 	In this region we propagate estimates towards $i^+$ from the red-shift region to the blue-shift region using a $T=\partial_t$ invariant vector field $X$ and a $t$-independent twisting function $f$. Take $\red$ fixed from \cref{prop:redshift} and let $\rblue >r_-$ be close to $r_-$. We will use the no-shift vector field in two different parts of the paper: First, we will use it in the proof of \cref{prop:reddecay2} in the appendix in order to prove well-definedness of the Fourier projections. In this case we will choose $\rblue$ in principle arbitrarily close to $r_-$. The estimate degenerates as we take $\rblue\to r_-$, however for the purpose of \cref{prop:reddecay2} such an estimate is sufficient. Our second application of the no-shift vector field is to propagate decay of the low-frequency part $\psi_\flat$ in the interior (see already \cref{sec:interiorlow}). Here, we will take $\rblue = \rblue(M,Q,l)$ only depending on the black hole parameters as determined in \cref{prop:decayinblueshift}.
 	
 	In either case,  we will choose 
 	\begin{align}X = X_{\mathrm{ns}}:= \partial_u + \partial_v\label{eq:xnoshift}\end{align} as our vector field. (Indeed, any future directed and $T$ invariant vector field would work.) We define our twisting function as \begin{align}
 	\label{eq:fnoshift}f_{\mathrm{ns}}(r) = e^{\beta_{\mathrm{ns}} r}
 	\end{align} for some $\beta_{\mathrm{ns}} = \beta_{\mathrm{ns}}(\rblue) >0$ large enough such that 
 	\begin{align}
 	\V = -\frac{\Box_g f_{\mathrm{ns}}}{f_{\mathrm{ns}}} - \frac{\alpha}{l^2}  = {\Omega^2 } \beta_{\mathrm{ns}}^2 + {\beta_{\mathrm{ns}}} \partial_r (\Omega^2) + \frac{2 \beta_{\mathrm{ns}}}{r}  \Omega^2 - \frac{\alpha}{l^2} \gtrsim 1
 	\end{align} uniformly in $[\rblue,\red]$. In particular, since $r\in [\rblue,\red]$ is bounded away from $r_+,r_-$, we have \begin{align}\tilde{J}^X_\mu[\phi] n_{\Sigma_r}^\mu \gtrsim |\tilde{\nabla}_u {\phi}|^2 + |\tilde{\nabla}_v {\phi}|^2 + |\slashed\nabla {\phi}|^2 + |{\phi}|^2\end{align} for a smooth function $\phi$.  
Our main estimate in the no-shift region is
\begin{prop} \label{prop:noshiftprop} Let $\phi$ be a smooth solution to \eqref{eq:wave} and $r_0 \in [\rblue, \red]$. Then for any $v_\ast \geq 1 $  we have
	\begin{align}\label{eq:noshiftestimate}
		\int_{\Sigma_{r_0}(v_\ast, 2v_\ast )} \tilde J_\mu^X[{{\phi}}] n^\mu_{\Sigma_{ r}}\dvol_{\Sigma_{ r}} \lesssim \int_{\Sigma_{\red}(v_{\red }(u_{r_0} (v_\ast )), 2v_\ast )} \tilde J_\mu^X[{{\phi}}] n^\mu_{\Sigma_{ r}}\dvol_{\Sigma_{ r}},
	\end{align}
	where we remark that 
	$v_\ast  - v_{\red}(u_{r_0}(v_\ast ))) = const.$
	\begin{proof} 	
 	We apply the energy identity (spacetime integral of \eqref{eq:twiteddefjx}) with $X  = \partial_u + \partial_v$ (cf.\ \eqref{eq:xnoshift}) and $f_{\mathrm{ns}}$ as in \eqref{eq:fnoshift} in the region $\{ r_0 \leq r\leq  \red \} \cap \{u < u_{\rblue}(v_\ast)\}\cap \{ v \leq 2v_\ast \}$. The choice of $f_{\mathrm{ns}}$ guarantees the twisted dominated energy condition for the twisted energy-momentum tensor. Together with the coarea formula as well as the facts that $[r_\ast(r_0),r_\ast(\red)]$ is compact and $X$ is $T$ invariant, we conclude
 	\begin{align}
 	\int_{\Sigma_{r_0}(v_\ast,2v_\ast )} \tilde J_\mu^X[{\phi}] n^\mu_{\Sigma_{r}} \dvol_{\Sigma_{r}} \leq &B_1 \int_{r_0 \leq \bar r \leq \red} \int_{\Sigma_{\bar r}(v_{\bar r}(u_{r_0}(v_\ast)) , 2v_\ast  )} \tilde J^X_\mu[{\phi}] n^\mu_{\Sigma_{\bar r}} \dvol_{\Sigma_{\bar r}} \d\bar r \nonumber \\ & +  \int_{\Sigma_{{\red}}(v_{\red}(u_{r_0}(v_\ast) ), 2v_\ast )} \tilde J^X_\mu[{\phi}] n^\mu_{\Sigma_\red} \dvol_{\red}
 	\end{align}
 	for a constant $B_1 = B_1(M,Q,l,\alpha,\Sigma_0,\red,\rblue)$.
 	Similarly, after setting \begin{align}E(\tilde v, \tilde r):= \int_{\Sigma_{\tilde r}(\tilde v, 2v_\ast )} \tilde J_\mu^X[{{\phi}}] n^\mu_{\Sigma_{ r}}\dvol_{\Sigma_{ r}} \end{align}
 	for $\tilde r\in [r_0,\red]$, 
 	we also have
 	\begin{align}
 	E(v_{\tilde r}(u_{r_0}(v_\ast )), \tilde r) \leq \tilde B_1 \int_{\tilde r \leq \bar r \leq \red} E(v_{\bar r } (u_{r_0}(v_\ast )), \bar r) \d\bar r + E(v_{\red}(u_{r_0}(v_\ast )),\red)
 	\end{align}
 	for a constant $\tilde B_1 = \tilde B_1 (M,Q,l,\alpha,\Sigma_0)$. 
 	An application of Grönwall's inequality yields
 	\begin{align}\label{eq:ev1}
 	E(v_{\tilde r}(u_{r_0} (v_\ast )), \tilde r)\lesssim E(v_{\red}(u_{r_0} (v_\ast )), \red)
 	\end{align}
 	which implies the result.
 	 	\end{proof}
 	\end{prop}

We will use an additional vector field in the interior in the blue-shift region $(r_-, \rblue]$. We will however only define it later in the paper in \cref{sec:blueshift} when we actually use it to propagate estimates for the low-frequency part $\psi_\flat$ all the way to the Cauchy horizon. 
\begin{notation} In the main part of the paper we will makes use of the Fourier transform and convolution associated to the coordinate $t$ in $(t,r,\theta,\varphi)$ coordinates as in \eqref{eq:trcoords}. We denote $\mathcal{F}_T$ as the Fourier transform (and $\mathcal{F}_T^{-1}$ as its inverse) defined as
	 \begin{align}\label{eq:deffouriertransform}
	 \mathcal{F}_T[f](\omega,r,\theta,\varphi) := \frac{1}{\sqrt{2\pi}} \int_{\mathbb{R}} e^{-i\omega t} f(t,r,\theta,\varphi) \d t
	 \end{align} 
	in the coordinates $(t,r,\varphi,\theta)$ of $\mathcal{R}_A, \mathcal{R}_B$ and $\mathcal{B}$, respectively. Here, we assume that $t\mapsto f(t,r,\theta,\varphi)$ is (at least) a tempered distribution and \eqref{eq:deffouriertransform}, in general, is to be understood in the distributional sense. Moreover, the convolution $\ast$ associated to the coordinate $t$ is defined as 
	 \begin{align}\label{eq:defconvolution}
	 (f\ast g) (t,r,\theta,\varphi) :=\int_{\mathbb{R}} f(t-s,r,\theta,\varphi) g(s,r,\theta,\varphi) \d s ,
	 \end{align}
	 where we again assume that $t\mapsto f(t,r,\theta,\varphi)$ is a tempered distribution and $t\mapsto g(t,r,\theta,\varphi)$ is a Schwartz function. Here, \eqref{eq:defconvolution}, in general, is to be understood in the distributional sense.
	 	\end{notation}
\section{Main theorem and frequency decomposition}
\label{sec:mainthm}
Now, we are in the position to state our main result
\begin{theorem}\label{thm:main} Let the Reissner--Nordström--AdS parameters $(M,Q,l)$ and the Klein--Gordon mass $\alpha < \frac 94$ be as in \eqref{eq:parameters}. 
	Let $\psi \in C^\infty(\mathcal{M}_\mathrm{RNAdS} \setminus \mathcal{CH})$ be a solution to \eqref{eq:wave} arising from smooth and compactly supported initial data  $(\psi,\mathcal{T} \psi)\restriction_{\Sigma_0} = (\psi_0,\psi_1) \in C_c^\infty(\Sigma_0) \times C_c^\infty(\Sigma_0)$ on $\Sigma_0$ with Dirichlet (reflecting) boundary conditions imposed at $\mathcal{I}_A$ and $\mathcal{I}_B$ (cf. \cref{prop:wellposedness}). Then, $\psi$ is \textbf{uniformly bounded} in the interior region $\mathcal{B}$ satisfying
	\begin{align}\label{eq:uniformboundedness}
	\sup_{\mathcal{B}} |\psi| & \lesssim D[\psi]^{\frac 12},
\end{align}
where $D[\psi]$ is defined as 	
\begin{align} D[\psi]:= E_1[\psi](0) +\sum_{i,j=1}^3 E_1[\mathcal{W}_i \mathcal{W}_j \psi] (0).\label{eq:dpsi}	
\end{align}
Moreover, $\psi$ extends \textbf{{continuously}} to the Cauchy horizon, i.e.\ $\psi \in C^0(\mathcal{M}_{\mathrm{RNAdS}})$. 
\begin{rmk}
The data term $D[\psi]$ in \eqref{eq:dpsi} can be controlled by the initial data $(\psi_0,\psi_1)$ such that \eqref{eq:uniformboundedness} can be written in terms of initial data as 
\begin{align}\nonumber
	\sup_{\mathcal{B}} |\psi|	\leq  C(M,Q,l,\alpha,\Sigma_0)
	&\Big( \|\psi_0 \|_{H^{1,0}_{\mathrm{RNAdS}(\Sigma_0)}} + \|\psi_1 \|_{H^{0,-2}_{\mathrm{RNAdS}(\Sigma_0)}}   \\ &\;\;\;\;\;\; + \sum_{i,j=1}^3 \|\mathcal{W}_i \mathcal{W}_j \psi_0 \|_{H^{1,0}_{\mathrm{RNAdS}(\Sigma_0)}} +\sum_{i,j=1}^3 \|\mathcal{W}_i \mathcal{W}_j\psi_1 \|_{H^{0,-2}_{\mathrm{RNAdS}(\Sigma_0)}} \Big)
	\end{align}
	for a constant $C(M,Q,l,\alpha,\Sigma_0)$ only depending on the parameters $M,Q,l,\alpha$ and the choice of initial hypersurface $\Sigma_0$.
	\end{rmk}

\end{theorem}
		\begin{rmk}\cref{thm:main} can be extended to a more general class of initial data using standard density arguments. In the context of uniform boundedness and continuity at the Cauchy horizon, it is enough to consider smooth and localized initial data. Nevertheless, note that for more general initial data in appropriate Sobolev spaces, already well-posedness becomes more delicate \cite{wellposedgustav}.
	\end{rmk}
	\begin{proof}[Proof of \cref{thm:main}] We split up the proof in four steps, where \emph{Step~3} and \emph{Step~4} are the main parts relying on \cref{sec:lowfreq} and \cref{sec:highfreq}.
			\paragraph{Step 1: Decomposition into low and high frequencies}
			Let \begin{align}\psi \in C^\infty(\mathcal{M}_\mathrm{RNAdS} \setminus \mathcal{CH})\label{eq:defpsi}\end{align}be as in the assumption of \cref{thm:main}.
Now, in $\mathcal{R}_A$, $\mathcal{R}_B$ and in $\mathcal{B}$, define the low frequency part $\psi_\flat$ and the high frequency part $\psi_\sharp$ as 
		\begin{align}\label{eq:fouriertransforminterior}\psi_\flat  :=\frac{1}{\sqrt{2\pi }} \mathcal{F}_T^{-1} \left[ \chi_{\omega_0}\right]  \ast \psi  \text{ and } \psi_\sharp := \psi - \psi_\flat,\end{align}
	where \begin{align}
		\label{eq:chiomega0}
	\chi_{\omega_0}\in C_c^\infty(\mathbb{R}) \text{ such that } \chi_{\omega_0}(\omega) =0 \text{ for }|\omega|\geq  \omega_0\text{ and }\chi_{\omega_0} (\omega)=1 \text{ for }|\omega| \leq \frac 12 \omega_0.	\end{align}
	From  \cref{prop:welldefofpsiflat} in the appendix we know that the low and high frequency parts $\psi_\flat$ and $\psi_\sharp$ in \eqref{eq:fouriertransforminterior} are well-defined and $\psi_\flat$ and $\psi_\sharp$ extend to smooth solutions of \eqref{eq:wave} on $\mathcal{M}_{\mathrm{RNAdS}} \setminus \mathcal{CH}$. 
	  The cut-off frequency  $\omega_0 = \omega_0(M,Q,l,\alpha)>0$ will be chosen in the proof of \cref{prop:propimprove} only depending on $M,Q,l,\alpha$. For convenience we can also assume that $\chi_{\omega_0}$ is a symmetric function which implies that $\psi_\flat$ and $\psi_\sharp$ will be real-valued as long as $\psi$ was real valued. 
		  This concludes Step~1.
		 
		 Having decomposed the solution in low and high frequency parts $\psi_\flat$ and $\psi_\sharp$, we shall now see how the initial data $D[\psi_\flat]$ and $D[\psi_\sharp]$, respectively, can be bounded by the initial data $D[\psi]$ of $\psi$. 
			\paragraph{Step 2: Estimating the initial data of the decomposed solution}
			This step is the content of the following proposition.
			\begin{prop}
Let $\psi$ be as in \eqref{eq:defpsi} and $\psi_\flat,\psi_\sharp$ be as in \eqref{eq:fouriertransforminterior} and recall the definition of $D[\cdot]$ from \eqref{eq:dpsi}. Then,\begin{align}
	D[\psi_\flat]\lesssim D[\psi] \text{ and } D[\psi_\sharp] \lesssim D[\psi].
\end{align}
\begin{proof}
		Since $\psi =  \psi_\flat + \psi_\sharp $, it suffices to obtain a bound of the type $D[\psi_\flat] \lesssim D[\psi]$, where $D[\cdot]$ is defined in \eqref{eq:dpsi}. 
	Because of the Dirichlet conditions imposed at infinity, the energy fluxes through $\mathcal{I}_A$ and $\mathcal{I}_B$ vanish (see \eqref{eq:energyidnofluxthroughI}), and we estimate \begin{align*}D[\psi_\flat] \lesssim \tilde{D}[\psi_\flat],\end{align*} where $\tilde{D}[\psi_\flat]$ is a higher order energy on the hypersurface \begin{align*}\tilde{\Sigma}_0 := \left( \mathcal{R}_A \cap \{t_{\mathcal{R}_A}=0\} \right)  \cup \mathcal{B}_- \cup \left( \mathcal{R}_B \cap \{t_{\mathcal{R}_B}=0\} \right) \end{align*} to be made precise in the following.
		Note also that the normal vector field on $\mathcal{R}_A\cap \tilde \Sigma_0$ is $n_{\tilde \Sigma_0} = \frac{r}{\sqrt \Delta}\partial_t$.
		
		More precisely, due to the support properties of the initial data, there exists a relatively compact 3-dimensional spherically symmetric submanifold $K \subset  \tilde{\Sigma}_0$ with $\mathcal{B}_-\subset K$\footnote{We introduce $K$ just for a technical reason: The energy density $e_1[\cdot]$ defined on $\tilde \Sigma_0\cap \mathcal{R}_A$ degenerates at the bifurcation sphere $\mathcal{B}_-$.} and such that 
		\begin{align}\nonumber
		D[\psi_\flat] \lesssim \tilde D[\psi_\flat] :=& \| \psi_\flat \|_{H^1(K)}^2 + \| n_{\tilde{\Sigma}_0}\psi_\flat \|^2_{L^2(K)} + \sum_{i,j=1}^3 \| \mathcal{W}_i \mathcal{W}_j \psi_\flat \|_{H^1(K)}^2 + \sum_{i,j=1}^3 \| \mathcal{W}_i \mathcal{W}_j n_{\tilde\Sigma_0} \psi_\flat \|_{L^2(K)}^2 \\ \nonumber &+  \int_{\tilde \Sigma_0 \cap \mathcal{R}_A \setminus K}  \left(   e_1[\psi_\flat ]  + \sum_{i,j=1}^3  e_1[\mathcal{W}_i \mathcal{W}_j \psi_\flat ] \right) r^2 \sin\theta \d r \d \theta \d \varphi  \\ &+ \int_{\tilde \Sigma_0 \cap \mathcal{R}_B\setminus K} \left( e_1[\psi_\flat ] + \sum_{i,j=1}^3  e_1[\mathcal{W}_i \mathcal{W}_j \psi_\flat ]\right) r^2 \sin\theta \d r \d \theta \d \varphi. \label{eq:initialdataest}
		\end{align}
		Estimate \eqref{eq:initialdataest} follows from general theory \cite{lecture_notes}, that is a (higher order) energy estimate followed by an application of Grönwall's lemma.
		In order to estimate the energy on the compact hypersurface $K$ we decompose $K$ in $K\cap \mathcal{R}_A $ and $K \cap \mathcal{R}_B$ and estimate the energy on each of those slices independently. Again, in view of the fact that $\mathcal{R}_A$ and $\mathcal{R}_B$ can be treated analogously, we only show the estimate in $\mathcal{R}_A$. 
		Note that all the terms of 
		\begin{align*}
		\| \psi_\flat \|_{H^1(K\cap \mathcal{R}_A)}^2 &+ \| n_{\tilde{\Sigma}_0}\psi_\flat \|^2_{L^2(K\cap\mathcal{R}_A)}+ \sum_{i,j=1}^3	\| \mathcal{W}_i \mathcal{W}_j  \psi_\flat \|_{H^1(K\cap \mathcal{R}_A)}^2 + \sum_{i,j=1}^3 \| \mathcal{W}_i \mathcal{W}_j n_{\tilde{\Sigma}_0}\psi_\flat \|^2_{L^2(K\cap\mathcal{R}_A)} \\ & + \int_{\tilde \Sigma_0 \cap \mathcal{R}_A} \left( e_1[\psi_\flat ]  + \sum_{i,j=1}^3  e_1[\mathcal{W}_i \mathcal{W}_j \psi_\flat ]\right) r^2 \sin\theta \d r \d \theta \d \varphi
		\end{align*}
		are of the form
\begin{align*}
	\int_{\{t=0\}\cap \mathcal{R}_A} f  |\partial^k \psi_\flat|^2 \sin\theta \d r  \d\theta \d\varphi
\end{align*}
for appropriate $T$ invariant weight functions $f\geq 0$ and $T$ invariant coordinate derivatives $\partial \in \{ \partial_t, \partial_r , \partial_\theta,\partial_\varphi \}$ of order $k=0,1,2,3$.  Using that 
\begin{align*}
	\psi_\flat = \frac{1}{\sqrt{2\pi }} \mathcal{F}_T^{-1} \left[ \chi_{\omega_0}\right]  \ast \psi,
\end{align*} 
where $ \mathcal{F}_T^{-1} \left[ {\chi_{\omega_0}}\right]=: \eta$ is a fixed Schwartz function, we conclude---again since $T$ is Killing---that
\begin{align*} 	\int_{\{t=0\}\cap \mathcal{R}_A} f(r)  |\partial^k \psi_\flat|^2(0,r,\varphi,\theta)  \d r \d\sigma_{\mathbb{S}^2}  &= 	\int_{\{r\geq r_+\} \times \mathbb{S}^2 } f(r)  |\eta \ast \partial^k \psi|^2(0,r,\varphi,\theta)  \d r \d\sigma_{\mathbb{S}^2}  \\ & = 
\int_{\{r\geq r_+\} \times \mathbb{S}^2 }  f(r)  \left| \int_{\mathbb{R}} \eta(-s) \partial^k \psi (s,r,\varphi,\theta) \d{s} \right|^2 \d{r} \d\sigma_{\mathbb{S}^2} \\ & \leq 
 \int_{\mathbb{R}} |\eta(s)|\d{s}  \int_{\mathbb{R}}  |\eta(-s)| \int_{\{r\geq r_+\} \times \mathbb{S}^2 }  f(r)   |\partial^k \psi (s,r,\varphi,\theta) |^2   \d{r} \d\sigma_{\mathbb{S}^2} \d{s} \\ & \lesssim  \sup_{s\in\mathbb{R}} \int_{\mathbb{R}}  f(r)|\partial^k \psi (s,r,\varphi,\theta)|^2 \d\sigma_{\mathbb{S}^2} \\ &\lesssim \int_{\{t=0\}\cap\mathcal{R}_A} f(r) |\partial^k \psi |^2(0,r,\varphi,\theta)  \d r\d\sigma_{\mathbb{S}^2}  \lesssim \tilde D[\psi], 
\end{align*}
where we have used boundedness of higher order energies in the exterior which are proved in \cite{Gustav_early,decaykg} and restated in \cref{prop:basicenergyest}. Also note that we can interchange the derivatives with the convolution since $T$ is a Killing vector field. Thus, we conclude that $\tilde D[\psi_\flat] \lesssim \tilde D[\psi]$ and again by Cauchy stability and the vanishing of the energy flux at $\mathcal{I}$ (see \eqref{eq:energyidnofluxthroughI}), we can bound $ \tilde D[\psi] \lesssim D[\psi]$ which finally shows $D[\psi_\flat] \lesssim
D[\psi]$. Hence, $D[\psi_\sharp] \lesssim D[\psi]$ also holds true.			\end{proof}\end{prop}

The previous analysis in Step~1 and Step~2 allows us to treat the low and high frequency parts $\psi_\flat$ and $\psi_\sharp$ completely independently. 
	\paragraph{Step 3: Uniform boundedness for \texorpdfstring{$\psi_\flat$}{psiflat} and \texorpdfstring{$\psi_\sharp$}{psisharp}}
	This step is at the heart of the paper and will be proved in \cref{sec:lowfreq} and \cref{sec:highfreq}. According to \cref{prop:lowbound} and \cref{prop:highbound},
	\begin{align}
		\sup_{\mathcal{B}} |\psi_\flat|^2 \lesssim D[\psi_\flat]
	\end{align}
	and \begin{align}
			\sup_{\mathcal{B}} |\psi_\sharp|^2 \lesssim D[\psi_\sharp].
	\end{align}
	Thus, in view of \emph{Step~2}, we conclude
		\begin{align}
		\sup_{\mathcal{B}} |\psi|^2   \lesssim \sup_{\mathcal{B}} |\psi_\flat|^2  +\sup_{\mathcal{B}} | \psi_\sharp|^2 \lesssim D[\psi_\flat] + D[\psi_\sharp]\lesssim D[\psi]
		\end{align}
		which shows \eqref{eq:uniformboundedness}. 
		
		\paragraph{Step 4: Continuous extendibility beyond the Cauchy horizon} Again, this is proved \cref{sec:lowfreq} and \cref{sec:highfreq}. In particular, in \cref{continuityflat} and \cref{continuitysharp} it is proved that $\psi_\flat$ and $\psi_\sharp$, respectively, are  continuously extendible beyond the Cauchy horizon. Thus, $\psi = \psi_\flat + \psi_\sharp$ can be continuously extended beyond the Cauchy horizon which concludes the proof.
	\end{proof}
\section{Low frequency part \texorpdfstring{$\psi_\flat$}{psiflat}}
\label{sec:lowfreq}
We will begin this section by showing that $\psi_\flat$ decays superpolynomially in the exterior regions $\mathcal{R}_A$ and $\mathcal{R}_B$ (\cref{sec:exteriorlow}). This strong decay in the exterior regions then leads to uniform boundedness of $\psi_\flat$ in the interior $\mathcal{B}$ and continuous extendibility of $\psi_\flat$ beyond the Cauchy horizon. This will be shown in \cref{sec:interiorlow}. In the following, it suffices to only consider $\mathcal{R}_A$ because the region $\mathcal{R}_B$ can be treated completely analogously. 
\subsection{Exterior estimates}\label{sec:exteriorlow}
We will now consider $\psi_\flat$ in the exterior region $\mathcal{R}_A$ and show an integrated energy decay estimate which will eventually lead to the superpolynomial decay for $\psi_\flat$.  First, however, we review the separation of variables for solutions to \eqref{eq:wave}.
\begin{definition} Let $\phi \in CH_{\mathrm{RNAdS}}^2$  be a solution to \eqref{eq:wave} satisfying \begin{align}\label{eq:l1int} \sum_{0\leq i,j\leq 2} \int_{\mathbb R} | \partial_t^i \partial_r^j \langle \phi, Y_{\ell m} \rangle_{\mathbb{S}^2}(r,t) | \d t < \infty\end{align} for $r\in (r_-,r_+)$, $r\in (r_+,\infty)$ and every $|m|\leq \ell$. In the regions $\mathcal{B}$ and $\mathcal{R}_A$, respectively, set 
\begin{align}\label{eq:definitionu}
	u[\phi](r,\omega,\ell,m) := \frac{r}{\sqrt{2\pi}}\int_{\mathbb R} e^{-i\omega t} \langle {\phi}, Y_{\ell m} \rangle_{L^2(\mathbb S^2)} \d t,
\end{align}
where $(Y_{\ell m})_{|m|\leq \ell}$ are the standard spherical harmonics. 
\label{defn:u}
\end{definition}
\begin{prop}\label{rmk:welldefnofu}
	Let $\psi$ be as in \eqref{eq:defpsi} and $\psi_\flat$, $\psi_\sharp$ be as in \eqref{eq:fouriertransforminterior}.
Then, $u[\psi](r,\omega,\ell,m)$, $u[\psi_\flat](r,\omega,\ell,m)$ and $u[\psi_\sharp](r,\omega,\ell,m)$ as in \cref{defn:u} are well-defined and smooth functions of $r,\omega$ in $\mathcal{R}_A$ and $\mathcal{B}$.
\begin{proof} First, note that $\psi^{\ell m} := \langle \psi, Y_{\ell m}\rangle Y_{\ell m}$ is a solution to \eqref{eq:wave}, supported on the fixed angular parameter tuple $(\ell,m)$. Thus, in view of \cref{prop:basicenergyest} and \cref{prop:decayofpsil}, $\psi^{\ell m}(t,r,\theta,\varphi)$ and all its derivatives decay exponentially in $t$ in $\mathcal{R}_A$ and in $\mathcal{B}$ on any $\{ r=const. \}$ slice.
\end{proof}
\end{prop}
\begin{prop}
Let $\phi \in CH_{\mathrm{RNAdS}}^2$ be a $C^2$-solution to \eqref{eq:wave} satisfying \eqref{eq:l1int}. Let $u[\phi]$ be defined as in \eqref{eq:definitionu}.  Then, $u[\phi]$ solves the \textbf{radial o.d.e.} (in $\mathcal{B}$ and $\mathcal{R}_A$)
\begin{align}\label{eq:radialode}
-u^{\prime \prime} + (V_\ell-\omega^2 ) u = 0,
\end{align}
where ${}^\prime = \frac{\d }{\d r_\ast}$, 
\begin{align}
V_\ell(r) = h \left(\frac{\frac{\d h}{\d r}}{r} + \frac{\ell(\ell+1)}{r^2} - \frac{\alpha}{l^2} \right)
\end{align}
and
\begin{align}
h = \frac{\Delta}{r^2} = 1- \frac{2M}{r} + \frac{r^2}{l^2} + \frac{Q^2}{r^2}.
\end{align}
Moreover, in the exterior region $\mathcal{R}_A$ we have $\lim_{r\to\infty} |r^{\frac 12} u[\phi]| = 0$, $ \lim_{r\to\infty} |r^{- \frac 12} u[\phi]^\prime |= 0$.
Finally, note that
\begin{align}
\frac{\d V_\ell}{\d r} = \frac{\d h}{\d r} \left(\frac{\frac{\d h}{\d r}}{r} + \frac{\ell(\ell+1)}{r^2} - \frac{\alpha}{l^2}\right) + h  \left(-\frac{\frac{\d h}{\d r}}{r^2}  + \frac{\frac{d^2 h}{\d r^2}}{r}  - \frac{2\ell(\ell+1)}{r^3}\right).
\end{align}
\begin{proof}The fact that $u[\phi]$ solves the radial o.d.e.\ is a direct computation. For the decay statement as $r\to\infty$, note that $u[\phi](r,\omega, \ell,m) = u[\phi_{\ell m}](r,\omega,\ell,m)$, where $\phi_{\ell m}:=  \langle \phi, Y_{\ell m}\rangle_{\mathbb S^2} Y_{\ell m}$.  In particular, \eqref{eq:expdecay} (together with \cref{rmk:rmkdecay}) then implies    $\int_{-\infty}^{\infty} \left( \int_{r_+}^\infty r^2 |\langle \phi, Y_{\ell m}\rangle_{\mathbb S^2}|^2 \d r\right)^{\frac 12} \d t <\infty.$
	Thus, 
	\begin{align} 
	\left(	\int_{r_+}^\infty  |u[\phi]|^2 
	 \d r \right)^{\frac 12} \lesssim \left(	\int_{r_+}^\infty   \left(\int_{-\infty}^{\infty} r^2 |\langle \phi, Y_{\ell m}\rangle_{\mathbb S^2}| \d t\right)^2  \d r
	  \right)^{\frac 12} 
	\leq   \int_{-\infty}^{\infty} \left( \int_{r_+}^\infty r^2 |\langle \phi, Y_{\ell m}\rangle_{\mathbb S^2}|^2 \d r\right)^{\frac 12} \d t
	< \infty.\label{eq:usquare}
	\end{align}
Since $u[\phi]$ solves \eqref{eq:radialode}, analyzing the indicial equation at the regular singularity $r=\infty$ (see~\cite[Section~2.2.2]{dold}), shows that  $ | r^{\frac 12} u[\phi]| = O(r^{-\sqrt{\frac{9}{4} - \alpha}}) $ and $ |r^{-\frac 12} u[\phi]^\prime  |= O(r^{-\sqrt{\frac{9}{4} - \alpha}}) $ as $r\to\infty$ in order to satisfy \eqref{eq:usquare}.\footnote{The integrability condition \eqref{eq:usquare} corresponds to the Dirichlet boundary condition at infinity on the level of the o.d.e.}
\end{proof}
\end{prop}
Next, we prove that the potential $V_\ell$ has a local maximum for large enough angular parameter $\ell_0$.
\begin{prop}\label{prop:potential}There exists an $\tilde \ell_0(M,Q,l,\alpha) \in \mathbb{N}$ such that for all $\ell \geq \tilde \ell_0$, the potential $V_\ell$ has a local maximum $ r_{\ell,\mathrm{max}} > r_+$ and $V^\prime_\ell \geq 0$ for $r_+\leq r \leq r_{\ell,\mathrm{max}}$. Moreover,
$r_{\ell,\mathrm{max}} \to r_{\mathrm{max}}:=\frac{3}{2}M + \sqrt{\frac{9}{4}M^2 - 2 Q^2}$	as $\ell \to \infty$.
\end{prop}
	\begin{proof}
		Note that for $\ell$ large enough, $V_\ell$ is non-negative in a neighborhood of $r_+$ with $r\geq r_+$. Also, $V_\ell$ vanishes at $r=r_+$. Hence, it suffices to show that $\frac{\d V_\ell}{\d r} $ is negative somewhere for $r\geq r_+$. But note that \begin{align}\label{eq:dvdr}\frac{\d V_\ell}{\d r} = F(r) + {r}^{-3} \ell(\ell+1) \left( r\frac{\d h}{\d r} - {2h}\right) =  F(r) +2 r^{-3} \ell(\ell+1) \left( \frac{3M}{r} -1 - \frac{2Q^2}{r^2}\right) 
		\end{align}  
		for some function $F(r)$ which is independent of $\ell$. 
	Now, first choose $r>r_+$ large enough only depending on $M,Q$ such that the last term is negative. Then, choose $\ell$ large enough such that it dominates the first term which proves that a $r_{\ell,\mathrm{max}}$ as in the statement exists. The limiting behavior $r_{\ell,\mathrm{max}} \to \frac{3}{2}M + \sqrt{\frac{9}{4}M^2 - 2 Q^2}$	as $\ell \to \infty$ also follows from \eqref{eq:dvdr}. This concludes the proof.
	\end{proof}
	Now, we are in the position to prove a frequency localized integrated decay estimate in the exterior region for the bounded frequencies $|\omega|\leq 2 \omega_0$. 
\begin{prop}\label{prop:propimprove}
	Let $u(r_\ast) = u^{(\omega,m,\ell)} (r_\ast)$ solve the radial o.d.e.\ \eqref{eq:radialode} in the exterior $\mathcal{R}_A$ and assume that $\lim_{r\to\infty} |r^{\frac 12} u | = 0 $ and $\lim_{r\to\infty} |r^{-\frac 12} u^\prime |= 0$. Moreover, let $|\omega | \leq 2 \omega_0$, where $\omega_0(M,Q,\ell,\alpha)>0$ small enough will be fixed in the following proof. Then, we have 
	\begin{align}\label{eq:microlocaliled}
	\int_{R_\ast^{-\infty}}^{r = \infty} \frac{ \Delta}{r^4} \left( |u^\prime|^2 + |u|^2(  \ell (\ell+1) + r^2 )\right) \d r_\ast \lesssim - \tilde Q (R_\ast^{-\infty})
	\end{align}
	for all $R_\ast^{-\infty}$ small enough such that $r (R_\ast^{-\infty}) < r_0$, where $ r_0= r_0 (M,Q,l,\alpha)> r_+$ is determined in the following proof. Here, the boundary term $\tilde Q(R_\ast^{-\infty})$ satisfies
	\begin{align}\label{eq:boundryterms1}
		|\tilde Q(R_\ast^{-\infty}) | \lesssim ( |\omega|^2 |u|^2 + |u'|^2) (1+ O_\ell(r-r_+)  )\text{ as } R^{-\infty}_\ast \to -\infty.
	\end{align}
\end{prop}
\begin{proof}
We will first argue that  it suffices to prove \eqref{eq:microlocaliled} for $\ell\geq \ell_0(M,Q,l,\alpha) $ for some fixed $\ell_0(M,Q,l,\alpha) \in \mathbb N_0$. Note that \eqref{eq:microlocaliled} for $\ell \leq \ell_0$ is an easier variant of \cite[Proposition~7.4]{decaykg}. Indeed, we perform the same steps in \cite[Lemma~7.3 and Proposition~7.4]{decaykg} but instead take $a=0$, $\omega_+ =0$ and $H=0$ throughout \cite[Section~7]{decaykg}. This leads to \cite[Proposition~7.4]{decaykg} with $L$ replaced by $\ell_0$. The estimate on the boundary term follows from \cite[Section~9.3]{decaykg}. 

We will now consider $\ell \geq \ell_0$, where $\ell_0$ is determined below. Let $r_0, r_1$ depending only on $M,Q,l,\alpha$ be such that $r_+ < r_0 < r_1 < r_\text{max}- \delta$, where $r_{\text{max}}$ is defined in \cref{prop:potential}. Here, $\delta = \delta(\ell_0)>0$ is such that $V^\prime \geq 0$ for all $r_+\leq r \leq r_\text{max} - \delta$, cf.\ \cref{prop:potential}. We can make $\delta(\ell_0)$ as small as we want by choosing $\ell_0$ sufficiently large. Now, we choose $\omega_0(M,Q,l,\alpha)>0$ small enough and $\ell_0$ large enough such that 
\begin{align}
\begin{split}\label{eq:V+h}
&V - \omega^2 + \frac{\Delta }{4l^2 r^2}\gtrsim  \ell(\ell+1) + \frac{\Delta }{r^2} \; \text{ for $r\geq r_0$,}\\
& V -\omega^2 \geq 0 \;\;\; \text{ for } r_0 \leq r \leq r_1,
\end{split}
\end{align}
 and for all $|\omega|\leq 2  \omega_0$, $\ell \geq \ell_0$. 
For smooth $f(r_\ast)$ and $\tilde h(r_\ast)$, we define the currents
\begin{align}
&	Q^f := f \left[ |u^\prime|^2 + (\omega^2 - V) |u|^2 \right] + f^\prime \operatorname{Re}(u^\prime \bar u) - \frac{1}{2}f^{\prime \prime} |u|^2,\\
& Q^{\tilde h} := \tilde h \operatorname{Re}(\bar u u^\prime) - \frac{1}{2}\tilde h^\prime |u|^2
\end{align}
with
\begin{align}
	&{Q^f}^\prime = \frac{\d Q^f}{\d r_\ast} = 2 f^\prime |u^\prime|^2 - f V^\prime |u|^2 - \frac{1}{2}f^{\prime\prime\prime} |u|^2 , \\
	&{Q^{\tilde h}}^\prime  = \frac{\d Q^{\tilde h}}{\d r_\ast}=\tilde  h \left[|u^\prime|^2 + (V-\omega^2) |u|^2 \right] - \frac{1}{2} \tilde h^{\prime\prime} |u|^2 ,
\end{align}
where we recall that $~^\prime$ denotes the derivative $\frac{\d }{\d r_\ast}$. 
Thus,
\begin{align}\nonumber
{Q^f}^\prime + {Q^{\tilde h}}^\prime = |u^\prime|^2 (2f^\prime + \tilde h) + & |u|^2 \left( - f V^\prime - \frac{1}{2} f^{\prime \prime \prime} + \tilde h(V-\omega^2) - \frac{1}{2}\tilde h^{\prime \prime} \right).
\end{align}
We choose a smooth $f\leq 0$ such that
\begin{itemize}
	\item $f$ is monotonically increasing,
	\item $f = - 1/r^2$ in a neighborhood of $r=r_+$,
	\item $f\leq -c_1$ for $r_+\leq r \leq r_1$ and some $c_1(M,Q,l)>0$,
	\item $\Delta \lesssim f^\prime \lesssim \Delta$ for $r_+\leq r \leq r_1$,
	\item $|f^{\prime \prime \prime}|\lesssim \Delta$,
	\item $f=0$ for $r\geq r_\text{max}-\delta$.
\end{itemize}
and a smooth $ \tilde h\geq 0$ such that
\begin{itemize}
	\item $ \tilde h =0 $ for $r \leq r_0$,
	\item $|\tilde h^{\prime \prime}|\lesssim 1$ for $r_0<r_1$,
	\item $\tilde h=1$ for $r \geq  r_1$.
\end{itemize}
Then, we have
\begin{align}\label{eq:cases}
	{Q^f}^\prime + {Q^{\tilde h}}^\prime 
	\geq \begin{cases}
	2 f^\prime |u^\prime|^2 + |u|^2 (- f V^\prime - \frac{1}{2} f^{\prime\prime\prime}) & \text{ for } r_+\leq r\leq r_0,\\
	2 f^\prime |u^\prime|^2 + |u|^2 (- f V^\prime - \frac{1}{2} f^{\prime\prime\prime} - \frac 12 \tilde h^{\prime \prime}  ) & \text{ for } r_0 \leq r \leq r_1,\\
	|u^\prime|^2 + |u|^2 (- \frac 12 f^{\prime\prime\prime} + (V-\omega^2)) &\text{ for } r \geq r_1 .
	\end{cases} 
\end{align}
Thus, choosing $\ell_0(M,Q,l,\alpha)$ large enough (and $\omega_0(M,Q,l,\alpha) >0 $ possibly smaller) and using \eqref{eq:cases}, \eqref{eq:V+h}, \eqref{eq:dvdr} and the properties of $f$ and $\tilde h$, we have
\begin{align}
	{Q^f}^\prime + {Q^{\tilde h}}^\prime  \gtrsim \frac{ \Delta}{r^4} \left( |u^\prime|^2 + |u|^2 ( \ell (\ell+1) + r^2)  \right) 
\end{align}
for $r_+ \leq r \leq r_{\text{max}} - \delta$ and
\begin{align}\label{eq:hardyterm}
		{Q^f}^\prime + {Q^{\tilde h}}^\prime \gtrsim |u^\prime|^2 + (V-\omega^2)|u|^2 \geq  |u^\prime|^2 -  |u|^2 \frac{\Delta }{4l^2 r^2} + \tilde c  \left( \ell(\ell+1)  + \frac{\Delta}{r^2}\right)|u|^2
\end{align}
for $r \geq r_{\text{max}} - \delta$ and some $\tilde c(M,Q,l,\alpha)>0$. 
Integrating ${Q^f}^\prime + {Q^{\tilde h}}^\prime $  in the region $r_\ast \in (R_\ast^{-\infty}, r_\ast (r=+\infty ) )$ and applying the following Hardy inequality (see \cite[Lemma~7.1]{decaykg})
 \begin{align} \int_{r = r_{\text{max}} - \delta}^{r=\infty} |u^\prime|^2 \d r_\ast \geq \int_{r = r_{\text{max}} - \delta}^{r=\infty} \frac{\Delta }{4l^2 r^2}  |u|^2 \d r_\ast \end{align} to control the negative signed term in \eqref{eq:hardyterm}, yields
\begin{align}\label{eq:cutoffestimate}
	\int_{R_\ast^{-\infty}}^{r=+\infty} \frac{ \Delta}{r^4} \left( |u^\prime|^2 \chi_{\{ r\leq r_{\text{max}}-\delta \} }+ |u|^2 ( \ell (\ell+1) + r^2 )\right) \d r_\ast \lesssim - Q^f (R_\ast(-\infty)).
\end{align}
Note that we use $\lim_{r\to\infty} |r^{\frac 12} u |= 0 $ and $\lim_{r\to\infty} |r^{- \frac 12} u^\prime | =0$ to apply the Hardy inequality. 
To obtain control of $|u^\prime|^2$ in the region $r\geq r_{\text{max}} - \delta$ in \eqref{eq:cutoffestimate} we just add a small portion of the integral over \eqref{eq:hardyterm}. This proves
\begin{align}\label{eq:cutoffestimate1}
\int_{R_\ast^{-\infty}}^{r=+\infty}    \frac{ \Delta}{r^4}\left( |u^\prime|^2  +|u|^2 ( \ell (\ell+1) + r^2 )\right) \d r_\ast \lesssim - Q^f (R_\ast(-\infty)),
\end{align}
where $|Q^f(R_\ast^{-\infty})| \lesssim ( |\omega|^2 |u|^2 + |u'|^2) (1 + O_\ell(r-r_+)$ as $R_\ast^{-\infty} \to -\infty$ is satisfied by the construction of $f$.
\end{proof}
With the frequency localized integrated energy decay estimate of \cref{prop:propimprove} we will now prove a \emph{local} integrated energy decay estimate in physical space.
Indeed, a naive application of Plancherel's theorem to \eqref{eq:microlocaliled} gives a \emph{global} integrated energy estimate. However, localizing this energy decay requires some sort of cut-off which does not respect the compact frequency support. Nevertheless, by carefully choosing a localization, we can show that the error term  decays superpolynomially in time. At this point we shall remark that we do expect $\psi_\flat$ to decay exponentially. However, for our problem, superpolynomial decay in the exterior is (more than) sufficient. 
\begin{prop}\label{prop:microiled} 	Let ${{\psi}_\flat}$ be as in \eqref{eq:fouriertransforminterior}. Then, for any $q>1$, $   \tau_1 \geq 0 $ and in view of \eqref{eq:hardyinequ}, we have the integrated energy decay estimate
	\begin{align}
	\nonumber &
		\int_{\mathcal{R}_A \cap \{t^\ast \geq 2 \tau_1 \}} \left[ r^{-2} |\partial_{t^\ast} {\psi}_\flat|^2 + r^{-2} |\partial_{r_\ast} {\psi}_\flat|^2 +    |\slashed\nabla {\psi}_\flat|^2 +  |{{\psi}_\flat}|^2\right] r^2 \d t^\ast \d r \sin\theta \d\theta\d\varphi \\  &\lesssim\int_{\Sigma_{ \tau_1} \cap \mathcal{R}_A} J^T_\mu[{{\psi}_\flat}] n^\mu +  \frac{C(q)}{{1+ \tau_1^q }} \,{ \int_{\Sigma_0} J^T_\mu[\psi_\flat] n^\mu_{\Sigma_0} \dvol_{\Sigma_0}},\label{eq:est}
	\end{align}
	where $C(q)>0$ is a constant only depending on $q$. 
	Moreover, for any $\tau_2 \geq 2\tau_1$, this directly implies
	\begin{align}\nonumber\int_{\Sigma_{\tau_2}\cap \mathcal{R}_A} J_\mu^T[{{\psi}_\flat}]  n_{\Sigma_{\tau_2}}^\mu \dvol_{\Sigma_{\tau_2}} + & \int_{\mathcal{R}_A \cap \{t^\ast \geq 2 \tau_1 \}} \left[ r^{-2} |\partial_{t^\ast} {\psi}_\flat|^2 + r^{-2} |\partial_{r_\ast} {\psi}_\flat|^2 +    |\slashed\nabla {\psi}_\flat|^2 +  |{{\psi}_\flat}|^2\right] r^2 \d t^\ast \d r \sin\theta \d\theta\d\varphi
	\\ &	\lesssim \int_{\Sigma_{ \tau_1} \cap \mathcal{R}_A} J^T_\mu[{{\psi}_\flat}] n^\mu +\frac{ C(q)}{{1+ \tau_1^q}}{\int_{\Sigma_0\cap \mathcal{R}_A} J^T_\mu[\psi_\flat] n^\mu_{\Sigma_0} \dvol_{\Sigma_0}}
	\end{align}
	for the $T$-energy.
\end{prop}

\begin{proof} 
In order to show  \eqref{eq:est} we will first construct an auxiliary solution $\Psi$ of \eqref{eq:wave}. We set initial data for $\Psi$ on $\Sigma_{ \tau_1}$ as $(\Psi_0,\Psi_1) := (\psi_\flat, \mathcal{T}\psi_\flat)\restriction_ {\Sigma_{ \tau_1} \cap \mathcal{R}_A}$. Then, we will define data $\Psi_2$ on  $\mathcal{H}_A^+ \cap \{ t^\ast \leq \tau_1\}$ such that the data can be extended to a $C^k$ function in a neighborhood of $\mathcal{H}_A^+ \cap \{t^\ast =  \tau_1\}$ for some finite regularity $k$. Choosing the regularity $k$ large enough will guarantee well-posedness. More precisely, in local coordinates $(t^\ast,r,\theta,\varphi)$ and for $r= r_+$, we define
	\begin{align}
	 \Psi_2 (t^\ast  ,r_+,\varphi,\theta) := \sum_{j=1}^k \lambda_j \psi_\flat\restriction_{ \{ t^\ast \geq  \tau_1 \}} (- j( t^\ast-{\tau}_1) + {\tau}_1, r_+, \varphi,\theta )
	\end{align} 
	for $t^\ast \leq   \tau_1$ and some uniquely determined $(\lambda_j)_{1\leq j \leq k}$ such that 
	 \begin{align} \mathbb{R}\times \mathbb{S}^2 \ni (t^\ast ,\varphi,\theta) \mapsto \begin{cases}	 \Psi_2 (t^\ast  ,r_+,\varphi,\theta) & \text{ for } t^\ast \leq  \tau_1 \\ \psi_\flat  (t^\ast  ,r_+,\varphi,\theta) 
	& \text{ for } t^\ast >  \tau_1 \end{cases}\end{align} is $C^k$. Indeed, the function is smooth everywhere except at $t^\ast =  \tau_1$.   \begin{figure}[ht]\centering
	 	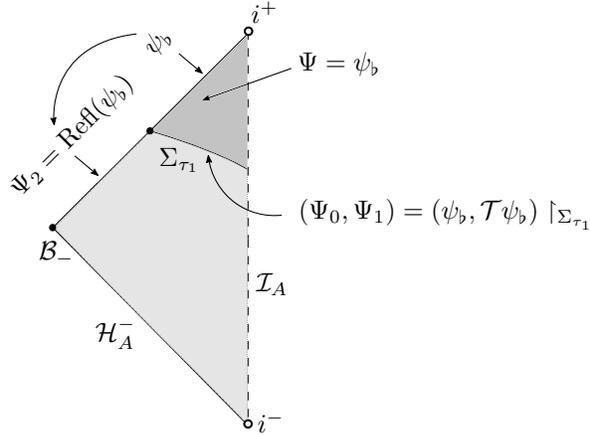
	 		 	\caption{In the darker shaded region $J^+(\Sigma_{\tau_1})\cap \mathcal{R}_A$ we have that $\Psi = \psi_\flat$ and in the lighter shaded region we can estimate the energy of $\Psi$ in terms of $\psi_\flat$. This holds true as $\Psi_2$ is the $C^k$ reflection of $\psi_\flat$ from $\mathcal{H}_A^+ \cap \{ t^\ast \geq \tau_1\}$ to $\mathcal{H}_A^+ \cap \{ t^\ast < \tau_1\}$. }
	 	\label{fig:data}
	 \end{figure}
 Now, we consider the mixed boundary value-Cauchy-characteristic problem, where we impose data as follows. On the null hypersurface $\mathcal{H}_A^+ \cap \{ t^\ast \leq  \tau_1 \}$ we impose $\Psi_2$. This null cone intersects the spacelike hypersurface $\Sigma_{ \tau _1}$ on which we have prescribed $(\Psi_0,\Psi_1)$ as data. As before, we assume the Dirichlet condition on $\mathcal{I}_A$.  For fixed $k >0$ large enough, this is a well-posed problem and can be solved backwards and forwards in $\mathcal{R}_A$  \cite[Theorem~2]{mixed_prob}. We will call the arising solution $\Psi$ and by uniqueness note that $\Psi = \psi_\flat$ on $ ( \mathcal{R}_A\cup \mathcal{H}_A^+ ) \cap J^+(\Sigma_{ \tau_1})$. Indeed, analogously to $\psi_\flat$, we have $\Psi \in CH^2_{\mathrm{RNAdS}}$ and by choosing $k$ large enough, we can make $\Psi$ arbitrarily regular, in particular $C^2$. Moreover, $\Psi$ decays logarithmically and $\langle\Psi,Y_{\ell m}\rangle Y_{\ell m}$  decays exponentially towards $i^+$ and $i^-$ on a $\{ r =const.\}$ hypersurface.\footnote{We will use this statement only in a qualitative way such that $u[\Psi_\flat]$ is well-defined in \eqref{eq:defPsiflat4} and satisfies \eqref{eq:microlocaliled}.}  Refer to \cref{fig:data} for a visualization of the Cauchy-characteristic problem with Dirichlet boundary conditions.
	  
	Analogously to $\psi = \psi_\flat + \psi_\sharp$, we decompose the new solution $\Psi$ in low and high frequencies $\Psi = \Psi_\flat + \Psi_\sharp$: We define
	  \begin{align}
	  	\Psi_\flat := \frac{1}{\sqrt{2\pi}}\mathcal{F}_T^{-1} [ \chi_{2\omega_0}]\ast  \Psi , \text{ and } \Psi_\sharp := \Psi - \Psi_\flat,
	  \end{align}
	  where $\chi_{2\omega_0}$ is a smooth cutoff function such that $\chi_{2\omega_0} = 1$ for $|\omega| \leq \omega_0$ and $\chi_{2\omega_0} = 0$ for $|\omega|\geq 2 \omega_0$. 
	 Now, note that from the $T$-energy identity  \eqref{eq:energyidnofluxthroughI} we have 
	 \begin{align}
	 	\int_{\mathcal{H}_A^+ (  \tau_1, \infty )} J^T_\mu [\psi_\flat] n^\mu_{\mathcal{H}} \dvol_{\mathcal{H}} = 	\int_{\Sigma_{ \tau_1}\cap\mathcal{R}_A} J^T_\mu[\psi_\flat] n^\mu_{\Sigma_{ \tau_1}} \dvol_{\Sigma_{ \tau_1}}
	 \end{align}
	 as the flux through $\mathcal{I}_A$ vanishes in view of the Dirichlet boundary condition at $\mathcal{I}_A$. Here, we use the notation $\mathcal{H}_A^+ (  a, b ) := \mathcal{H}_A^+ \cap \{a< t^\ast < b \}$.
Moreover, from the $T$ energy identity, we have
\begin{align} \nonumber
	 	\int_{\mathcal{H}_A^- } J^T_\mu[\Psi] n^\mu_{\mathcal{H}}
	 	\dvol_{\mathcal{H}} & = 	
 	\int_{\Sigma_{ \tau_1}\cap\mathcal{R}_A} J^T_\mu[\psi_\flat] n^\mu_{\Sigma_{ \tau_1}} \dvol_{\Sigma_{ \tau_1}} 
 	 + \int_{\mathcal{H}_A^+ ( -\infty,  \tau_1)} 
 	 J^T_\mu [\Psi] n^\mu_{\mathcal{H}} \dvol_{\mathcal{H}} \\ & \lesssim \int_{\Sigma_{ \tau_1}\cap\mathcal{R}_A} J^T_\mu[\psi_\flat] n^\mu_{\Sigma_{ \tau_1}} \dvol_{\Sigma_{ \tau_1}}  + 	\int_{\mathcal{H}_A^+ (  \tau_1, \infty )} J^T_\mu [\psi_\flat] n^\mu_{\mathcal{H}} \dvol_{\mathcal{H}} 
 	\nonumber\\ &  \lesssim 
 	 \int_{\Sigma_{ \tau_1}\cap\mathcal{R}_A} J^T_\mu[\psi_\flat] n^\mu_{\Sigma_{ \tau_1}} \dvol_{\Sigma_{ \tau_1}} .  \label{eq:usedreflection}
\end{align}
We have used the estimate  \begin{align*}\int_{\mathcal{H}_A^+ ( -\infty,  \tau_1)} 
J^T_\mu [\Psi] n^\mu_{\mathcal{H}} \dvol_{\mathcal{H}} \lesssim	\int_{\mathcal{H}_A^+ (  \tau_1, \infty )} J^T_\mu [\psi_\flat] n^\mu_{\mathcal{H}} \dvol_{\mathcal{H}}\end{align*} which follows from our construction of the initial data.
Thus,
\begin{align}\label{eq:energyestimatehtosigma}
	\int_{\mathcal{H}_A^- } J^T_\mu[\Psi] n^\mu_{\mathcal{H}} \dvol_{\mathcal{H}} + \int_{\mathcal{H}_A^+ } J^T_\mu[\Psi] n^\mu_{\mathcal{H}} \dvol_{\mathcal{H}}
	\lesssim 
	 \int_{\Sigma_{ \tau_1}\cap\mathcal{R}_A} J^T_\mu[\Psi] n^\mu_{\Sigma_{ \tau_1}} \dvol_{\Sigma_{ \tau_1}}   .
\end{align}

Now, note that $u[\Psi_\flat]$ defined as  \begin{align}\label{eq:defPsiflat4}
u[\Psi_\flat] (r,\omega,\ell,m) = \frac{r}{\sqrt{2\pi}}\int_{\mathbb R} e^{-i\omega t} \langle {{\Psi}_\flat}, Y_{\ell m} \rangle_{L^2(\mathbb S^2)} \d t
\end{align}
satisfies the assumptions of  \cref{prop:propimprove} such that \eqref{eq:microlocaliled} holds true for $u[\Psi_\flat]$. 
We now integrate the frequency localized energy estimate \eqref{eq:microlocaliled} associated to $ u[\Psi_\flat]$ in $\omega$ and sum over all spherical harmonics. There are two main terms appearing and we will estimate them in the following. This step is similar to \cite[Sections~9.1 and 9.3]{decaykg} so we will be rather brief. 
An application of Plancherel's theorem for the integrated left hand side of \eqref{eq:microlocaliled} yields
\begin{align}\nonumber 
\int_{\mathcal{R}_A}& \left[ |\partial_t {\Psi}_\flat|^2 +  |\partial_{r_\ast} {\Psi}_\flat|^2 +  r^2 |\slashed\nabla {\Psi}_\flat|^2 + r^2 |{{\Psi}_\flat}|^2\right]\d t^\ast \d r \sin\theta \d\theta\d\varphi \\ 
&\lesssim
	\lim_{R_\ast^{-\infty} \to -\infty} \sum_{m\ell}   \int_{\mathbb R} \d \omega \int_{R_\ast^{-\infty}}^{r=\infty} \d r_\ast \frac{\Delta}{r^4} \left[ \omega^2 |u[\Psi_\flat]|^2 + |u[\Psi_\flat]^\prime|^2 +  \ell(\ell+1) |{u[\Psi_\flat]}|^2 +  r^2 |u[\Psi_\flat]|^2 \right].
\end{align}
To estimate the boundary term on the right hand side of \eqref{eq:microlocaliled}, we first decompose $u[\Psi_\flat]$ as $u[\Psi_\flat] = a(\omega,m,\ell) u_1 + b(\omega,m,\ell) u_2$, where $u_1, u_2$ are defined as the unique solutions to the radial o.d.e.\ \eqref{eq:radialode} in the exterior satisfying $u_1 = e^{i\omega r_\ast} + O_\ell (r-r_+)$ and $u_2 = e^{-i \omega r_\ast} + O_\ell(r-r_+)$ as $r\to r_+$ ($r_\ast \to -\infty$). Here, $a = a(\omega,\ell, m)$ and $b = b (\omega,\ell, m)$ are the unique coefficients of the decomposition.
Then, in view of \eqref{eq:boundryterms1} and  $u_1^\prime = i \omega u_1 + O_\ell(r-r_+)$, $u_2^\prime = -i \omega u_2 + O_\ell(r-r_+)$, we estimate 
\begin{align}|\tilde Q| & \nonumber \lesssim \left( |\omega|^2 |a(\omega)|^2 |u_1|^2 +  |\omega|^2 |b(\omega)|^2 |u_2|^2 \right) ( 1 + O_\ell (r-r_+)) \\ & = \left( |\omega|^2 |a(\omega)|^2  +  |\omega|^2 |b(\omega)|^2  \right) ( 1 + O_\ell (r-r_+)) \end{align} as $r\to r_+$. Now, using that $\omega a(\omega)$, $\omega b(\omega)$ are in $L_\omega^1(\mathbb R)$ and in $L_\omega^2(\mathbb R)$ (note that they have compact support), an application of the Riemann--Lebesgue Lemma, the Fourier inversion theorem and Plancherel's theorem shows that $\sum_{m \ell } \int_{\mathbb R} |\omega|^2 (|a(\omega,\ell,m)|^2 + |b(\omega,\ell,m)|^2 ) \d \omega \lesssim  \int_{\mathcal{H}_A^+} |T\Psi_\flat|^2  + \int_{\mathcal{H}_A^-} |T\Psi_\flat|^2  \leq 2\int_{\mathcal{H}_A^+} |T\Psi_\flat|^2   $, where the last inequality follows from the $T$ energy identity  $\int_{\mathcal{H}_A^+} |T\Psi_\flat|^2  = \int_{\mathcal{H}_A^-} |T\Psi_\flat|^2 $ in the region $\mathcal{R}_A$. Thus, we conclude the global integrated energy decay statement
\begin{align}\label{eq:Psiflat}
	\int_{\mathcal{R}_A}& \left[  \frac{1}{r^2} |\partial_t {\Psi}_\flat|^2 +  \frac{1}{r^2} |\partial_{r_\ast} {\Psi}_\flat|^2 +  |\slashed\nabla {\Psi}_\flat|^2 + |{{\Psi}_\flat}|^2\right]\dvol \lesssim \int_{\mathcal{H}_A^+}|T{\Psi}_\flat|^2 .
\end{align}
Hence, in view of $\psi_\flat = \Psi$ in $\{ t^\ast \geq  \tau_1\}\cap \mathcal{R}_A$ we have
\begin{align}\nonumber& \int_{\mathcal{R}_A \cap \{ t^\ast \geq  2\tau_1\}} \left[  \frac{1}{r^2} |\partial_t {\psi}_\flat|^2 +  \frac{1}{r^2} |\partial_{r_\ast} {\psi}_\flat|^2 +  |\slashed\nabla {\psi}_\flat|^2 + |{{\psi}_\flat}|^2\right]\dvol \\
& = \int_{\mathcal{R}_A \cap \{ t^\ast \geq  2\tau_1 \}} \left[  \frac{1}{r^2} |\partial_t {\Psi}|^2 +  \frac{1}{r^2} |\partial_{r_\ast} {\Psi}|^2 +  |\slashed\nabla {\Psi}|^2 + |{{\Psi}}|^2\right]\dvol \nonumber \\
& \lesssim \int_{\mathcal{R}_A \cap \{ t^\ast \geq  2\tau_1 \}} \left[  \frac{1}{r^2} |\partial_t {\Psi_\flat}|^2 +  \frac{1}{r^2} |\partial_{r_\ast} {\Psi_\flat}|^2 +  |\slashed\nabla {\Psi_\flat}|^2 + |{{\Psi_\flat}}|^2\right]\dvol \nonumber \\
& ~ + \int_{\mathcal{R}_A \cap \{ t^\ast \geq  2\tau_1 \}} \left[  \frac{1}{r^2} |\partial_t {\Psi_\sharp}|^2 +  \frac{1}{r^2} |\partial_{r_\ast} {\Psi_\sharp}|^2 +  |\slashed\nabla {\Psi_\sharp}|^2 + |{{\Psi_\sharp}}|^2\right]\dvol \nonumber \\
		 \nonumber  & \lesssim \int_{\mathcal{H}^+_A}|T{\Psi}_\flat|^2 + \int_{t^\ast \geq  2\tau_1} \int_{\Sigma_{t^\ast }\cap\mathcal{R}_A} J^T_\mu [\Psi_\sharp] n_{\Sigma_{t^\ast }}^\mu  \dvol_{\Sigma_{t^\ast }} \d t^\ast\\ \nonumber
				& \lesssim \int_{\mathcal{H}^+_A}|T{\Psi}|^2 + \int_{t^\ast \geq  2\tau_1} \int_{\Sigma_{t^\ast }\cap\mathcal{R}_A} J^T_\mu [\Psi_\sharp] n_{\Sigma_{t^\ast }}^\mu  \dvol_{\Sigma_{t^\ast }} \d t^\ast
		\\ \nonumber &\lesssim \int_{\Sigma_{ \tau_1}\cap\mathcal{R}_A} J^T_\mu[\Psi] n^\mu_{\Sigma_{ \tau_1}} \dvol_{\Sigma_{ \tau_1}} + \int_{t^\ast \geq  2\tau_1} \int_{\Sigma_{t^\ast }\cap\mathcal{R}_A} J^T_\mu [\Psi_\sharp] n_{\Sigma_{t^\ast }}^\mu  \dvol_{\Sigma_{t^\ast }} \d t^\ast\\  & =  \int_{\Sigma_{ \tau_1}\cap\mathcal{R}_A} J^T_\mu[\psi_\flat] n^\mu_{\Sigma_{ \tau_1}} \dvol_{\Sigma_{ \tau_1}} + \int_{t^\ast \geq  2\tau_1} \int_{\Sigma_{t^\ast }\cap\mathcal{R}_A} J^T_\mu [\Psi_\sharp] n_{\Sigma_{t^\ast }}^\mu  \dvol_{\Sigma_{t^\ast }} \d t^\ast.\label{eq:last}
\end{align}
Here, we have also used \eqref{eq:Psiflat}, \eqref{eq:hardyinequ} and the fact that $\int_{\mathcal{H}^+_A}|T{\Psi}_\flat|^2 \lesssim \int_{\mathcal{H}^+_A}|T{\Psi}|^2 $.
Moreover, the estimate $\int_{\mathcal{H}^+_A}|T{\Psi}|^2 \lesssim \int_{\Sigma_{ \tau_1}\cap\mathcal{R}_A} J^T_\mu[\Psi] n^\mu_{\Sigma_{ \tau_1}} \dvol_{\Sigma_{ \tau_1}}$ follows from \eqref{eq:energyestimatehtosigma}.

Finally, we are left with the term $\int_{t^\ast \geq  2\tau_1} \int_{\Sigma_{t^\ast }\cap\mathcal{R}_A} J^T_\mu [\Psi_\sharp] n_{\Sigma_{t^\ast }}^\mu  \dvol_{\Sigma_{t^\ast }} \d t^\ast$. We will show that this term decays at a superpolynomial rate. 
First, introduce the notation $\chi_\sharp := 1 - \chi_{2\omega_0}$ and set $\check{\chi_{2\omega_0}}:= \mathcal{F}^{-1}_T(\chi_{2\omega_0})$, $\check{\chi_\sharp}:= \mathcal{F}^{-1}_T (\chi_\sharp)$, which are well-defined in the distributional sense. Then,	\begin{align}\Psi_\sharp =\frac{1}{\sqrt{2\pi}} \check{\chi_\sharp} \ast \Psi =\frac{1}{\sqrt{2\pi}} \check{\chi_\sharp} \ast (\Psi-\psi_\flat)  \end{align}
since $\check{\chi}_\sharp \ast \psi_\flat =0$ in view of their disjoint Fourier support. In particular, for $t^\ast \geq  \tau_1$ we have
\begin{align}
	\Psi_\sharp  = \frac{1}{\sqrt{2\pi}}\check{\chi_\sharp} \ast (\Psi-\psi_\flat)  = \frac{1}{\sqrt{2\pi}}(\sqrt{2\pi} \delta - \check{\chi_{2\omega_0}}) \ast (\Psi-\psi_\flat) =  -\frac{1}{\sqrt{2\pi}} \check{\chi_{2\omega_0}} \ast (\Psi-\psi_\flat)
\end{align}
as $\delta\ast (\Psi-\psi_\flat)  = \Psi-\psi_\flat =0$ for $t^\ast \geq  \tau_1$. To make notation easier we define $\phi:= -\frac{1}{\sqrt{2\pi}} (\Psi-\psi_\flat)$ which is only supported for $t^\ast \leq  \tau_1$ and satisfies $\Psi_\sharp = \check{\chi_{2\omega_0}} \ast \phi$. 
Now, as a result of the $T$ invariance of $\dvol_{\Sigma_{t^\ast}}$ and  $J_\mu^T[\cdot]n^\mu_{\Sigma_{t^\ast}}$, as well as \eqref{eq:hardyinequ}, we have that
 \begin{align*}
	& \int_{t^\ast \geq 2  \tau_1} \int_{\Sigma_{t^\ast }\cap\mathcal{R}_A} J^T_\mu [\Psi_\sharp] n_{\Sigma_{t^\ast }}^\mu  \dvol_{\Sigma_{t^\ast }} \d t^\ast
	\\ & \lesssim  	 \int_{t^\ast \geq 2  \tau_1}\int_{(r_+,\infty)\times \mathbb{S}^2}\left(\frac{1}{r^2} |\partial_{t^\ast} \Psi_\sharp|^2 + \frac{\Delta}{r^2}|\partial_r \Psi_\sharp|^2 + |\slashed\nabla\Psi_\sharp|^2 \right) r^2 \d\sigma_{\mathbb{S}^2}\d r\d t^\ast 
	\\&	 \leq  	 \int_{t^\ast \geq 2  \tau_1} \int_{(r_+,\infty)\times \mathbb{S}^2} \Bigg[ r^{-2} \left| \int_{-\infty}^{t( \tau_1,r)} \check{\chi_{2\omega_0}}(t(t^\ast,r) - s )  (\partial_{t^\ast} \phi) (s) \d s \right|^2
	\\ & \qquad + \frac{\Delta}{r^2} \left| \int_{-\infty}^{t( \tau_1,r)} \check{\chi_{2\omega_0}}(t(t^\ast,r) - s)  (\partial_{r} \phi) (s) \d s \right|^2 + \left| \int_{-\infty}^{t( \tau_1,r)} |\check{\chi_{2\omega_0}}(t(t^\ast,r) - s)|  |\slashed\nabla \phi| (s)  \d s \right|^2 \Bigg] r^2\d\sigma_{\mathbb{S}^2}\d r\d t^\ast
	\\ &	 \leq \int_{-\infty}^{\infty} | \check{\chi_{2\omega_0}}(s) | \d s	 \int_{t^\ast \geq 2  \tau_1} \int_{(r_+,\infty)\times \mathbb{S}^2} \left[ \int_{-\infty}^{ \tau_1}| \check{\chi_{2\omega_0}}(t^\ast- s^\ast) | r^{-2}|\partial_{t^\ast} \phi|^2 (s^\ast) \d s^\ast \right.
	 \\ & \qquad \left. +  \int_{-\infty}^{ \tau_1} | \check{\chi_{2\omega_0}}(t^\ast - s^\ast)| \frac{\Delta}{r^2} |\partial_{r} \phi|^2 (s^\ast) \d s^\ast   +  \int_{-\infty}^{ \tau_1} |\check{\chi_{2\omega_0}}(t^\ast - s^\ast) | |\slashed\nabla \phi|^2 (s^\ast) \d s^\ast \right] r^2 \d\sigma_{\mathbb{S}^2}\d r\d t^\ast\\&
	 \lesssim \int_{t^\ast \geq 2  \tau_1} \int_{-\infty}^{ \tau_1} |\check{\chi_{2\omega_0}}(t^\ast - s^\ast) | \left( \int_{(r_+,\infty)\times \mathbb{S}^2} \left[ r^{-2}|\partial_{t^\ast}\phi|^2(s^\ast) + \frac{\Delta}{r^2}|\partial_r \phi|^2(s^\ast) + |\slashed\nabla \phi|^2(s^\ast)\right] r^2 \d\sigma_{\mathbb{S}^2} \d r\right) \d s^\ast \d t^\ast\\& \lesssim \int_{\Sigma_0\cap \mathcal{R}_A} J^T_\mu[\phi] n^\mu_{\Sigma_0} \dvol_{\Sigma_0} \int_{t^\ast \geq 2  \tau_1} \int_{-\infty}^{ \tau_1} |\check{\chi_{2\omega_0}}(t^\ast - s^\ast) |  \d s^\ast \d t^\ast\\ &
	 \lesssim_q \int_{\Sigma_0\cap \mathcal{R}_A} J^T_\mu[\psi_\flat] n^\mu_{\Sigma_0} \dvol_{\Sigma_0} \int_{t^\ast \geq 2  \tau_1} \int_{-\infty}^{ \tau_1} \frac{1}{|t^\ast -s^\ast|^{q+2}}\d s^\ast \d t^\ast\\ &
	 \lesssim_q \frac{\int_{\Sigma_0\cap \mathcal{R}_A} J^T_\mu[\psi_\flat] n^\mu_{\Sigma_0} \dvol_{\Sigma_0}}{ 1+ \tau_1^q}.
\end{align*} 
Here, we have used the boundedness of the $T$-energy (cf.\ \eqref{eq:boundednenergyflux}), i.e. \begin{align}\int_{\Sigma_{t^{\ast}}\cap \mathcal{R}_A} J^T_\mu[\phi] n^\mu_{\Sigma_{t^{\ast}}} \dvol_{\Sigma_{t^{\ast}}} \leq \int_{\Sigma_0\cap \mathcal{R}_A} J^T_\mu[\phi] n^\mu_{\Sigma_0} \dvol_{\Sigma_0}\lesssim \int_{\Sigma_0\cap \mathcal{R}_A} J^T_\mu[\psi_\flat] n^\mu_{\Sigma_0} \dvol_{\Sigma_0}. \end{align}
Finally, we have also used that the Schwartz function $\check{\chi_{2\omega_0}}$ decays superpolynomially at any power $q>1$. This concludes the proof in view of \eqref{eq:last}.
\end{proof}
In order to remove the degeneracy of the $T$-energy at the event horizon, we will use the by now standard red-shift vector field \cite{lecture_notes}. As usual, the red-shift vector field $N$ is a future-directed $T$ invariant timelike vector field which has a positive bulk term $K^N\geq 0$ near the event horizon. In a compact $r$ region bounded away from the event horizon $\mathcal{H}_A^+$, the bulk term $K^N$ of $N$ is sign-indefinite but this will be absorbed in the spacetime integral of the $T$ current in \cref{prop:microiled}. Also, note that $N=T$ for large enough $r$.
In the negative mass AdS setting, we refer to \cite[Section~4.2]{Gustav_early} for an explicit construction of the red-shift vector field $N$. Note that the red-shift vector field $N$ has the property that
\begin{align}	 
\int_{\Sigma_{t^\ast} \cap \mathcal{R}_A}	J^N_\mu[\psi_\flat] n^\mu_{\Sigma_{t^\ast}} \dvol_{\Sigma_{t^\ast} } \sim \int_{\Sigma_{t^\ast} \cap \mathcal{R}_A } e_1[\psi_\flat] r^2 \d r \sin\theta\d \theta \d \varphi 
\end{align} for $\psi_\flat$ as in \eqref{eq:fouriertransforminterior}.
\begin{prop}\label{IED} Let ${{\psi}_\flat}$ be as in \eqref{eq:fouriertransforminterior}. Then for any $\tau_2 \geq 2\tau_1 \geq 0$, we have
	\begin{align}
	\nonumber
	 \int_{\Sigma_{\tau_2} \cap \mathcal{R}_A} & J_\mu^N[{{\psi}_\flat}]  n_{\Sigma_{\tau_2}}^\mu \dvol_{\Sigma_{\tau_2}} + \int_{\mathcal{H}_A^+ \cap \{ 2\tau_1\leq t^\ast \leq \tau_2\}} \left(   |\partial_{t^\ast}{{\psi}_\flat}|^2 + |\slashed \nabla {{\psi}_\flat}|^2 + |{\psi}_\flat|^2 \right) \d{t^\ast}\d{\sigma}_{\mathbb S^2}\\ &+ \int_{2\tau_1}^{\tau_2} \int_{\Sigma_{t^\ast}\cap \mathcal{R}_A} J_\mu^N[ {{\psi}_\flat}] n^\mu_{\Sigma_{t^\ast}}  \dvol_{\Sigma_{t^\ast}} \d t^\ast  \lesssim_q	 \int_{\Sigma_{\tau_1} \cap \mathcal{R}_A} J_\mu^N[{{\psi}_\flat}]  n^\mu + \frac{\int_{\Sigma_0 \cap \mathcal{R}_A} J^T_\mu[\psi_\flat] n^\mu_{\Sigma_0} \dvol_{\Sigma_0}}{1+\tau_1^q}
	\end{align}
	and in particular,
	\begin{align}
			\nonumber
			\int_{\Sigma_{\tau_2} \cap \mathcal{R}_A} J_\mu^N[{{\psi}_\flat}]  n_{\Sigma_{\tau_2}}^\mu  \dvol_{\Sigma_{\tau_2}} + \int_{2\tau_1}^{\tau_2} \int_{\Sigma_{t^\ast}\cap \mathcal{R}_A} & J_\mu^N[ {{\psi}_\flat}] n^\mu_{\Sigma_{t^\ast}}  \dvol_{\Sigma_{t^\ast}} \d t^\ast 
			\\ \nonumber & \lesssim_q	 \int_{\Sigma_{\tau_1}\cap \mathcal{R}_A} J_\mu^N[{{\psi}_\flat}]  n^\mu + \frac{\int_{\Sigma_0 \cap \mathcal{R}_A} J^N_\mu[\psi_\flat] n^\mu_{\Sigma_0} \dvol_{\Sigma_0}}{1+\tau_1^q}\\ & \lesssim  \int_{\Sigma_{\tau_1}\cap \mathcal{R}_A} J_\mu^N[{{\psi}_\flat}]  n^\mu + \frac{E^A_1[\psi_\flat](0)}{1+\tau_1^q}.
	\end{align}

\begin{proof}
	We apply the energy identity (the spacetime integral of \eqref{eq:defjx}) with the red-shift vector field $N$ for ${{\psi}_\flat}$ in the region $\mathcal{R}_A\cap \{ 2\tau_1\leq t^\ast \leq \tau_2 \}$, where $2{\tau_1}\leq  \tau_2$. 
	After taking care of the negative lower order term via a Hardy inequality and absorbing the sign-indefinite bulk of $N$ away from the horizon (in the region $\{r\geq r_0\}$ for some $r_0 > r_+$) in the spacetime integral of $J^T$ on the right hand side (see \cite[Section~4]{Gustav_early} for further details), we arrive at
	\begin{align}\nonumber
		\int_{\Sigma_{\tau_2}\cap \mathcal{R}_A}& J_\mu^N[{{\psi}_\flat}]  n_{\Sigma_{\tau_2}}^\mu \dvol_{\Sigma_{\tau_2}}+ \int_{\mathcal{H}^+_A\cap \{2\tau_1\leq t^\ast \leq \tau_2\}} J^N_\mu[{{\psi}_\flat}]n^\mu_{\mathcal{H}} \dvol_{\mathcal{H}}+\int_{2\tau_1}^{\tau_2} \int_{\Sigma_{t^\ast}\cap \mathcal{R}_A} J_\mu^N[ {{\psi}_\flat}] n^\mu_{\Sigma_{t^\ast}}  \dvol_{\Sigma_{t^\ast}} \d t^\ast\\& \lesssim \int_{2\tau_1}^{\tau_2} \int_{\Sigma_{t^\ast}\cap \mathcal{R}_A \cap \{ r\geq r_0 \} } J_\mu^T[ {{\psi}_\flat}] n^\mu_{\Sigma_{t^\ast}}  \dvol_{\Sigma_{t^\ast}} \d t^\ast +	\int_{\Sigma_{2\tau_1}\cap \mathcal{R}_A} J_\mu^N[{{\psi}_\flat}]  n_{\Sigma_{\tau_1}}^\mu \dvol_{\Sigma_{\tau_1}} .\label{eq:nlaw}
	\end{align} 
	First, note that the integrated energy term $\int_{2\tau_1}^{\tau_2} \int_{\Sigma_{t^\ast}\cap \mathcal{R}_A \cap \{ r\geq r_0 \} } J_\mu^T[ {{\psi}_\flat}] n^\mu_{\Sigma_{t^\ast}}  \dvol_{\Sigma_{t^\ast}} \d t^\ast $ on the right-hand side of \eqref{eq:nlaw} can be controlled by the left-hand side of \cref{prop:microiled}. Then, remark that the integral along the horizon $\int_{\mathcal{H}^+_A\cap \{2\tau_1\leq t^\ast \leq \tau_2\}}J^N_\mu[{{\psi}_\flat}]n^\mu_{\mathcal{H}} \dvol_{\mathcal{H}}$ is sign-indefinite due to the (possible) negative mass. However, this can be absorbed in the bulk term using an $\epsilon$ of the integrated bulk term of the red-shift vector field $N$ and some of the bulk term of the integrated energy estimate in \cref{prop:microiled}, cf.\ \cite[Equation~(70)]{Gustav_early}. Finally, using the integrated energy estimate from \cref{prop:microiled} again, we conclude
\begin{align}
		 \nonumber
		 \int_{\Sigma_{\tau_2} \cap \mathcal{R}_A} & J_\mu^N[{{\psi}_\flat}]  n_{\Sigma_{\tau_2}}^\mu \dvol_{\Sigma_{\tau_2}} + \int_{\mathcal{H}_A^+ \cap \{ 2\tau_1\leq t^\ast \leq \tau_2\}} \left(   |\partial_{t^\ast}{{\psi}_\flat}|^2 + |\slashed \nabla {{\psi}_\flat}|^2 + |{\psi}_\flat|^2 \right) \d{t^\ast}\d{\sigma}_{\mathbb S^2}\\ &+ \int_{2\tau_1}^{\tau_2} \int_{\Sigma_{t^\ast}\cap \mathcal{R}_A} J_\mu^N[ {{\psi}_\flat}] n^\mu_{\Sigma_{t^\ast}}  \dvol_{\Sigma_{t^\ast}} \d t^\ast  \lesssim_q	 \int_{\Sigma_{\tau_1}\cap \mathcal{R}_A} J_\mu^N[{{\psi}_\flat}]  n_{\Sigma_{\tau_1}}^\mu + \frac{\int_{\Sigma_0 \cap \mathcal{R}_A} J^T_\mu[\psi_\flat] n^\mu_{\Sigma_0} \dvol_{\Sigma_0}}{1+\tau_1^q}.\label{eq:dec}
\end{align}
\end{proof}
\end{prop}
Now we obtain
\begin{prop}\label{prop:psiflat}Let $\psi_\flat$ be defined as in \eqref{eq:fouriertransforminterior}. Then, for any $q>1$ and $\tau \geq 0$ we have
	\begin{align}\label{eq:superpolynomialdecay}
		 \int_{\Sigma_{\tau} \cap \mathcal{R}_A} J_\mu^N[{{\psi}_\flat}]  n^\mu_{\Sigma_\tau} \lesssim_q \frac{1}{1+\tau^q}  \int_{\Sigma_{0}\cap \mathcal{R}_A} J_\mu^N[{{\psi}_\flat}]  n_{\Sigma_0}^\mu \dvol_{\Sigma_0}  \lesssim_q \frac{1}{1+\tau^q}E^A_1[\psi_\flat](0)
	\end{align}
	and 
	\begin{align}\label{eq:decayalongeventhorizon}
\int_{\mathcal{H}(\tau,+\infty)} |\partial_{t^\ast} {{\psi}_\flat}|^2 + ( |\slashed\nabla {{\psi}_\flat}|^2 + |{\psi}_\flat|^2 ) \lesssim_q \frac{1}{1+\tau^q}  \int_{\Sigma_{0}\cap \mathcal{R}_A} J_\mu^N[{{\psi}_\flat}]  n_{\Sigma_0}^\mu \dvol_{\Sigma_0} \lesssim_q \frac{1}{1+\tau^q}E^A_1[\psi_\flat](0) .
	\end{align}
\end{prop}
\begin{proof}In view of \cref{IED} it suffices to prove \eqref{eq:superpolynomialdecay}.
	Upon setting \begin{align*}f(s) := \int_{\Sigma_{s} \cap \mathcal{R}_A}  J_\mu^N[{{\psi}_\flat}]  n_{\Sigma_{s}}^\mu \dvol_{\Sigma_{s}}, \end{align*}
	we have from \cref{IED} that
	\begin{align*}
		f(t_2) + \int_{2t_1}^{t_2} f(s)  \d s \lesssim_q f(t_1)  + \frac{f(0)}{1+t_1^q}
	\end{align*}
	for any $t_2\geq 2 t_1 \geq 0$. 
	The claim follows now from  \cref{lem:expdecay} below.
\end{proof}
\begin{lemma}\label{lem:expdecay}
Let $f\colon [0,\infty)\to [0,\infty)$ be a continuous function satisfying
\begin{align}\label{eq:lem}
	f(t_2) + \int_{2t_1}^{t_2} f(s) \d s \leq \alpha(q) \left( f(t_1) + \frac{f(0)}{1+t_1^q}\right)
\end{align}
for any $q>1$, $0\leq 2t_1 \leq t_2$ and some $\alpha(q) >0$ only depending on $q$. Then, for all $q>1$, there exists a constant $C(\alpha(q),q)>0$ only depending on $\alpha$ and $q$ such that
\begin{align}
	f(t) \leq\frac{ C(\alpha(q),q)   }{1+t^q} f(0)
\end{align}
for all $t\geq 0$. 
\begin{proof} Fix $q>1$. First, note that from \eqref{eq:lem} we have for any $t_2\geq 2 t_1>0$
	\begin{align*}
		f(t_2) \leq \alpha(q) \left(f(t_1) + \frac{f(0)}{1+t_1^q}\right).
	\end{align*}
	Without loss of generality, let $t>10$ be arbitrary. Then, take a dyadic sequence $\tau_{k+1} = 2\tau_k$, where $\tau_0 =1$. Now, there exists a $n\in \mathbb{N}_0$ such that $t\in [\tau_{n+3} , \tau_{n+4}]$. Then, again from \eqref{eq:lem} we have
	\begin{align*}
		\int_{\tau_{n+1}}^{\tau_{n+2}} f(s) \d s \leq \alpha(q) \left( f(\tau_n) + \frac{f(0)}{1+\tau_n^q}\right) 
	\end{align*}
	from which we conclude that there exists a $\xi \in [\tau_{n+1}, \tau_{n+2}]$ such that
	\begin{align*}
		f(\xi)  \leq \alpha(q) \left( \frac{f(\tau_n)}{\tau_{n+1}} + \frac{f(0)}{1+\tau_n^{q+1}}\right).
	\end{align*}
	Hence, since $2\xi\leq \tau_{n+3} \leq t \leq \tau_{n+4}$, 
	\begin{align}\label{eq:lem2}
		f(t) & \leq \alpha(q) \left( f(\xi) + \frac{f(0)}{1+\tau_{n+1}^q}\right) \leq \alpha(q) \left(\alpha(q) \left(\frac{f(\tau_n)}{\tau_{n+1}} + \frac{f(0)}{1+\tau_n^{q+1}} \right) +\frac{f(0)}{1+\tau_{n+1}^q }\right)  .
	\end{align}
	Now, note that $\tau_n \sim t$ and hence, $f(t) \leq C(1,\alpha(q)) \frac{1}{1+t}$. This improved decay can now be fed into \eqref{eq:lem2} to obtain a decay of the form $f(t) \leq C(2,\alpha(q)) \frac{1}{1+t^2}$. This procedure can be iterated until one obtains \begin{align}
		f(t) \leq  \frac{C(q,\alpha(q))}{1+t^q}f(0).
	\end{align}
\end{proof}
\end{lemma}
\subsection{Interior estimates}\label{sec:interiorlow}
 Having obtained the superpolynomial decay for $\psi_\flat$ in the exterior and in particular on the event horizon, we will now use this to show uniform boundedness in the black hole interior. 
We will first propagate the superpolynomial decay on the horizon established in \cref{prop:psiflat} further into the interior. To do so we will make use of the twisted red-shift.
\subsubsection{Red-shift region}
\label{subsubsec:twistedredshift}
	With the help of the constructed twisted red-shift current in \cref{prop:redshift}, we obtain
	\begin{prop}\label{prop:reddecay}
		Let $r_0 \in [\red,r_+)$. Let ${{\psi}_\flat}$ defined as in \eqref{eq:fouriertransforminterior} and recall that from \cref{prop:psiflat} we have
			\begin{align}
			\int_{\mathcal{H}(v_1,v_2)} \tilde J_\mu^N[{{\psi}_\flat}] n^\mu_{\mathcal{H}^+} \dvol_{\Hp^+} \lesssim_q \frac{1}{1+v_1^q}E^A_1[\psi_\flat](0) 
			\end{align}
			for $1\leq v_1 \leq v_2$.
			 Then, 
		 \begin{align}\label{eq:reddecay1}
			\int_{ \cv{v_1} (r_0,r_+) } \tilde J_\mu ^N[{{\psi}_\flat}] n_{\cv{v}}^\mu \dvol_{\cv{v}} \sim \int_{-\infty}^{u_{r_0}(v_1)} \int_{\mathbb{S}^2}\frac{1}{\Omega^2} |\tilde{\nabla}_u {{\psi}_\flat}|^2 + \Omega^2( |\slashed\nabla {{\psi}_\flat}|^2  + \mathcal{V} |{{\psi}_\flat}|^2)  \d\sigma_{\mathbb S^2} \d u &\lesssim_q \frac{1}{1+v_1^q}E_1[\psi_\flat](0),\\ 
			\int_{\Sigma_{r_0}(v_1,v_2)} \tilde J_\mu^N[{{\psi}_\flat}] n^\mu_{\Sigma_r} \dvol_{\Sigma_r} \sim \int_{v_1}^{v_2} \int_{\mathbb S^2} \frac{1}{\sqrt{\Omega^2}} |\tilde \nabla_u {{\psi}_\flat}|^2   + \sqrt{\Omega^2} \Big( |\tilde \nabla_v {{\psi}_\flat}|^2 + |\slashed \nabla {{\psi}_\flat}|^2 + &\mathcal{V} {|{\psi}_\flat|}^2 \Big)  \d v \d\sigma_{\mathbb{S}^2}\nonumber  \\   &\lesssim_q \frac{E_1[\psi_\flat](0)}{1+v_1^q} \label{eq:reddecay2}\end{align}
		for any $1\leq v_1 \leq v_2$.
		\begin{proof}
			 From \cref{prop:redshiftprelim}, estimate \eqref{eq:decayalongeventhorizon} in \cref{prop:psiflat} and upon defining  
			\begin{align}
			\tilde 	E(v) := \int_{\underline{C}_v(r_0,r_+)}\tilde J_\mu^N[{{\psi}_\flat}] n^\mu_{\underline{C}_v} \dvol_{\underline{C}_v},
			\end{align}
			we obtain
			\begin{align}\label{eq:480}
				\tilde E(v_2) +  \int_{v_1}^{v_2}\tilde E(v) \d{v} \lesssim_q \tilde E(v_1) +\frac{E^A_1[\psi_\flat](0)}{1+v_1^q},
			\end{align}
			for any $1\leq v_1 \leq v_2$. This implies
			\begin{align}
			\tilde	E(v) \lesssim_q (\tilde E(v=1) + E^A_1[\psi_\flat](0) ) \frac{1}{1+v^q}
			\end{align}
			for any $v\geq 1$.
			This follows from an argument very similar to \cref{lem:expdecay}.
			Note that we have by general theory \cite{lecture_notes} that $\tilde E(v=1) \lesssim E_1[\psi_\flat](0)$.  Thus,
			\begin{align}
			\tilde 	E(v) \lesssim_q E_1[\psi_\flat](0) \frac{1}{1+v^q}
					\end{align} 
			for $v\geq 1$ which proves \eqref{eq:reddecay1}. The estimate
			\eqref{eq:reddecay2} now follows from \eqref{eq:reddecay1} and \cref{prop:redshiftprelim}.  
		\end{proof}
	\end{prop}
		
	\subsubsection{No-shift region}
	\label{subsubsec:noshiftregion}
Now, we will propagate the decay towards $i^+$ further into the black hole for $r\in [\red,\rblue]$, where $\rblue >r_-$ is determined in the proof of \cref{prop:decayinblueshift}.
		\begin{prop}Let ${{\psi}_\flat}$ defined as in \eqref{eq:fouriertransforminterior}.
			For any $r_0 \in [\rblue, \red]$, $q>1$ and any $v_\ast \geq1$ we have
		\begin{align}
		\label{eq:decayvv+1}
&\int_{\Sigma_{r_0}(v_\ast, 2v_\ast)} \tilde J_\mu^X[{{\psi}_\flat}] n^\mu_{\Sigma_{r}} \dvol_{\Sigma_{r}} \lesssim_q \frac{E_1[\psi_\flat](0)}{1+v^q_\ast}.\end{align}
Moreover, for any $1<p<q$ we also have
\begin{align}
&\int_{\Sigma_{r_0}(v_\ast, +\infty)} (\vp + \up )\tilde J_\mu^X[{{\psi}_\flat}] n^\mu_{\Sigma_{r}} \dvol_{\Sigma_{r}} \lesssim_{q,p} {E_1[\psi_\flat](0)}.
\label{eq:decayv2}
\end{align}
		\end{prop}
		\begin{proof}
						Applying \cref{prop:noshiftprop} with $\phi = \psi_\flat$ we have \eqref{eq:noshiftestimate} for $\psi_\flat$. To estimate the right-hand side of \eqref{eq:noshiftestimate} we use \cref{prop:reddecay} and the fact that the difference $v_\ast  - v_{\red}(u_{r_0}(v_\ast ))) = const.$ to obtain
							\begin{align}\label{eq:ev2}
							 \int_{\Sigma_{\red}(v_{\red }(u_{r_0} (v_\ast )), 2v_\ast )} \tilde J_\mu^X[{{\psi}_\flat}] n^\mu_{\Sigma_{ r}}\dvol_{\Sigma_{ r}} \lesssim_q \frac{E_1[\psi_\flat](0)}{1+v^q_\ast}
					\end{align}
					from which \eqref{eq:decayvv+1} follows.
							Finally, \eqref{eq:decayv2} is a consequence of the fact that $ \langle v \rangle^p \sim \langle u \rangle^p$ (using $\rblue \leq r \leq \red$) and the following well-known lemma.	\end{proof}
							\begin{lemma}
								Let $f\colon [1,\infty) \to \mathbb{R}_{\geq 0}$ be continuous and assume that there exists a $q\in \mathbb R$, $q>1$ such that $\int_x^{2x} f(s) \d s \leq \frac{D}{x^q}$ for all $x\geq 1$ and some constant $D>0$. Let $1<p<q$ be fixed. Then,
								$\int_1^\infty s^p f(s)  \d s < C(q,p) D$ for a constant $C(p,q)>0$ only depending on $p$ and $q$. 
								\begin{proof}
									Set $x_i := 2^i $. Then,
									$\int_1^\infty s^p f(s)  \d s = \sum_{i=0}^\infty \int_{x_i}^{x_{i+1}} s^p f(s) \d s\leq 2^p D \sum_{i=0}^\infty 2^{ip-iq}<C(q,p) D.$			\end{proof}
							\end{lemma}
							
						\begin{rmk}	\label{rmk:defp}From now on we will consider $p$ and $q$ as \textbf{fixed} and constants appearing in $\lesssim$, $\gtrsim$ and $\sim$ can additionally depend on $1<p<q$.\end{rmk}
	
By doing the analogous analysis in the neighborhood of the left component of $i^+$ we obtain
	\begin{prop}\label{prop:boundonblue}Let ${{\psi}_\flat}$ defined as in \eqref{eq:fouriertransforminterior}. Then, for any $r_0\in[\rblue,r_+)$ we have
		\begin{align} \int_{\Sigma_{r_0}} (\vp + \up ) \left( |\tilde{\nabla}_u {{\psi}_\flat}|^2 + |\tilde{\nabla}_v {{\psi}_\flat}|^2 + |\slashed\nabla {{\psi}_\flat}|^2 + |{{\psi}_\flat}|^2 \right) \dvol_{\Sigma_{r}} \lesssim {E_1[\psi_\flat](0)}	.	\end{align}
	\end{prop}
	Commuting with angular momentum operators $(\mathcal{W}_i)_{1\leq i \leq 3}$, an application of the Sobolev embedding $H^2(\mathbb{S}^2) \hookrightarrow L^\infty(\mathbb S^2)$ and using the fact that $p>1$, we also conclude
	\begin{prop}\label{prop:boundednessred}Let ${{\psi}_\flat}$ defined as in \eqref{eq:fouriertransforminterior}. Then,
	\begin{align}
		\sup_{\mathcal{B}\cap \{ \rblue \leq r < r_+\}} |{{\psi}_\flat}|^2\lesssim E_1[\psi_\flat] (0) +\sum_{i,j=1}^3 E_1[\mathcal{W}_i \mathcal{W}_j \psi_\flat] (0).
	\end{align}
	\end{prop}
	Finally, we will use the decay towards $i^+$ to show uniform boundedness in the interior and continuity all the way up to and including the Cauchy horizon for $\psi_\flat$.	
		\subsubsection{Blue-shift region}\label{sec:blueshift}
	We will now introduce the twisting function and vector field which we will use in the blue-shift region. Recall that we look for a twisting function $f$ which satisfies $\V\gtrsim 1$, where  \begin{align}\V = - \left(\frac{\Box_g f}{f}+ \frac{\alpha}{l^2}\right).\end{align}
To do so, we set $f:= e^{\beta_{\mathrm{blue} }r}$ and obtain
		\begin{align}\label{eq:f1}
		\mathcal V = -\frac{\Box_g f}{f}  - \frac{\alpha}{l^2}=  \beta_\mathrm{blue}^2  {\Omega^2 } +  \beta_\mathrm{blue} \partial_r (\Omega^2) + \frac{2}{r}\beta_\mathrm{blue}  \Omega^2  - \frac{\alpha}{l^2}.
		\end{align}
		Note that for $\rblue > r_-$ close enough to $r_-$, we have 
		\begin{align}
		 \partial_r \Omega^2 \geq c_{\mathrm{blue}}
		\end{align}
		for all $\rblue \geq r \geq r_-$ and some constant $c_{\mathrm{blue}}>0$ only depending on the black hole parameters. Thus, we obtain $\V \gtrsim 1$ uniformly in the blue-shift region $\rblue \geq r \geq r_-$ by choosing $\beta_\mathrm{blue}>0$ large enough and $\rblue$ close enough to $r_-$. 
		In the blue-shift region we define the vector field 
		\begin{align}\label{eq:defnsn}
			S_N:=r^N(\langle u\rangle^p \partial_u +\langle v\rangle^p \partial_v)
		\end{align}
		for some potentially large $N>0$ and $p>1$ as in \cref{rmk:defp}. We will show in the following that $\sup_{\theta,\varphi}|\psi_\flat (u_0,v_0,\theta,\varphi)|$ is uniformly bounded from initial data $D[\psi_\flat]$ independently of $(u_0,v_0) \in J^+(\Sigma_{\rblue})\cap \mathcal{B}$.
		To do so, we will apply the energy identity (spacetime integral of \eqref{eq:twiteddefjx}) in the region \begin{align}
		\mathcal{R}_f = \mathcal{R}_f(u_0,v_0) =  J^+(\Sigma_{\rblue}) \cap J^-(v_0,u_0) = J^+(\Sigma_{\rblue}) \cap \{ u \leq u_0\} \cap \{v \leq v_0\} 
		\end{align}
		which we depict in \cref{fig:blueshift}.
		\begin{figure}[H]
			\centering
		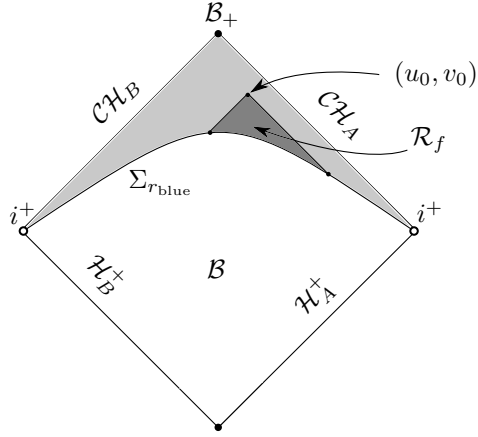
			\caption{Illustration of the region $\mathcal{R}_f$ as the darker shaded region in the Penrose diagram of the interior $\mathcal{B}$. The lighter shaded region is the blue-shift region. }
			\label{fig:blueshift}
		\end{figure}
		This leads to
		\begin{align} \nonumber 
			\int_{\mathcal{C}_{u_0}(v_{\rblue}(u_0), v_0)}& \tilde{J}^{S_N}_\mu[\psi_\flat] n_{\mathcal{C}_{u_0}}^\mu \dvol_{\mathcal{C}_{u_0}} + 	\int_{\underline{\mathcal{C}}_{v_0}(u_{\rblue}(v_0), u_0)}  \tilde{J}^{S_N}_\mu[\psi_\flat] n_{\underline{\mathcal{C}}_{v_0}}^\mu \dvol_{\underline{\mathcal{C}}_{v_0}} + \int_{\mathcal{R}_f} \tilde K^{S_N}[\psi_\flat] \dvol\\ & = \int_{\Sigma_{\rblue}\cap J^-(v_0,u_0)} \tilde J^{S_N}_\mu[\psi_\flat] n^\mu_{\Sigma_{\rblue}} \dvol_{\Sigma_{\rblue}}, \label{eq:energyidentity}
		\end{align}
		where $\psi_\flat$ is defined in \eqref{eq:fouriertransforminterior}.
		In the following we will show, that after choosing $N>0$ large enough and an appropriate integration by parts to control error terms, we can control the flux terms by initial data. This gives
		\begin{prop}\label{prop:decayinblueshift} Let ${{\psi}_\flat}$ defined as in \eqref{eq:fouriertransforminterior}. Then,
				\begin{align} \nonumber 
				\int_{\mathcal{C}_{u_0}(v_{\rblue}(u_0), v_0)} \tilde{J}^{S_N}_\mu[\psi_\flat]  n_{\mathcal{C}_{u_0}}^\mu \dvol_{\mathcal{C}_{u_0}} + &	\int_{\underline{\mathcal{C}}_{v_0}(u_{\rblue}(v_0), u_0)}  \tilde{J}^{S_N}_\mu[\psi_\flat]  n_{\underline{\mathcal{C}}_{v_0}}^\mu \dvol_{\underline{\mathcal{C}}_{v_0}} \\ & \lesssim \int_{\Sigma_{\rblue}\cap J^-(v_0,u_0)} \tilde J^{S_N}_\mu[\psi_\flat]  n^\mu_{\Sigma_{\rblue}} \dvol_{\Sigma_{\rblue}} \lesssim E_1[\psi_\flat](0)
				 \label{eq:energyidentityestim}
				\end{align}
				and 
				\begin{align}& \nonumber
	\int_{\mathcal{C}_{u_0}(v_{\rblue}(u_0), v_0)} \big( \langle v \rangle^p  |\partial_v \psi_\flat|^2 + (|\slashed\nabla \psi_\flat|^2 + |\psi_\flat|^2) \Omega^2 \big) \d{v} \d\sigma_{\mathbb{S}^2} \\ \nonumber &+ 	\int_{\underline{\mathcal{C}}_{v_0}(u_{\rblue}(v_0), u_0)}   \left( \langle u \rangle^p |\partial_u \psi_\flat|^2 + (|\slashed\nabla \psi_\flat|^2 + |\psi_\flat|^2) \Omega^2 \right) \d{v} \d\sigma_{\mathbb{S}^2} 
	 \\ & \lesssim \int_{\Sigma_{\rblue}\cap J^-(v_0,u_0)} \tilde J^{S_N}_\mu[\psi_\flat]  n^\mu_{\Sigma_{\rblue}} \dvol_{\Sigma_{\rblue}} \lesssim E_1[\psi_\flat](0) \label{eq:nontwistingenergy}
				\end{align}
				for any $(u_0,v_0) \in J^+(\Sigma_{\rblue})$. 
				Commuting with the angular momentum operators $(\mathcal{W}_i)_{1\leq i \leq 3}$ also gives
			\begin{align}& 			\int_{\mathcal{C}_{u_0}(v_{\rblue}(u_0), v_0)}  \langle v \rangle^p \big(  |\partial_v \psi_\flat|^2 +
			 \sum_{i,j} |\partial_v \mathcal{W}_i \mathcal{W}_j \psi_\flat|^2  \big) \d{v} \d\sigma_{\mathbb{S}^2}  \lesssim E_1[\psi_\flat](0) + \sum_{i,j=1}^3 E_1[\mathcal{W}_j\mathcal{W}_i \psi_\flat](0). \label{eq:nontwistingenergy_an}
			\end{align}
				\end{prop}
			\begin{proof}
			The general strategy of the proof is to apply \eqref{eq:energyidentity} and to show that \begin{align}\label{eq:generalstrag}
					\int_{\mathcal{R}_f} \tilde K^{S_N} \dvol \geq 0 + \text{boundary terms},
				\end{align} where the boundary terms are small (lower orders in $\Omega$) and by choosing $\rblue$ closer to $r_-$,  can be absorbed in the positive flux terms  on the left hand side of \eqref{eq:energyidentity}. In the first part, we compute the flux terms for our vector field $S^N$ defined in \eqref{eq:defnsn}. Then, in the second part, we will estimate the bulk term and indeed show \eqref{eq:generalstrag}. From this we will then deduce \eqref{eq:energyidentityestim}.
				\paragraph{\textbf{Part I: Flux terms of \texorpdfstring{$\boldsymbol{S_N}$}{Sq}}}
		We obtain three flux terms from \eqref{eq:energyidentity}. The future flux terms read (cf.~\cref{prop:appendix})
		\begin{align}\nonumber
			\int_{\mathcal{C}_{u_0}(v_{\rblue}(u_0), v_0)}& \tilde{J}^{S_N}_\mu[\psi_\flat]  n_{\mathcal{C}_{u_0}}^\mu \dvol_{\mathcal{C}_{u_0}} \\& = 	\int_{\mathcal{C}_{u_0}(v_{\rblue}(u_0), v_0)} \left(\vp |\tilde{\nabla}_v {{\psi}_\flat}|^2 + \Omega^2 \frac{\up}{4} (|\slashed\nabla {{\psi}_\flat}|^2 + \V |{{\psi}_\flat}^2 |)\right) r^{N+2} \d{v} \d\sigma_{\mathbb{S}^2} \label{eq:fluxterm1}
					\end{align}
					and
					\begin{align}\nonumber
					\int_{\underline{\mathcal{C}}_{v_0}(u_{\rblue}(v_0), u_0)}& \tilde{J}^{S_N}_\mu[\psi_\flat]  n_{\mathcal{C}_{v_0}}^\mu \dvol_{\underline{\mathcal{C}}_{v_0}}  \\ &= 	\int_{\underline{\mathcal{C}}_{v_0}(u_{\rblue}(v_0),u_0)} \left(\up |\tilde{\nabla}_u {{\psi}_\flat}|^2 + \Omega^2 \frac{\vp}{4} (|\slashed\nabla {{\psi}_\flat}|^2 + \V |{{\psi}_\flat}|^2 )\right) r^{N+2} \d{u} \d\sigma_{\mathbb{S}^2}.
					\end{align}
The past flux term on the spacelike hypersurface $\Sigma_{\rblue}$
is uniformly bounded by initial data from \cref{prop:boundonblue}:
\begin{align} \int_{\Sigma_{\rblue}\cap J^-(v_0,u_0)} \tilde J^{S_N}_\mu[\psi_\flat] n^\mu_{\Sigma_{\rblue}} \dvol_{\Sigma_{\rblue}}\lesssim E_1[\psi_\flat](0). \end{align} 
		\paragraph{\textbf{Part II: Bulk term of \texorpdfstring{$\boldsymbol{S_N}$}{Sq}}}  
		We will now estimate the bulk term \begin{align*}
		\int_{\mathcal{R}_f} \tilde K^{S_N} \dvol\end{align*} appearing in the energy identity \eqref{eq:energyidentity}. The terms appearing in $\tilde K^{S_N}$ can be read off in \eqref{eq:bulk} with $S_N^u = X^u = r^N \langle u \rangle^p$ and $S_N^v = X^v = r^N \langle v \rangle^p$. To estimate all terms, we will also integrate by parts and substitute terms of the form $\partial_u \partial_v {{\psi}_\flat}$ using the equation $\Box_g{{\psi}_\flat} =0$. The boundary terms arising from the integration by parts will then be absorbed in the future flux terms appearing in \emph{Part I: Flux terms of $S_N$}. In the following we shall treat each terms of $\tilde K^X$ as in \eqref{eq:bulk} with  $X=S_N$ individually.
		
		\paragraph{First term of \eqref{eq:bulk}}The first term of \eqref{eq:bulk} is non-negative:
		\begin{align}
			-\frac{2}{\Omega^2} \left(\langle v\rangle^p \partial_u (r^N) |\tilde{\nabla}_v {{\psi}_\flat}|^2 + \langle u\rangle^p \partial_v (r^N) |\tilde{\nabla}_u {{\psi}_\flat}|^2 \right) = N r^{N-1} (\langle v\rangle^p |\tilde{\nabla}_v {{\psi}_\flat}|^2 + \langle u\rangle^p |\tilde{\nabla}_u {{\psi}_\flat}|^2).\label{eq:bulkmainterm}
		\end{align}
		This means that---by choosing $N>0$ large enough---we will be able to absorb sign-indefinite terms of the form $r^{N-1} \langle v \rangle^p |\tilde{\nabla}_v {{\psi}_\flat}|^2$ and $r^{N-1} \langle u \rangle^p  |\tilde{\nabla}_u {{\psi}_\flat}|^2$. This will be used in the following.
		
		Before we treat the second term appearing in \eqref{eq:bulk}, which is sign-indefinite, we look at the angular and potential term in the second line of \eqref{eq:bulk}. 
		\paragraph{Angular and potential term: Second line of \eqref{eq:bulk}}Now, we look at the term involving angular derivatives. In the region $\mathcal{R}_f$ we have
		\begin{align}
			- & \left( \frac{1}{2}(\partial_v (r^N \vp) + \partial_u (r^N \up ) ) - \frac{r^N}{4} \left( \partial_{r} \Omega^2\right) (\vp + \up )\right) \left( |\slashed\nabla {{\psi}_\flat}|^2 + \V|{{\psi}_\flat}|^2\right)\nonumber
			\\ &\gtrsim r^N (\vp + \up)  \left( |\slashed\nabla {{\psi}_\flat}|^2 + \V |{{\psi}_\flat}|^2\right) .\label{eq:bulkangularcontrol}
		\end{align}
		The terms arising when $\partial_v$ hits $\langle v \rangle^p$ and  when $\partial_u$ hits $\langle u \rangle^p$ are sign-indefinite and of the form
		\begin{align}
			- \frac{p}{2} r^N \left( \langle v \rangle^{p-2} v + \langle u \rangle^{p-2} u \right) \left( |\slashed \nabla \psi_\flat|^2 + \mathcal{V} |\psi_\flat|^2 \right).
		\end{align}
		They are absorbed in $r^N (\vp + \up)  \left( |\slashed\nabla {{\psi}_\flat}|^2 + \V |{{\psi}_\flat}|^2\right)$. Indeed, for any fixed $\epsilon = \epsilon(p)>0$, we can choose $r_{\mathrm{blue}}$ even closer to $r_-$ (depending on $\epsilon$) such that  $ |v| \langle v \rangle^{p-2} \leq   \langle v \rangle^{p-1} \leq \epsilon (\langle v \rangle^p + \langle v - 2 r_\ast \rangle^p) $ holds in $\mathcal{R}_f$ and similarly for $|u|\langle u \rangle^{p-2}$. 
	 Also recall that we have chosen the twisting function such that $\V \gtrsim 1$. 
		
\paragraph{Second, sign-indefinite term of \eqref{eq:bulk}}	Now, note that the second term in the first line of \eqref{eq:bulk} \begin{align}-2 r^{N-1} (\vp  + \up ) \operatorname{Re}\left( \overline{ \dvpsi_\flat} \dupsi_\flat \right)\label{eq:secondindterm}\end{align} is sign-indefinite, however, we can absorb it in other positive terms after integrating by parts in the region $\mathcal{R}_f$ as we will see in the following. In order to integrate by parts, it is useful to express the twisted derivatives with ordinary derivatives. The integration by parts will generate boundary terms. As mentioned above, we estimate these boundary terms with the fluxes in the energy identity. This will be done later in \eqref{eq:errortermcu} and we will not write the boundary terms explicitly in the following. 
We will also have to control (sign-indefinite) ordinary derivatives by positive terms in \eqref{eq:bulkmainterm} and \eqref{eq:bulkangularcontrol}. Note that this is possible since 	\begin{align}\label{eq:changetwistedtopartial}
\vp |\partial_v {{\psi}_\flat}|^2=  \vp |\dvpsi|^2  - \vp  \Omega^2\operatorname{Re}\left(\overline{ {{\psi}_\flat}} \partial_v {{\psi}_\flat}\right) - \frac 14 \vp  \Omega^4 |{{\psi}_\flat}|^2,\end{align}
where the right hand side of \eqref{eq:changetwistedtopartial} is controlled by \eqref{eq:bulkmainterm}, \eqref{eq:bulkangularcontrol} and potentially choosing $\rblue$ closer to $r_-$. The analogous statement holds true for $\up |\partial_u {{\psi}_\flat}|^2$. 
	
	The integrated term we have to estimate reads
	\begin{align} \label{eq:analog}	\int_{\mathcal{R}_f} & -2 r^{N-1} (\vp  + \up)  \frac{1}{f^2} \operatorname{Re}\left( \overline{ \partial_v (f{{\psi}_\flat})} \partial_u(f {{\psi}_\flat})   \right) \Omega^2 r^2  \d{u}\d{v} \d\sigma_{\mathbb{S}^2}   .
	\end{align}
	We only look at
	\begin{align} \Big| \int_{\mathcal{R}_f} &  r^{N+1} \vp  \frac{1}{f^2} \operatorname{Re} \left(\overline{ \partial_v (f{{\psi}_\flat})} \partial_u(f {{\psi}_\flat})   \right) \Omega^2  \d{u}\d{v} \d\sigma_{\mathbb{S}^2}\Big| \nonumber  
	\end{align}
	as the term in \eqref{eq:analog} involving $\langle u \rangle^p$ is estimated in an analogous manner.
	Using the explicit form of $f$ and noting that we have control over $(\vp  + \up)\Omega^4 |{{\psi}_\flat}|^2$ from \eqref{eq:bulkangularcontrol}, it suffices to estimate 
	\begin{align}\nonumber
 \Big| 	\int_{\mathcal{R}_f}   r^{N+1} \vp & \operatorname{Re}\left(  \overline{\partial_v {{\psi}_\flat}} \partial_u {{\psi}_\flat} \right)\Omega^2  \d{u}\d{v} \d\sigma_{\mathbb{S}^2}\Big|  + \Big| 	\int_{\mathcal{R}_f}   \Omega^2 
 \vp \operatorname{Re}\left( \overline{{ {\psi}_\flat} }(\partial_v {{\psi}_\flat})\right) \Omega^2 \d{u}\d{v} \d\sigma_{\mathbb{S}^2}\Big| \\ &+  \Big| 	\int_{\mathcal{R}_f}   \Omega^2 
 \vp\operatorname{Re}\left(  \overline{{{\psi}_\flat}} (\partial_u {{\psi}_\flat}) \right) \Omega^2 \d{u}\d{v} \d\sigma_{\mathbb{S}^2}\Big|   \label{eq:bulk3terms}.
	\end{align}
		Now, note that the second term  of \eqref{eq:bulk3terms} (excluding the factor $ \Omega^2$ appearing in the volume form) reads $ r^{-2}\Omega^2 
		\vp \operatorname{Re}\left( \overline{{ {\psi}_\flat} }(\partial_v {{\psi}_\flat})\right) $ and is controlled by \eqref{eq:bulkmainterm} and \eqref{eq:bulkangularcontrol} using Cauchy's inequality and by potentially choosing $\rblue$ even closer to $r_-$. Now, in both terms, the first and third term of \eqref{eq:bulk3terms}, we integrate by parts in $u$. We also use $\operatorname{Re}\left( \overline{{\psi}_\flat} \partial_u {{\psi}_\flat} \right)= \frac{1}{2} \partial_u (|{{\psi}_\flat}|^2)$. Then, it follows that---up to boundary contributions which will be dealt with below in \eqref{eq:errortermcu}---we have to control the terms
	\begin{align}\nonumber
 \Big| 	\int_{\mathcal{R}_f}  N  r^{N} \vp \operatorname{Re}&\left( \overline{{\psi}_\flat} \partial_v {{\psi}_\flat}\right) \Omega^4  \d{u}\d{v} \d\sigma_{\mathbb{S}^2}\Big|  + \Big| 	\int_{\mathcal{R}_f}    r^{N+1} \vp \operatorname{Re}\left( \overline{{\psi}_\flat} (\partial_u \partial_v {{\psi}_\flat} )\right)\Omega^4  \d{u}\d{v} \d\sigma_{\mathbb{S}^2}\Big| \\ &+   \Big| 	\int_{\mathcal{R}_f}    
 \vp  |{{\psi}_\flat}|^2 \Omega^4 \d{u}\d{v} \d\sigma_{\mathbb{S}^2}\Big|.\label{eq:bulk2terms}
\end{align}
The first and third term (excluding $\Omega^2$ as above) of \eqref{eq:bulk2terms} are controlled by \eqref{eq:bulkmainterm}, \eqref{eq:bulkangularcontrol} and by potentially choosing $\rblue$ even closer to $r_-$.  For the second term of \eqref{eq:bulk2terms} we will use \eqref{eq:wave} which reads \begin{align}\nonumber
0 = \Box_{g_{\mathrm{RNAdS}}} {{\psi}_\flat} + \frac{\alpha}{\ell^2} {{\psi}_\flat} = \frac{-4}{\Omega^2}(\partial_u \partial_v {{\psi}_\flat}) + \frac{2}{r} \left(\partial_v {{\psi}_\flat} + \partial_u {{\psi}_\flat} \right) + \frac{1}{r^2} \slashed\Delta_{\mathbb{S}^2} {{\psi}_\flat} + \frac{\alpha}{\ell^2} {{\psi}_\flat}
\end{align} to substitute $\partial_u \partial_v {{\psi}_\flat}$. Replacing $\partial_u \partial_v {{\psi}_\flat}$ and integrating by parts on the sphere, we  estimate all but one term of \eqref{eq:bulk2terms} using \eqref{eq:bulkangularcontrol} and \eqref{eq:bulkmainterm}. The term which we cannot estimate with \eqref{eq:bulkangularcontrol} and \eqref{eq:bulkmainterm} is of the form 
\begin{align}
  \Big| 	\int_{\mathcal{R}_f}   r^N
  \vp\operatorname{Re}\left(  \overline{{{\psi}_\flat}} (\partial_u {{\psi}_\flat}) \right) \Omega^6 \d{u}\d{v} \d\sigma_{\mathbb{S}^2}\Big| = \frac 12  \Big| 	\int_{\mathcal{R}_f}   r^N
  \vp \partial_u( |\psi_\flat|^2) \Omega^6 \d{u}\d{v} \d\sigma_{\mathbb{S}^2}\Big|.
\end{align}
This is of a similar form as the third term in \eqref{eq:bulk3terms}, which we control---as before---via an integration by parts in $u$.
Finally we have controlled all terms except for boundary terms arising from the integration by parts.

The first boundary terms arose from integrating by parts the first term in \eqref{eq:bulk3terms}. It consists of two parts and is of the form 
\begin{align}\label{eq:errortermcu}
\Big|	\int_{\mathcal{C}_{u_0}\cap\{ v_{\rblue}(u_0)\leq v \leq v_0
		\}} &r^{N+1} \vp  \operatorname{Re}\left( \overline{{\psi}_\flat} (\partial_v {{\psi}_\flat} )\right)\Omega^2 \d{v} \d\sigma_{\mathbb{S}^2} \Big|\\ & + \Big|	\int_{\Sigma_{\rblue} \cap J^-(v_0,u_0) } r^{N+1} \vp  \operatorname{Re}\left( \overline{{\psi}_\flat} (\partial_v {{\psi}_\flat} )\right)\Omega^2 \d{v} \d\sigma_{\mathbb{S}^2} \Big|.\label{eq:boundarytermsigma}
\end{align}
The second term \eqref{eq:boundarytermsigma} is absorbed in the past flux term on the spacelike hypersurface $\Sigma_{\rblue}$ by choosing $\rblue$ possibly closer to $r_-$ and noting that $\dvol_{\Sigma_{\rblue}} = {\sqrt{\Omega^2}} r^2 \d v \d \sigma_{\mathbb S^2}$. The first term \eqref{eq:errortermcu} is controlled as follows
\begin{align}\nonumber
\Big|&	\int_{\mathcal{C}_{u_0}\cap\{ v_{\rblue}(u_0)\leq v \leq v_0
	\}} r^{N+1} \vp \operatorname{Re}\left( \overline{   {{\psi}_\flat}} (\partial_v {{\psi}_\flat} )\right) \Omega^2 \d{v} \d\sigma_{\mathbb{S}^2} \Big| \\ \nonumber &\leq \Big|	\int_{\mathcal{C}_{u_0}\cap\{ v_{\rblue}(u_0)\leq v \leq v_0
	\}} r^{N+1} \vp     |\partial_v {{\psi}_\flat} |^2 \sqrt{\Omega^2 }\d{v} \d\sigma_{\mathbb{S}^2} \Big|\\ & + \Big|	\int_{\mathcal{C}_{u_0}\cap\{ v_{\rblue}(u_0)\leq v \leq v_0
	\}} r^{N+1} \vp    |{\psi}_\flat|^2 (\Omega^2)^{\frac{1}{4}} (\Omega^2)^{\frac{1}{4}} \Omega^2 \d{v} \d\sigma_{\mathbb{S}^2} \Big|. \label{eq:afteribp1}
\end{align} Now, note that \begin{align}\vp (\Omega^2)^{\frac{1}{4}} \lesssim \langle r_\ast - u\rangle^p (\Omega^2)^{\frac{1}{4}} \lesssim 1+ \up (\Omega^2)^{\frac{1}{4}}, \label{eq:afteribp2} \end{align} 
where we have used that $r_\ast^p (\Omega^2)^{\frac 14} \lesssim 1$ for $r_\ast \geq r_\ast(\rblue)$ which holds true since $\Omega^2$ decays exponentially as $r_\ast \to \infty$. Using \eqref{eq:afteribp2} we absorb \eqref{eq:afteribp1} in the flux term \eqref{eq:fluxterm1} by potentially choosing $\rblue$ closer to $r_-$ such that $\Omega^2$ is uniformly small in the blue-shift region. Completely analogously, we control the other boundary terms which arose from integrating by parts.

Now, we are left with the terms of the last two lines in \eqref{eq:bulk}.
\paragraph{Terms from last two lines of \eqref{eq:bulk}}
We will only look at the terms with $v$ weights as the terms involving $u$ weights are estimated completely analogously. It suffices to estimate the terms
\begin{align}\label{eq:4partii}
r^N \Big| \frac{\Omega^2}{2r}\vp \V |{{\psi}_\flat}|^2 \Big| + r^N \Big|\vp \frac{\partial_v (f^2\V)}{2f^2} |{{\psi}_\flat}|^2 \Big| \end{align} 
and 
\begin{align}- r^N  \vp   \frac{\partial_v f^2}{2f^2}\overline{\tilde{\nabla}_\sigma {{\psi}_\flat}} \tilde \nabla^\sigma {{\psi}_\flat} . \label{eq:hardterm} 
\end{align}
Since $ \Big| \frac{\partial_v (f^2\V)}{2f^2}  \Big|  \lesssim\Omega^2$, we control the terms in \eqref{eq:4partii} using \eqref{eq:bulkangularcontrol} and by potentially choosing $\rblue$ closer to $r_-$. Expanding \eqref{eq:hardterm} yields
\begin{align}\label{eq:lastterm}
-r^N  \vp \frac{\partial_v f^2}{2f^2}\tilde{\nabla}_\sigma {{\psi}_\flat}\overline{ \tilde \nabla^\sigma {{\psi}_\flat}}  = - 2 \beta_{\mathrm{blue}} r^N \vp \mathrm{Re}\left( \overline{ \tilde{\nabla}_u {{\psi}_\flat}} \tilde{\nabla}_v {{\psi}_\flat}\right)  + \frac{\beta_{\mathrm{blue}}}{2}r^N \vp \Omega^2 |\slashed\nabla \psi_\flat |^2. 
\end{align}
The second term on the right-hand side is estimated by \eqref{eq:bulkangularcontrol} and potentially choosing $\rblue$ closer to $r_-$. The first term on the right-hand side of \eqref{eq:lastterm} has the same from as \eqref{eq:secondindterm} and is estimated in the same way as \eqref{eq:secondindterm}. 

Finally, we have estimated and absorbed all sign-indefinite terms in the energy identity to obtain \eqref{eq:generalstrag}. Thus, we have proved \eqref{eq:energyidentityestim}, which concludes the first part of the proof.

\paragraph{\textbf{Part III: Proof of \eqref{eq:nontwistingenergy} and \eqref{eq:nontwistingenergy_an}}}
Now, observe that the estimate \eqref{eq:nontwistingenergy} follows from \eqref{eq:energyidentityestim} and \eqref{eq:changetwistedtopartial}. More precisely, the error arising from interchanging the twisted derivatives with partial derivatives on $\mathcal{C}_u$ are estimated as
\begin{align}\nonumber &\vp |\partial_v {{\psi}_\flat}|^2=  \vp |\dvpsi|^2  + \vp  \Omega^2 \operatorname{Re}\left( \overline{{\psi}_\flat} \partial_v {{\psi}_\flat}\right) - \frac 14 \vp  \Omega^4 |{{\psi}_\flat}|^2 \\ \nonumber  &\leq \vp |\dvpsi|^2  + | \vp  \Omega^2  \operatorname{Re}\left(\overline{{{\psi}_\flat}}\partial_v {{\psi}_\flat}\right) |.
\end{align}
Finally, note that the error term on the right hand side is controlled as in \eqref{eq:errortermcu}. This works for $\underline{\mathcal{C}}_{v}$ completely analogously which concludes the proof. 
\end{proof}
\subsubsection{Uniform boundedness and continuity at the Cauchy horizon for bounded frequencies}
Now, \cref{prop:decayinblueshift} allows us to prove the uniform boundedness.
\begin{prop}\label{prop:lowbound}
	Let ${{\psi}_\flat}$ be as defined in \eqref{eq:fouriertransforminterior}.
	Then, 
	\begin{align}\label{eq:uniformflat}
	\sup_{\mathcal B\cap J^+(\Sigma_0)}	|{{\psi}_\flat}|^2 \lesssim E_1[\psi_\flat](0) + \sum_{i,j=1}^3 E_1[\mathcal{W}_i \mathcal{W}_j \psi_\flat](0) \lesssim D[\psi_\flat].
	\end{align} 
	\begin{proof}
In view of \cref{prop:boundednessred}, it suffices to prove \eqref{eq:uniformflat} only in $J^+(\Sigma_{\rblue})\cap \mathcal{B}$. Let $(u_0,v_0)\in J^+(\Sigma_{\rblue})\cap \mathcal{B}$ be arbitrary. Then, by \cref{prop:boundednessred}, \cref{prop:decayinblueshift} and the Sobolev embedding on the sphere $H^2(\mathbb{S}^2)\hookrightarrow L^\infty(\mathbb{S}^2)$, we have
\begin{align}\nonumber
	|\psi_\flat(u_0,v_0,\varphi,\theta)|^2 \lesssim  &\left( \int_{v_{\rblue}(u_0)}^{v_0}  |\partial_v \psi_\flat(u_0,v,\varphi,\theta)| \d v \right)^2 + |\psi_\flat(u_0, v_{\rblue}(u_0),\varphi,\theta)|^2\\ \nonumber  \nonumber \lesssim  & \int_{\mathcal{C}_{u_0} (v_{\rblue}(u_0),v_0)} \langle v\rangle^p |\partial_v \psi_\flat|^2 \d v\d\sigma_{\mathbb{S}^2}+ \sum_{i,j} \int_{\mathcal{C}_{u_0} (v_{\rblue}(u_0),v_0)} \langle v\rangle^p |\partial_v \mathcal{W}_i\mathcal{W}_j\psi_\flat|^2 \d v \d\sigma_{\mathbb{S}^2}\\  &+ E_1[\psi_\flat](0) +  \sum_{i,j=1}^3 E_1[\mathcal{W}_i \mathcal{W}_j \psi_\flat] \lesssim  E_1[\psi_\flat](0) +  \sum_{i,j=1}^3 E_1[\mathcal{W}_i \mathcal{W}_j \psi_\flat](0), \label{eq:boundedness}
\end{align}
where $(\mathcal{W}_i)_{i=1,2,3}$ are the angular momentum operators. This shows \eqref{eq:uniformflat}.
\end{proof}
\end{prop}
\begin{prop}\label{continuityflat}Let $\psi_\flat$ be as defined in \eqref{eq:fouriertransforminterior}. Then, $\psi_\flat$  is continuously extendible beyond the Cauchy horizon $\mathcal{CH}$.
\begin{proof}
Similarly to \eqref{eq:boundedness} we have
\begin{align}\nonumber
		|\psi_\flat(u_0,v_2,\varphi,\theta)  - \psi_\flat(u_0, v_1,\varphi,\theta)|^2 & \lesssim  \int_{v_1}^{v_2} \langle v\rangle^{-p} \d v   \int_{v_1}^{v_2}  \langle v\rangle^{p} |\partial_v \psi_\flat(u_0,v,\varphi,\theta)|^2 \d v \\ &\lesssim \int_{v_1}^{v_2} \langle v\rangle^{-p} \d v   \left(  E_1[\psi_\flat] +  \sum_{i,j=1}^3 E_1[\mathcal{W}_i \mathcal{W}_j \psi_\flat]\right) \label{eq:continuity1}
\end{align}
uniformly in $u_0,\varphi,\theta$. 
The same estimate holds after interchanging the roles of $u$ and $v$. After commuting the equation with $\mathcal{W}_3$, we have from \eqref{eq:uniformflat}
\begin{align}
	\sup_{\mathcal{B}} |\partial_\varphi \psi |^2 \lesssim  E_1[\partial_\varphi \psi_\flat](0) + \sum_{i,j=1}^3 E_1[\mathcal{W}_i \mathcal{W}_j \partial_\varphi \psi_\flat](0) < \tilde C < \infty
\end{align} for some constant $\tilde C < \infty$ depending on the initial data. (Recall that we assumed our initial data to be smooth and compactly supported.) Thus, for $\varphi_1 \leq \varphi_2$, we have
\begin{align}
|\psi_\flat(u_0,v_0,\varphi_2,\theta)  - \psi_\flat(u_0, v_0,\varphi_1,\theta_0)|^2 & \lesssim \int_{\varphi_1}^{\varphi_2}  \sup_{\mathcal{B}} |\partial_\varphi \psi_\flat | \leq \tilde C |\varphi_2 - \varphi_1|
\end{align}
uniformly in $u_0,v_0,\theta_0$. A similar estimate holds true for $\theta$. Applications of the fundamental theorem of calculus and a triangle inequality finally yield the continuity result for $\psi_\flat$.
	\end{proof}
\end{prop}
\section{High frequency part \texorpdfstring{$\psi_\sharp$}{psisharp}}
\label{sec:highfreq}
In the previous section we have shown the uniform boundedness for the low frequency part $\psi_\flat$. Now, we turn to $\psi_\sharp$, the high frequency part. The key ingredient for the proof of the uniform boundedness for $|\psi_\sharp|$ in the interior is \emph{(a)} the uniform boundedness of transmission and reflection coefficients associated to the radial o.d.e.\ \eqref{eq:radialode} which is proved in \cite{kehle2018scattering} for $\Lambda=0$, together with \emph{(b)} the finiteness of the (commuted) $T$-energy flux on the event horizon given by \eqref{eq:boundednenergyflux}. 
 
Now, recall the radial o.d.e.~\eqref{eq:radialode} which reads $-u^{\prime \prime} + V_\ell u = \omega^2 u$ in the interior, where $V_\ell$ decays exponentially as $r_\ast \to + \infty (r\to r_-)$ and $r_\ast \to -\infty (r\to r_+)$. 
For $\omega \neq 0$, so in particular for $|\omega| > \frac{\omega_0}{2}$, the radial o.d.e.\ admits the following pairs of mode solutions $(u_1,u_2)$ and $(v_1,v_2)$, where $u_1$ and $u_2$ are solutions to \eqref{eq:radialode} satisfying $u_1 = e^{i\omega r_\ast} + O_\ell (r-r_+)$ and $u_2 = e^{-i\omega r_\ast} + O_\ell (r-r_+)$ as $r_\ast \to - \infty $. Similarly, $v_1$ and $v_2$ satisfy $v_1 = e^{i\omega r_\ast} + O_\ell(r-r_-)$  and $v_2 = e^{-i \omega r_\ast} + O_\ell(r - r_-)$ as $r_\ast \to +\infty$. Now, for $\omega \neq 0$, the transmission and reflection coefficients $\mathfrak T(\omega,\ell)$ and $\mathfrak R(\omega,\ell)$ are defined as the unique coefficients satisfying
\begin{align}
	u_1 = \mathfrak T(\omega,\ell)  v_1 + \mathfrak R (\omega,\ell)v_2.
\end{align}
See \cite{kehle2018scattering} for more details. In the following we will state the uniform boundedness of $\mathfrak T(\omega,\ell)$ and $\mathfrak R(\omega,\ell)$ for $|\omega| \geq \frac{\omega_0}{2}$. In {\cite[Proposition~4.7, Proposition~4.8]{kehle2018scattering} this has been proven for $\Lambda=0$. However, the proof of Proposition~4.7 and Proposition~4.8 in \cite{kehle2018scattering} also applies if we include a non-vanishing cosmological constant.\footnote{Note that for  $\Lambda\neq 0$ the scattering coefficients $\mathfrak R$ and $\mathfrak T$ have a pole at $\omega=0$. However, for frequencies bounded away from $\omega =0$, so in particular for $|\omega| \geq \frac{\omega_0}{2}$ as in the present case, $\mathfrak T$ and $\mathfrak R$ are uniformly bounded for both cases $\Lambda =0$ \emph{and}  $\Lambda \neq 0 $.   See \cite{kehle2018scattering} for more details.}
	
\begin{lemma}[{\cite[Proposition~4.7, Proposition~4.8]{kehle2018scattering}}] Fix subextremal Reissner--Nordström--AdS black hole parameters  $(M,Q,l)$, a constant $\omega_0 >0$ and a Klein--Gordon mass parameter $\alpha<\frac 94$. Then, the scattering coefficients $\mathfrak T(\omega,\ell)$ and $\mathfrak R (\omega,\ell)$ as defined above satisfy
	\begin{align}\label{eq:boundednessoftr}
		 \sup_{|\omega| \geq \frac{\omega_0}{2}, \ell \in \mathbb{N}_0} \left(  |\mathfrak T(\omega,\ell) | + |\mathfrak{R}(\omega,\ell) | \right) \lesssim_{M,Q,l,\omega_0,\alpha} 1
	\end{align}
and the mode solutions $u_1,u_2$ and $v_1,v_2$ are uniformly bounded
\begin{align}\label{eq:u1u2bounded}
	 &\sup_{|\omega| \geq \frac{\omega_0}{2}, \ell \in \mathbb{N}_0}	\|u_1 \|_{L^\infty(\mathbb{R})} \lesssim_{M,Q,l,\omega_0,\alpha} 1, \; 	 \sup_{|\omega| \geq\frac{\omega_0}{2}, \ell \in \mathbb{N}_0}	\|u_2 \|_{L^\infty(\mathbb{R})} \lesssim_{M,Q,l,\omega_0,\alpha} 1,\\
	 & \sup_{|\omega| \geq \frac{\omega_0}{2}, \ell \in \mathbb{N}_0}	\|v_1 \|_{L^\infty(\mathbb{R})} \lesssim_{M,Q,l,\omega_0,\alpha} 1, \; 	 \sup_{|\omega| \geq \frac{\omega_0}{2}, \ell \in \mathbb{N}_0}	\|v_2 \|_{L^\infty(\mathbb{R})} \lesssim_{M,Q,l,\omega_0,\alpha} 1.
\end{align}
\begin{proof}
	Since we are the regime $|\omega|\geq \frac{\omega_0}{2}$, the proof for $\Lambda <0$ works exactly as for $\Lambda=0$ as shown in {\cite[Proposition~4.7, Proposition~4.8]{kehle2018scattering}}. Thus, we will be very brief.  
	
	We first consider the case $\ell \leq \ell_0$, where $\ell_0$ is chosen sufficiently large later in the second part. Note that $u_1$ solves the Volterra equation
	\begin{align}
		u_1(r_\ast) = e^{i\omega r_\ast} + \int_{-\infty}^{r_\ast} \frac{\sin(\omega (r_\ast - y))}{\omega} V(y) u_1(y) \d y.
	\end{align}
	As $|\omega|\geq \frac{\omega_0}{2}$ and since the potential $V$ is uniformly bounded (in the regime $\ell \leq \ell_0$) and decays exponentially as $r_\ast \to \pm \infty$ , standard estimates for Volterra integral equations (see \cite[Proposition~2.3]{kehle2018scattering}) yield \eqref{eq:u1u2bounded} for $u_1$ and similarly for $u_2$, $v_1$ and $v_2$. 
	
	For the regime $\ell \geq \ell_0$, we will use a WKB approximation. Indeed, choosing $\ell_0$ sufficiently large, we have that $p:= \omega^2 - V$ is positive for $r_\ast \in \mathbb R$ and smooth. Now, $u_1$ is a solution of the radial o.d.e.\ $u'' = -p u$. Just like in \cite[Equation (4.149)]{kehle2018scattering} we control the error term $F(r_\ast) = \int_{-\infty}^{r_\ast} p^{- \frac 14} |\frac{\d{}^2 }{\d y^2} p^{-\frac 14} | \d y$ of the WKB approximation and conclude that $u_1$ remains uniformly bounded. Similarly, this holds true for $u_2$, $v_1$ and $v_2$ and for the scattering coefficients $\mathfrak R$ and $\mathfrak T$ which concludes the proof.
\end{proof}
\end{lemma}
Another result which we will use from \cite{kehle2018scattering} is the representation formula for $\psi_\sharp$ in the separated picture. It is essential that $|\omega|\geq \frac{\omega_0}{2}$ to apply the same steps as in \cite[Proof of Proposition~5.1]{kehle2018scattering}.
\begin{lemma}[{\cite[Proof of Proposition~5.1]{kehle2018scattering}}]\label{lem52}
Let $\psi_\sharp$ as in \eqref{eq:fouriertransforminterior}. Then, we have 
\begin{align}\nonumber
	\psi_\sharp (t,r,\varphi,\theta) = 	& \frac{1}{\sqrt{2\pi} r}\sum_{\ell \in \mathbb{N}_0} \sum_{|m|\leq \ell} Y_{\ell m}(\theta,\varphi)  \int_{|\omega|\geq \frac{\omega_0}{2}} \mathcal{F}_{\mathcal{H}^+_A}\left[\psi_\sharp\restriction_{\mathcal{H}_A^+}\right](\omega, m,\ell) u_1(\omega,\ell,r) e^{i\omega t} \d \omega\\ &+ \frac{1}{\sqrt{2\pi} r}\sum_{\ell \in \mathbb{N}_0} \sum_{|m|\leq \ell} Y_{\ell m}(\theta,\varphi) \int_{|\omega|\geq \frac{\omega_0}{2}} \mathcal{F}_{\mathcal{H}^+_B}\left[\psi_\sharp\restriction_{\mathcal{H}_B^+}\right](\omega, m,\ell) u_2(\omega,\ell,r) e^{i\omega t} \d \omega, \label{eq:psisharp}
\end{align}
where
\begin{align}\label{eq:56}
\mathcal{F}_{\mathcal{H}^+_A}[\phi](\omega,m,\ell) :=\frac{r_+}{\sqrt{2\pi}} \int_{\mathbb R} e^{-i\omega v} \langle \phi, Y_{\ell m}\rangle_{\mathbb S^2} \d v\end{align}
and
\begin{align}
\mathcal{F}_{\mathcal{H}^+_B}[\phi](\omega,m,\ell) :=\frac{r_+}{\sqrt{2\pi}} \int_{\mathbb R} e^{i\omega u} \langle \phi, Y_{\ell m}\rangle_{\mathbb S^2} \d u.\end{align}
\begin{proof}[Proof of \cref{lem52}]
	This proof is very similar to \cite[Proof of Proposition~5.1]{kehle2018scattering} so we will be rather brief.

	Let $\psi_\sharp$ as in \eqref{eq:fouriertransforminterior}. Since the expansion in spherical harmonics converges pointwise, it suffices to prove \eqref{eq:psisharp} for $\psi_\sharp^{\ell m} := \langle \psi_\sharp, Y_{\ell m} \rangle_{\mathbb{S}^2} Y_{\ell m}$ for fixed $m,\ell$. Now, define $u[\psi_\sharp^{\ell m}]$ as in \eqref{eq:definitionu} such that 
	\begin{align}
		\psi_\sharp^{\ell m} = \frac{1}{\sqrt{2\pi} r} Y_{\ell m} \int_{|\omega|\geq\frac{\omega_0}{2}} u[\psi_\sharp^{\ell m}] e^{i\omega t} \d \omega.
	\end{align} This is well-defined in the interior in view of \cref{rmk:welldefnofu}. Moreover, $u[\psi_\sharp^{\ell m}]$ solves the radial o.d.e.\ and can be expanded in the basis $u_1$ and $u_2$ ($|\omega|> \frac{\omega_0}{2}$):
	\begin{align}
		u[\psi_\sharp^{\ell m}](r_\ast ,\omega,m, \ell) = a(\omega,m, \ell) u_1(r_\ast,\ell,\omega) + b(\omega,m,\ell) u_2(r_\ast,\ell,\omega).
	\end{align}
	Now, first note \cref{prop:decayofpsil} implies that $\omega \mapsto u[\psi_\sharp^{\ell m}](r,\omega) $ is a Schwartz function for $r\in (r_-,r_+)$. Since 
	\begin{align}
	|	a(\omega,m,\ell)| = \left|\frac{\mathfrak W(u[\psi_\sharp^{\ell m}] , u_2)}{\mathfrak W(u_1,u_2)}\right| = \left|\frac{\mathfrak W(u[\psi_\sharp^{\ell m}] , u_2)}{2 \omega}\right| \lesssim  \left|\mathfrak W(u[\psi_\sharp^{\ell m}],u_2) \right|
	\end{align}
in view of $|\omega| \geq \frac{\omega_0}{2}$, we conclude that $\omega \mapsto a(\omega,m,\ell)$ is in $L^1(\mathbb{R})$ for fixed $\ell,m$. Recall that the Wronskian $\mathfrak W(f,g) := f' g - fg'$ is independent of $r_\ast$ for two solutions of the radial o.d.e.~\eqref{eq:radialode}. We have also used that $\| u_2 \|_{L^\infty} \lesssim 1$ and $\| u_2^{\prime} \|_{L^\infty}\lesssim_{\ell}1+| \omega |$ for $|\omega|\geq \frac{\omega_0}{2}$ (cf.\ \cite[Proposition~4.7 and Proposition~4.8]{kehle2018scattering}).   Similarly, we have that $\omega \mapsto b(\omega,m,l)$ is in $L^1(\mathbb{R})$. 
	Using
	\begin{align}\label{eq:psisharplm}
		 \psi_\sharp^{\ell m}= Y_{\ell m}  \frac{1}{\sqrt{2\pi}r} \int_{|\omega|\geq \frac{\omega_0}{2}}  \left( a(\omega,m, \ell) u_1(r,\omega,\ell) + b(\omega,m,\ell) u_2(r,\omega,\ell) \right) e^{i\omega t}\d\omega
	\end{align}
and a direct adaptation of \cite[Proof of Proposition~5.1]{kehle2018scattering} finally shows $a(\omega,m,\ell) = \mathcal{F}_{\mathcal{H}_A^+}[\psi^{\ell m}_\sharp\restriction_{\mathcal{H}_A^+}](\omega,m,\ell)$, $b(\omega,m,\ell) = \mathcal{F}_{\mathcal{H}_B^+}[\psi^{\ell m}_\sharp\restriction_{\mathcal{H}_B^+}](\omega,m,\ell)$.\footnote{More precisely, following the lines starting from equation (5.20) in \cite[Proof of Proposition~5.1]{kehle2018scattering} which contain an application of Lebesgue's dominated convergence, the Riemann--Lebesgue lemma and the inverse Fourier transform yields the result.} This shows the representation formula \eqref{eq:psisharp} for $\psi_\sharp$. 
\end{proof}
\end{lemma}
We will now prove the uniform boundedness for $\psi_\sharp$.
\begin{prop}\label{prop:highbound}
	Let $\psi_\sharp$ be as defined in \eqref{eq:fouriertransforminterior}. 
	Then, 
	\begin{align}\label{eq:uniformsharp}
	\sup_{\mathcal B\cap J^+(\Sigma_0)}	|{{\psi}_\sharp}|^2 \lesssim E_1[\psi_\sharp](0)+ \sum_{i,j=1}^3 E_1[\mathcal{W}_i \mathcal{W}_j \psi_\sharp](0) \lesssim D[\psi_\sharp].
	\end{align} 
	\begin{proof} 
We start with the representation of $\psi_\sharp$ as in \eqref{eq:psisharp}.
For convenience, we will only estimate the term involving $\mathcal{F}_{\mathcal{H}^+_A}[\phi](\omega,m,\ell)$ and assume without loss of generality that $\mathcal{F}_{\mathcal{H}^+_B}[\phi](\omega,m,\ell)=0$. Indeed the term $\mathcal{F}_{\mathcal{H}^+_B}[\phi](\omega,m,\ell)$ can be treated analogously.
Now, in view of \eqref{eq:u1u2bounded}, 
we conclude
\begin{align}\nonumber
	|	\psi_\sharp(r,t,\varphi,\theta)|^2 &\lesssim\left|\sum_{\ell\in \mathbb{N}_0} \sum_{{m\in\mathbb{Z}}, {|m|\leq \ell}} Y_{\ell m}(\varphi,\theta) \int_{|\omega|\geq \frac{\omega_0}{2}}\mathcal{F}_{\mathcal{H}_A}\left[\psi_\sharp\restriction_{\mathcal{H}_A^+}\right](\omega,m,\ell) \d \omega\right|^2 \\\nonumber &
	\leq 
	  \sum_{\ell\in \mathbb{N}_0}  \sum_{{m\in\mathbb{Z}}, {|m|\leq \ell}} \int_{|\omega|\geq \frac{\omega_0}{2}}  (1 + \ell)^3 \omega^2 \left| \mathcal{F}_{\mathcal{H}_A}\left[\psi_\sharp\restriction_{\mathcal{H}_A^+}\right](\omega, m,\ell)\right|^2 \d \omega \\ \nonumber & \cdot 
\sum_{\ell\in \mathbb{N}_0}  \sum_{{m\in\mathbb{Z}}, {|m|\leq \ell}} \frac{|Y_{\ell m}(\varphi,\theta)|^2}{(1+ \ell)^3}\int_{|\omega | \geq \frac{\omega_0}{2} }\frac{1}{\omega^2} \d \omega
	  \\ \nonumber &
	\lesssim \sum_{\ell\in \mathbb{N}_0} \sum_{{m\in\mathbb{Z}}, {|m|\leq \ell}} \int_{|\omega|\geq \frac{\omega_0}{2}} (1 + \ell)^3 \omega^2 \left| \mathcal{F}_{\mathcal{H}_A}\left[\psi_\sharp\restriction_{\mathcal{H}_A^+}\right](\omega,m,\ell)\right|^2 \d \omega\\ & \lesssim  	\int_{\mathcal{H}_A^+} |T\psi_\sharp|^2  \d v \d \sigma_{\mathbb{S}^2}+
\sum_{i,j=1}^3	\int_{\mathcal{H}_A^+} |T\mathcal{W}_{i}\mathcal{W}_j \psi_\sharp|^2\d v \d \sigma_{\mathbb{S}^2}. \label{estimateonuniformboundforpsisharp}
\end{align}
Here, we have used that \begin{align}
\label{eq:unsold}
\sum_{m = -\ell}^{\ell} |Y_{\ell m}(\varphi,\theta)|^2 = \frac{2\ell +1}{4\pi}\end{align} which is known as Unsöld's Theorem~\cite[Eq.~(69)]{unsold1927beitrage}. 

Finally, on the right hand side of \eqref{estimateonuniformboundforpsisharp} we only see the commuted $T$-energy flux. An application of the $T$-energy identity in the exterior and an energy estimate in a  compact spacetime region shows that the commuted $T$-energy flux on the event horizon is controlled from the initial data (cf.\ \eqref{eq:boundednenergyflux} in \cref{prop:wellposedness}). Thus, in view of \eqref{estimateonuniformboundforpsisharp} we conclude 
\begin{align}
|\psi_\sharp (r,t,\varphi,\theta)|^2 \lesssim E_1[\psi_\sharp](0) + \sum_{i,j=1}^3 E_1[\mathcal{W}_i \mathcal{W}_j \psi_\sharp](0).
\end{align}
\end{proof}
\end{prop}

\begin{prop}
Let $\psi_\sharp$ be as defined in \eqref{eq:fouriertransforminterior}. Then, $\psi_\sharp$  is continuously extendible across the Cauchy horizon $\mathcal{CH}$. 
	\label{continuitysharp}
\begin{proof}
Let $(u_n,v_n,\theta_n,\varphi_n) \to (\tilde u,\tilde v,\tilde \theta,\tilde \varphi)$ be a convergent sequence. We will also allow $\tilde u=+\infty$ and $\tilde v=+\infty$ as limits which correspond to limits to the Cauchy horizon.
We represent $\psi_\sharp$ again as in \eqref{eq:psisharp}. Similar to the proof of \cref{prop:highbound}, it is enough to consider the case where $\mathcal{F}_{\mathcal{H}^+_B}[\psi_\sharp\restriction_{\mathcal{H}_B^+}]$ vanishes. Hence, 
\begin{align}\psi_\sharp (t,r,\varphi,\theta) = 	& \frac{1}{\sqrt{2\pi} r}\sum_{\ell \in \mathbb{N}_0} \sum_{|m|\leq \ell} Y_{\ell m}(\theta,\varphi)  \int_{|\omega|\geq \frac{\omega_0}{2}} \mathcal{F}_{\mathcal{H}^+_A}\left[\psi_\sharp\restriction_{\mathcal{H}_A^+}\right](m,\ell,\omega) u_1(\omega,\ell,r) e^{i\omega t} \d \omega.\end{align}

First from \eqref{eq:unsold} we have $\sup_{\varphi,\theta} |Y_{\ell m}(\varphi,\theta)|\lesssim 1+\ell$ and from \eqref{eq:u1u2bounded} we have that \begin{align*}\sup_{u,v} |u_1 e^{i\omega t(u,v)}| = \sup_{t,r} |u_1 e^{i \omega t}|\lesssim 1.\end{align*}  
Then, a similar estimate as in \eqref{estimateonuniformboundforpsisharp} and an application of Lebesgue's dominated convergence theorem allow us to interchange the limit $n\to\infty$ with the sum $\sum_{\ell \in \mathbb{N}_0} \sum_{|m|\leq \ell}$. Since $Y_{\ell m}(\theta_n,\varphi_n) \to Y_{\ell m}(\tilde \theta,\tilde\varphi) $ pointwise as $n\to\infty$, it remains to show that
\begin{align*}&\int_{|\omega|\geq \frac{\omega_0}{2}} \mathcal{F}_{\mathcal{H}^+_A}\left[\psi_\sharp\restriction_{\mathcal{H}_A^+}\right](m,\ell,\omega) u_1(\omega,\ell,r(u_n,v_m)) e^{i\omega t(u_n,v_n)} \d \omega \\ & = \int_{|\omega|\geq \frac{\omega_0}{2}}\mathcal{F}_{\mathcal{H}^+_A}\left[\psi_\sharp\restriction_{\mathcal{H}_A^+}\right](m,\ell,\omega) \Big( \mathfrak T (\omega,\ell) v_1(\omega,\ell,r(u_n,v_n))  + \mathfrak R (\omega,\ell) v_2(\omega,\ell,r(u_n,v_n))\Big)  e^{i\omega t(u_n,v_n)} \d \omega \end{align*}
 converges as $n\to\infty$ for fixed angular parameters $m,\ell$. But, in view of \eqref{eq:boundednessoftr}, depending on whether $\tilde v = +\infty$ or $\tilde u = +\infty$, we can deduce the continuity using Lebesgue's dominated convergence and the Riemann--Lebesgue lemma. Both are justified by a slight adaptation of the steps which resulted in \eqref{eq:psisharplm}. This concludes the proof.
\end{proof}
\end{prop}

\appendix
\addtocontents{toc}{\protect\setcounter{tocdepth}{1}}
\section{Appendix}
\setcounter{equation}{0}  
\setcounter{prop}{0} 
\renewcommand{\theequation}{A.\arabic{equation}}
\renewcommand{\theprop}{A.\arabic{prop}}
\renewcommand{\thecor}{A.\arabic{cor}}
\subsection{Twisted energy-momentum tensor in null coordinates in the interior}
We will write out the components of the twisted energy-momentum tensor in the interior.
\label{appendix}
 \begin{prop}\label{prop:appendix} Consider null coordinates $(u,v,\theta,\varphi)$ in the interior region $\mathcal{B}$. Recall that the metric is given by \eqref{eq:metricinterior}.  Let $f\in C^\infty(\mathcal B )$ be a spherically symmetric nowhere vanishing real valued function and $X$ be a smooth vector field of the form $X = X^u \partial_u + X^v \partial_v$.
 	
 	The components of the twisted energy-momentum tensor \eqref{eq:twistedemt} associated to $f$ are given by
 	 	\begin{align*}
 	&\tilde{\mathbf T}_{uu} = |\tilde \nabla_{u} {{\phi}} |^2 = f^2 \left| \partial_u \left(\frac{{\phi}}{f}\right)\right|^2, \tilde T_{vv} = |\tilde \nabla_{v} {\phi} |^2 =f^2 \left| \partial_v \left(\frac{\phi}{f}\right)\right|^2,\\
 	& 	\tilde{\mathbf T}_{uv} = \tilde{\mathbf T}_{vu}= \frac{\Omega^2}{4}\left(|\slashed\nabla{\phi}|^2 + \V |{\phi}|^2\right), \\
 	&\tilde{\mathbf T}_{\theta\theta} = |\partial_\theta {\phi}|^2 + \frac{2r^2}{\Omega^2} \operatorname{Re}\left( \overline{\tilde{\nabla}_u {\phi} }\tilde{\nabla}_v {\phi}\right) - \frac{r^2}{2}\left( |\slashed\nabla{\phi}|^2 + \V|\phi|^2\right),\\
 	&\tilde{\mathbf T}_{\varphi\varphi} = |\partial_\varphi {\phi}|^2 + \frac{2r^2\sin^2\theta }{\Omega^2}\operatorname{Re}\left(\overline{ \tilde{\nabla}_u {\phi}} \tilde{\nabla}_v {\phi}\right) - \frac{r^2 \sin^2\theta}{2}\left( |\slashed\nabla{\phi}|^2 + \V |\phi|^2\right).
 	\end{align*}
 	
 	The deformation tensor ${}^X\pi := \frac 12 \mathcal{L}_X g $ is given by
 	\begin{align*}
 	&{}^X\pi^{vv} =- \frac{2}{\Omega^2} \partial_u X^v, 
 	{}^X\pi^{uu} =- \frac{2}{\Omega^2} \partial_v X^u, {}^X\pi^{uv} =- \frac{1}{\Omega^2} \left(\partial_u X^u + \partial_v X^v\right) - \frac{2}{\Omega^2} \left( \frac{\partial_v \sqrt{\Omega^2}}{\sqrt{\Omega^2}} X^v + \frac{\partial_u \sqrt{\Omega^2}}{\sqrt{\Omega^2}}X^u\right), \\ & {}^X\pi^{\theta\theta} =- \frac{\Omega^2}{2 r^3}(X^v+X^u), {}^X\pi^{\varphi\varphi} =- \frac{\Omega^2}{2 r^3 \sin^2\theta}(X^v+X^u). 
 	\end{align*}
 	
 	In the following we explicitly write down future-directed normals and induced volume forms for hypersurfaces of constant $r$ values $\Sigma_r$ and for null cones $\mathcal{C}_u$ and $\underline{\mathcal{C}}_v$ of constant $u$ and $v$ values, respectively.
 	\begin{align*}
 	&n_{\Sigma_r} = \frac{1}{\sqrt{\Omega^2}} (\partial_u + \partial_ v), \d\mathrm{vol}_{\Sigma_r} = {r^2}\sqrt{\Omega^2}\d\sigma_{\mathbb{S}^2} \d u = {r^2}{}\sqrt{\Omega^2}\d\sigma_{\mathbb{S}^2} \d v,\\
 	&n_{\underline{\mathcal{C}}_v}= \frac{2}{\Omega^2} \partial_u,  \d\mathrm{vol}_{\underline{\mathcal{C}}_v} = \frac{r^2}{2}\Omega^2 \d \sigma_{\mathbb S^2} \d u,\\ 
 	& n_{{\mathcal{C}}_u} = \frac{2}{\Omega^2} \partial_v, \d\mathrm{vol}_{{\mathcal{C}}_u}  = \frac{r^2}{2}\Omega^2 \d \sigma_{\mathbb S^2} \d v.
 	\end{align*}
 	Then, the fluxes of $X$ are given by
 	\begin{align}&\tilde J^X_\mu [\phi] n^\mu_{\mathcal{C}_u}  = \frac{2X^v}{\Omega^2} |\tilde \nabla_v {\phi}|^2 + \frac{X^u}{2} \left(|\slashed\nabla {\phi}|^2 + \V |\phi|^2 \right), \\\label{eq:jxv}
 	&\tilde J^X_\mu [\phi] n^\mu_{\underline{\mathcal C}_v}  = \frac{2X^u}{\Omega^2} |\tilde \nabla_u {\phi}|^2 + \frac{X^v}{2} \left(|\slashed\nabla {\phi}|^2 + \V |\phi|^2 \right), \\
 	&\tilde J^X_\mu[{\phi}] n_{\Sigma_r}^\mu  = \frac{1}{\sqrt{\Omega^2}} \left( X^u |\tilde \nabla_u {\phi}|^2 + X^v |\tilde \nabla_v {\phi}|^2 + \frac{\Omega^2}{4}(X^u + X^v) (|\slashed \nabla \phi|^2 + \V |\phi|^2) \right).
 	\end{align}
The twisted bulk term associated to the twisting function $f$ reads (cf.\ \cite{warnick_massive_wave})
 	\begin{align*}
 	\tilde K^X = {}^X\pi_{\mu \nu} \tilde T^{\mu\nu} + X^\nu \tilde S_\nu,
 	\end{align*}
 	where 
 	\begin{align*}
 	\tilde S_\nu = \frac{\tilde \nabla^\ast_\nu (f\V)}{2f} |\phi|^2 + \frac{\tilde \nabla^\ast_\nu f}{2f} \overline{\tilde{\nabla}_\sigma {\phi} }\tilde \nabla^\sigma {\phi}.
 	\end{align*}
 	In coordinates we have
 	\begin{align}\nonumber
 	\tilde	K^X =& - \frac{2}{\Omega^2} \left( \partial_u X^v |\tilde \nabla_v {\phi}|^2 + \partial_v X^u |\tilde{\nabla}_u {\phi}|^2 \right) - \frac{2}{r}(X^u + X^v) \operatorname{Re}(\overline{\tilde{\nabla}_u {\phi}} \tilde \nabla_v {\phi})\\ \nonumber
 	&- \left( \frac{1}{2}(\partial_v X^v + \partial_u X^u) -  \frac{\partial_{r} {\Omega^2}}{{4}}\left( X^v+ X^u \right)\right)\left( |\slashed\nabla {\phi}|^2 + \V |\phi|^2\right)\\
 	& + \frac{\Omega^2}{2r}(X^v + X^u) \V |\phi|^2 + X^u\left( -\frac{\partial_u (f^2\V)}{2f^2} |\phi|^2 - \frac{\partial_u f^2}{2f^2}\overline{\tilde{\nabla}_\sigma {\phi}} \tilde \nabla^\sigma {\phi}\right)\nonumber \\
 	& +  X^v\left( -\frac{\partial_v (f^2\V)}{2f^2} |\phi|^2 - \frac{\partial_v f^2}{2f^2}  \overline{\tilde{\nabla}_\sigma {\phi}} \tilde \nabla^\sigma {\phi}\right). \label{eq:bulk}
 	\end{align}
 	\end{prop}
 	\subsection{Construction of the twisted red-shift vector field}
 	\label{sec:twistedredshift}
 	In this section we will give the proof of \cref{prop:redshift}. 
 	\begin{proof}[Proof of \cref{prop:redshift}]We choose the ansatz $N= N^u \partial_u + N^v \partial_v$ for our red-shift vector field.
 		We will first estimate the twisted 1-jet $\tilde J$ and then the twisted bulk term $\tilde K$.
 		\paragraph{\texorpdfstring{$\tilde J$ current}{J current}}
 		From \eqref{eq:jxv}, we have
 		\begin{align} \tilde J^N_\mu [{\phi}] n_{\underline{\mathcal{C}}_v}^\mu  = \frac{2N^u}{\Omega^2} |\tilde \nabla_u {\phi}|^2 + \frac{N^v}{2} \left(|\slashed\nabla {\phi}|^2 + \V {|\phi|}^2 \right),\end{align}
 		where
 		\begin{align}\V = - \left(\frac{\Box_g f}{f}+ \frac{\alpha}{l^2}\right).\end{align}
 		First, if $f= f(r)$ we have 
 		\begin{align}
 		-\frac{\Box_g f}{f}= {\Omega^2} \frac{\ddot{f}}{f} + \left( \frac{2\Omega^2}{r} + {\partial_r(\Omega^2)}\right) \frac{\dot f}{f},
 		\end{align}
 		where $\dot f := \frac{\d f}{\d r}$. Thus, choosing $f=e^{- \beta_{\mathrm{red}}r}$ gives
 		\begin{align}\label{eq:f}
 		\mathcal{V} = -\left( \frac{\Box_g f}{f} + \frac{\alpha}{l^2}\right) = \beta_{\mathrm{red}}^2 {\Omega^2 } -  \partial_r (\Omega^2) \beta_{\mathrm{red}} - \frac{2\beta_{\mathrm{red}} }{r}  \Omega^2 - \frac{\alpha}{l^2}.
 		\end{align}
 		Note that for $\red < r_+$ close enough to $r_+$, we have 
 		\begin{align}
 		- \partial_r \Omega^2 \geq c_{\mathrm{red}}
 		\end{align}
 		for all $\red \leq r \leq r_+$ and some constant $c_{\mathrm{red}}>0$ only depending on the black hole parameters. The constant $c_{\mathrm{red}}>0$ does not decrease, when we choose $\red$ even closer $r_+$. 
 		Now, by choosing $\beta_{\mathrm{red}} >0$ large enough to absorb the negative contribution from $-\frac{\alpha}{l^2}$ and by choosing $\red$ close enough to $r_+$, we ensure that $\V \gtrsim 1$ in $r_\mathrm{red}\leq r \leq r_+$. This finally shows that if we take $N$ as a future directed vector field, the 1-jet $\tilde J^N_\mu n_{\underline{\mathcal{C}}_v}^\mu$ is positive definite. We will construct the explicit form of $N$ in the bulk term estimate.
 		\paragraph{\texorpdfstring{Bulk term $\tilde K^N$}{Bulk term KN}}
 		Now, we will estimate the bulk term.
 		We will choose the components of the timelike vector field $N= N^u\partial_u + N^v \partial_v $ as
 		\begin{align}N^u := \frac{1}{\Omega^2} - \frac{1}{\delta_1} \text{ and } N^v := 1-\frac{\Omega^2}{\delta_2}.
 		\end{align}
 		Note that $N$ is smooth in $\Rred$. Moreover, for fixed $\delta_1,\delta_2>0$ (only depending on the black hole parameters), we can choose $\red$ close enough to $r_+$ such that $N$ is future directed in $\Rred$. 
 		Then, note that 
 		\begin{align}\label{eq:kx1}
 		\tilde K^N[\phi]= &  \left( - \partial_r \Omega^2 \right)\left(  \frac{1}{\delta_2}   |\tilde \nabla_v {\phi}|^2 + \frac{1}{\Omega^4} |\tilde \nabla_u {\phi}|^2 \right) - \frac{2}{r}\left(\frac{1}{\Omega^2} - \frac{1}{\delta_1} +1 - \frac{1}{\delta_2} \Omega^2  \right) \operatorname{Re}(\overline{\tilde{\nabla}_u {\phi}}\tilde{\nabla}_v {\phi}) \\\label{eq:kx2}
 		& +\frac 14 \left( -\frac{\d \Omega^2 }{\d r} \right)  \left( \frac{1}{\delta_1} - 1 + \frac{2\Omega^2}{\delta_2} \right)  ( |\slashed \nabla {\phi}|^2 + \V |{\phi}|^2)\\
 		& + \frac{1}{2r}\left(1 + \left( 1-\frac{1}{\delta_1} \right) \Omega^2 - \frac{1}{\delta_2} \Omega^4 \right) \V |{\phi}|^2 \\ & + 
 		\left(	\frac{1}{\Omega^2} - \frac{1}{\delta_1} \right)\frac{ - \partial_u (f^2\V)}{2f^2} |{\phi}|^2 +	\left( \frac{1}{\Omega^2} - \frac{1}{\delta_1} \right) \frac{- \partial_u (f^2)}{2f^2}\overline{\tilde{\nabla}_\sigma {\phi}} \tilde \nabla^\sigma {\phi}\label{eq:kx4}\\
 		&+ \left(1-\frac{\Omega^2}{\delta_2}\right) \frac{ - \partial_v (f^2\V)}{2f^2} |\phi|^2 +	\left(1-\frac{\Omega^2}{\delta_2}\right) \frac{- \partial_v (f^2)}{2f^2} \overline{\tilde{\nabla}_\sigma {\phi}} \tilde \nabla^\sigma {\phi}
 		\label{eq:kx5}.\end{align}
 		In the following we will show that 
 		\begin{align}
 		\tilde K^N[\phi] 	& \gtrsim \frac{1}{\Omega^4} |\tilde \nabla_u {\phi}|^2 +  |\tilde \nabla_ v {\phi}|^2 + (|\slashed \nabla {\phi}|^2 + \V |\phi|^2). \label{eq:positivitybulk}
 		\end{align}
 		We will start with the sign-indefinite term appearing in \eqref{eq:kx1}. We estimate it as follows
 		\begin{align}\label{eq:redshift}
 		\Big| - \frac{2}{r}\left(\frac{1}{\Omega^2} -\frac{1}{\delta_1} +1 - \frac{1}{\delta_2} \Omega^2  \right) \operatorname{Re}(\overline{\tilde{\nabla}_u {\phi}}\tilde{\nabla}_v {\phi}) \Big| \lesssim \frac{\epsilon}{\Omega^4} |\tilde{\nabla}_u {\phi}|^2 + \frac{1}{\epsilon} |\tilde{\nabla}_v {\phi}|^2,
 		\end{align} 
 		where we have applied an $\epsilon$-weighted Young's inequality. We have also used that---by choosing $\red$ closer to $r_+$---we can make $\Omega^2$ uniformly smaller than any constant, in particular smaller than $\delta_1$ and $\delta_2$ once those are fixed. Choosing $\epsilon$ small enough, we absorb the term $\frac{\epsilon}{\Omega^4} |\tilde{\nabla}_u {\phi}|^2$ of \eqref{eq:redshift} in the first term of \eqref{eq:kx1}. Then, choosing $\delta_2(\delta_1,\epsilon)$ small enough, we can also absorb the term $ \frac{1}{\epsilon} |\tilde{\nabla}_v {\phi}|^2$ in the first term of \eqref{eq:kx1}. Completely analogously and by potentially choosing $\delta_2$ and $\delta_1$ even smaller, we estimate the terms of the form $\frac{1}{\Omega^2} \operatorname{Re}(\overline{\tilde{\nabla}_u {\phi}}\tilde{\nabla}_v {\phi})$ arising from \eqref{eq:kx4} and \eqref{eq:kx5}.
 		
 		Next, note that, in view of $\mathcal V\gtrsim 1$ and $ \Big|\frac{ - \partial_v (f^2\V)}{2f^2} \Big|\lesssim \Omega^2$, we choose $\delta_1$ small enough such that we absorb error terms coming from \eqref{eq:kx4} and \eqref{eq:kx5} in the term with the good sign in \eqref{eq:kx2}. By doing so we also have to make $\delta_2(\epsilon,\delta_1)>0$ small enough. 
 		Finally, once $\delta_1$ and $\delta_2$ are fixed, note that we can make terms involving higher orders of $\Omega^2$ arbitrarily small by choosing $\red$ close to $r_+$.  This finally shows \eqref{eq:positivitybulk} and concludes the proof.
 	\end{proof}
 	\subsection{Well-definedness of the Fourier projections \texorpdfstring{$\psi_\flat$ and $\psi_\sharp$}{psib and psis}}
 	\begin{prop}\label{prop:reddecay2}
 	Let $\psi \in C^\infty(\mathcal{M}_{\mathrm{RNAdS}}\setminus  \Ch )$ be as in \eqref{eq:defpsi} and 	let $r \in (r_-,r_+)$, $(\varphi,\theta) \in \mathbb S^2$ be fixed. Then, $t\mapsto \psi(t,r,\theta,\varphi)$ is a tempered distribution.
 		Moreover, higher derivatives  $t\mapsto \partial^k \psi(t,r,\theta,\varphi)$, where $\partial\in \{ \partial_t, \partial_r,\partial_\theta,\partial_\varphi \}$ are also tempered distributions.
 	\end{prop} 
 	\begin{proof}Fix $r \in (r_-,r_+)$, $(\varphi,\theta) \in \mathbb S^2$.	We will first prove that $t\mapsto \psi(t,r,\varphi,\theta)$ is slowly growing.\footnote{With \emph{slowly growing} we mean that $t\mapsto \psi(t,r,\varphi,\theta)$ and all its $\partial_t$ derivatives have at most polynomial growth as $ |t| \to \infty$.}  Since $\psi \in C^\infty(\mathcal{M}_{\mathrm{RNAdS}}\setminus \Ch )$ and in view of the facts that $\Box_g$ commutes with $T = \partial_t$ and our initial data are smooth and compactly supported, it suffices to obtain a polynomial bound for $\psi(t,r,\varphi,\theta)$. To do this we will propagate mild polynomial growth from the exterior region in the interior. (Note that this growth is far from being sharp but it will be sufficient for the purpose of proving well-definedness of $\psi_\flat$ and $\psi_\sharp$.)
 		
 	From \cref{prop:basicenergyest} and \cref{rmk:rmkdecay} we infer that $\psi$ and its derivatives remain bounded along the event horizon $\mathcal H$. A direct integration yields
 		\begin{align}
 		\int_{\mathcal{H}(v_1,v_2)} \tilde J_\mu^N[{{\psi}}] n^\mu_{\mathcal{H}^+} \dvol_{\Hp^+} \lesssim_{\psi_0,\psi_1} \langle v_2 \rangle,
 		\end{align}
 where $\langle v_2 \rangle$ denotes the Japanese bracket and $0\leq v_1\leq v_2$. The constant appearing in $\lesssim_{\psi_0,\psi_1}$ depends on some higher Sobolev norm of the initial data.
 		
 		Then, using the red-shift vector field (more precisely, applying \cref{prop:redshiftprelim}) yields
 		\begin{align}
 			\int_{\Sigma_{r_0}(v_1,v_2)} &\tilde J_\mu^N[{\psi}] n^\mu_{\Sigma_r} \dvol_{\Sigma_r}  \lesssim_{\psi_0,\psi_1} \langle v_2 \rangle
 		\end{align}
 		for any $r_0 \in [\red, r_+)$. If $r\in(r_-,r_+)$ as fixed above lies in the red-shift region $[\red, r_+)$, we directly conclude \eqref{eq:decayonconstrslice} after commuting with the angular momentum operators $\mathcal{W}_i$ and a Sobolev embedding on $\mathbb S^2$. If however $r \in (r_-,\red)$, we choose $\rblue = \rblue(r)$ small enough such that $r\in [\rblue,\red]$, i.e.\ $r$ lies in the no-shift region. Then, \cref{prop:noshiftprop} yields
 		\begin{align}\label{eq:noshiftestimateappendix}
 		\int_{\Sigma_{r}(v, 2v )} \tilde J_\mu^X[{{\psi}}] n^\mu_{\Sigma_{ r}}\dvol_{\Sigma_{ r}} \lesssim_{r} \int_{\Sigma_{\red}(v_{\red }(u_{r} (v )), 2v )} \tilde J_\mu^X[{{\psi}}] n^\mu_{\Sigma_{ r}}\dvol_{\Sigma_{ r}}\lesssim_{\psi_0,\psi_1,r} \langle v\rangle
 		\end{align}
 		for any $v\geq 1$. After commuting with angular momentum operators $\mathcal{W}_i$ and a Sobolev embedding on $\mathbb S^2$ we obtain
 		\begin{align}\label{eq:decayonconstrslice}
 		\int_{0}^{t} |\psi(t,r,\varphi,\theta) |^2 + |\partial_t \psi(t,r,\varphi,\theta) |^2 \d t \lesssim_{\psi_0,\psi_1, r} \langle t \rangle
 		\end{align}
 		from which we can deduce that $t \mapsto \psi(t,r,\varphi,\theta)$ is slowly growing (where we recall that $r,\varphi,\theta$ are fixed). Similarly, as $t\to -\infty$, we obtain the same conclusion. 
 		
 		Now, commuting with $\partial_t$, the angular momentum operators $\mathcal{W}_i$ and using elliptic estimates it follows that higher order derivatives are also slowly growing which concludes the proof.
 	\end{proof}
 	\begin{cor}\label{eq:welldefinedpsiflatpsisharp}
 		The Fourier projections $\psi_\flat$ and $\psi_\sharp$ in the interior $\mathcal{B}$ as in \eqref{eq:fouriertransforminterior} are well-defined and are smooth solutions of \eqref{eq:wave}.
 	\end{cor}
 	\begin{proof}
 		From \cref{prop:reddecay2} we know that $t \mapsto \psi(t,r,\varphi,\theta)$ is a tempered distribution in the interior for fixed $r,\varphi,\theta$. Thus, $\psi_\flat$ defined in \eqref{eq:fouriertransforminterior} is well defined as $\mathcal{F}_T^{-1}[\chi_{\omega_0}]$ is a Schwartz function. Moreover, $\psi_\flat$ is smooth because $\psi$ is smooth itself and by \cref{prop:reddecay2} we have that all higher derivatives $t \mapsto \partial^k \psi(t,r,\varphi,\theta)$ are tempered distributions, too. Now, this also implies that $\psi_\flat\in C^\infty(\mathcal B)$ solves \eqref{eq:wave} which concludes the proof in view of $\psi = \psi_\flat + \psi_\sharp$. 
 	\end{proof}
 	\begin{prop}\label{prop:welldefofpsiflat}
 		Let $\psi \in C^\infty(\mathcal{M}_{\mathrm{RNAdS}}\setminus \mathcal{CH})$ be defined as in \eqref{eq:defpsi}. Then, there exist $\psi_\flat\in C^\infty(\mathcal{M}_{\mathrm{RNAdS}}\setminus \mathcal{CH})$ and $\psi_\sharp \in C^\infty(\mathcal{M}_{\mathrm{RNAdS}}\setminus \mathcal{CH})$,  two solutions of \eqref{eq:wave}  with
 			\begin{align}\label{eq:psiflatexterior1}\psi_\flat  =\frac{1}{\sqrt{2\pi }} \mathcal{F}_T^{-1} \left[ \chi_{\omega_0}\right]  \ast \psi  \text{ and } \psi_\sharp = \psi - \psi_\flat,\end{align}
 	where $\chi_{\omega_0}$ is defined in \eqref{eq:chiomega0} and \begin{align}\label{eq:fourierprojection}\psi_\flat (t,r,\varphi,\theta) = \int_{\mathbb{R}} \frac{1}{\sqrt{2\pi}} \mathcal{F}_T^{-1}[\chi_{\omega_0}](s) \psi(t-s,r,\varphi,\theta) \d s\end{align}
 	in all coordinate patches $(t_{\mathcal{R}_A},r_{\mathcal{R}_A},\theta_{\mathcal{R}_A},\varphi_{\mathcal{R}_A} )$, $(t_{\mathcal{R}_B},r_{\mathcal{R}_B},\theta_{\mathcal{R}_B},\varphi_{\mathcal{R}_B})$ and $(t_{\mathcal{B}},r_{\mathcal{B}},\theta_{\mathcal{B}},\varphi_{\mathcal{B}} )$ in the regions $\mathcal{R}_A$, $\mathcal{R}_B$ and $\mathcal{B}$, respectively.
 	\end{prop}
 	\begin{proof}
First, from \cref{prop:basicenergyest} we know that $\psi$ and all higher derivatives decay logarithmically on the exterior regions $\mathcal{R}_A$ and $\mathcal{R}_B$.\footnote{This decay is only used in a qualitative way.} Hence, $\psi$ and all higher derivatives are smooth tempered distributions (for fixed $r,\varphi,\theta$) in the exterior regions $\mathcal{R}_A$ and $\mathcal{R}_B$ as functions of $t_{\mathcal{R}_A}$ and $t_{\mathcal{R}_B}$, respectively. Thus, the Fourier projections $\psi_\flat$ \eqref{eq:fourierprojection} is well-defined in $\mathcal{R}_A$ and $\mathcal{R}_B$ and it follows by Lebesgue's dominated convergence that $\psi_\flat$ is a smooth solution of \eqref{eq:wave}.
Moreover, from \cref{eq:welldefinedpsiflatpsisharp} we deduce that $\psi_\flat$ is also a well-defined smooth solution of \eqref{eq:wave} in the interior $\mathcal{B}$. 

Finally, $\psi_\flat$, defined a priori only in $\mathcal{R}_A$, $\mathcal{R}_B$ and $\mathcal{B}$, extends to a smooth solution of \eqref{eq:wave} on $\mathcal{M}_{\mathrm{RNAdS}}\setminus \mathcal{CH}$. This follows from using regular coordinates near the respective event horizons (outgoing Eddington--Finkelstein coordinates $(v,r,\theta,\varphi)$, where $v(t,r) = t + r_\ast, r(t,r) = r, \theta = \theta, \varphi = \varphi$ near $\mathcal{H}_A$ and ingoing Eddington--Finkelstein coordinates near $\mathcal{H}_B$) and writing $\psi_\flat$ again as a convolution in this coordinate system $\psi_\flat = \frac{1}{\sqrt{2\pi}} \mathcal{F}_T^{-1} [\chi_{\omega_0}] \ast \psi$. Note that $T= \partial_v$ in this coordinate system. This concludes the proof in view of $\psi=\psi_\flat + \psi_\sharp$. 
 	\end{proof}
 	\begin{prop}\label{prop:decayofpsil}
Assume that $\psi \in C^\infty(\mathcal{M}_{\mathrm{RNAdS}}\setminus \mathcal{CH})$ is a solution of \eqref{eq:wave} arising from smooth and compactly supported initial data as in \cref{prop:wellposedness}. Assume further that there exists an $L\in \mathbb{N}$ with $\langle \psi ,Y_{m\ell} \rangle_{L^2(\mathbb{S}^2)} =0 $ for $\ell \geq L$. Then, for every $r\in (r_-,r_+)$ and $(\theta,\varphi) \in \mathbb{S}^2$, the function $t \mapsto \psi(t,r,\varphi,\theta)$ is a Schwartz function. Moreover, higher derivatives $t\mapsto \partial^k \psi (t,r,\theta,\varphi)$, where $\partial\in \{ \partial_t, \partial_r,\partial_\theta,\partial_\varphi \}$ are also Schwartz functions.
\begin{proof}
The proof follows the same lines as the proof \cref{prop:reddecay2} with the difference that we have exponential decay on the event horizon 
\begin{align}\label{eq:expdec}
	\int_{v_1}^{v_2} \tilde J^N_\mu[\psi] n^\mu_{\mathcal{H}^+_A} \dvol_{\mathcal H^+_A} \lesssim D[\psi] \exp\left(- e^{-C(M,Q,l,\alpha) L}  v_1 \right),
\end{align}
where $D[\psi]$ is as in \eqref{eq:dpsi}.  Note that 
\eqref{eq:expdec} follows from \cite[Section~12]{decaykg}. Analogously to the proof of \cref{prop:reddecay2} we can propagate this decay to any $ \{r=const. \} $ hypersurface in the interior. This is very similar to \cite{MR3697197}. As before, by commuting with $\partial_t$ and $\mathcal{W}_i$ as well as using elliptic estimates, we see that on $ \{r=const. \} $, $\psi$ and higher derivatives $\partial^k \psi$ decay exponentially towards both components of $i^+$. This concludes the proof.
\end{proof}
\end{prop}
\printbibliography[heading=bibintoc]
\end{document}